\documentclass[11pt]{article}
\usepackage{fullpage}
\usepackage{natbib}
\bibpunct{(}{)}{;}{a}{,}{,}

\usepackage{times}

\usepackage{breakcites}

\usepackage{amsmath}
\usepackage{amsfonts}
\usepackage{mathrsfs}

\usepackage[utf8]{inputenc}

\usepackage{epstopdf}
\usepackage{mathtools}
\usepackage{graphicx}
\usepackage{tikz}
\usetikzlibrary{arrows.meta}

\usepackage{enumerate}
\usepackage{epsf,verbatim,amssymb,array,multicol,multirow}  
\usepackage{psfrag,bm,xspace}
\usepackage{hhline}
\usepackage{xstring}
\usepackage{ifthen}
\usepackage{booktabs}

\usepackage{physics}

\usepackage[linesnumbered,ruled,vlined,procnumbered]{algorithm2e}
\ifdefined \theorem 
\else
  \newtheorem{theorem}{Theorem}
\fi

\ifdefined \lemma 
\else
  \newtheorem{lemma}[theorem]{Lemma}
  
\fi

\ifdefined \corollary 
\else
\newtheorem{corollary}[theorem]{Corollary}
\fi

\newtheorem{fact}{Fact}

\ifdefined \proposition 
\else
\newtheorem{proposition}[theorem]{Proposition}
\fi

\ifdefined \definition 
\else
\newtheorem{definition}[theorem]{Definition}
\fi

\ifdefined \example 
\else
\newtheorem{example}{Example}
\fi

\ifdefined \remark 
\else
\newtheorem{remark}[theorem]{Remark}
\fi

\ifdefined \conjecture 
\else
\newtheorem{conjecture}[theorem]{Conjecture}
\fi

\newtheorem{innercustomgeneric}{\customgenericname}
\providecommand{\customgenericname}{}
\newcommand{\newcustomtheorem}[2]{%
  \newenvironment{#1}[1]
  {%
   \renewcommand\customgenericname{#2}%
   \renewcommand\theinnercustomgeneric{##1}%
   \innercustomgeneric
  }
  {\endinnercustomgeneric}
}

\newcustomtheorem{customthm}{Theorem}
\newcustomtheorem{customlemma}{Lemma}
\newcustomtheorem{customcorollary}{Corollary}

\makeatletter
\def\old@comma{,}
\catcode`\,=13
\def,{%
  \ifmmode%
    \old@comma\discretionary{}{}{}%
  \else%
    \old@comma%
  \fi%
}
\makeatother

\newcommand\numberthis{\addtocounter{equation}{1}\tag{\theequation}}

\usepackage{xcolor}
\definecolor{darkblue}{rgb}{0.1,0.1,0.8}
\definecolor{DarkGreen}{rgb}{0,0.6,0}
\definecolor{brickred}{rgb}{0.8, 0.25, 0.33}
\definecolor{britishracinggreen}{rgb}{0.0, 0.26, 0.15}
\definecolor{calpolypomonagreen}{rgb}{0.12, 0.3, 0.17}
\definecolor{ao(english)}{rgb}{0.0, 0.5, 0.0}
	\definecolor{cadmiumgreen}{rgb}{0.0, 0.42, 0.24}
\definecolor{burgundy}{rgb}{0.5, 0.0, 0.13}
\usepackage{etoolbox}

\providetoggle{Blue_revision}
\settoggle{Blue_revision}{true}

\providetoggle{Track}
\settoggle{Track}{true}
\newcommand{\addv}[3]{%
	\iftoggle{Track}{%
    	\IfEqCase{#1}{%
       	 	{a}{\ifthenelse{\equal{#2}{ON}}{{\color{cadmiumgreen}#3}}{#3}}%
        	{b}{\ifthenelse{\equal{#2}{ON}}{{\color{brickred}#3}}{#3}}%
       		{c}{\ifthenelse{\equal{#2}{ON}}{{\color{burgundy}#3}}{#3}}%
    	}[\PackageError{tree}{Undefined option to tree: #1}{}]%
	}{#3}%
}

 \usepackage{hyperref}
\usepackage{xcolor}

\hypersetup{ hidelinks }

\usepackage{graphicx}

\newcounter{relctr} %
\everydisplay\expandafter{\the\everydisplay\setcounter{relctr}{0}} %

\AtBeginDocument{} %

\global\long\def\RR{\mathbb{R}}

\global\long\def\EE{\mathbb{E}}
\global\long\def\PP{\mathbb{P}}

\global\long\def\11{\mathbbm{1}}

\newcommand{\bfx}{\mathbf{x}}

\newcommand{\bfX}{\mathbf{X}}
\newcommand{\bfY}{\mathbf{Y}}
\newcommand{\bfZ}{\mathbf{Z}}

\newcommand{\CA}{\mathcal{A}}

\newcommand{\ClC}{\mathcal{C}}

\newcommand{\CS}{\mathcal{S}}
\newcommand{\CT}{\mathcal{T}}

\newcommand{\one}[1]{\mathbf{1}_{\{#1\}}}

\global\long\def\+{\oplus}

\newcommand\pmm{\{-1,1\}}

\def\<{\langle}
\def\>{\rangle}

\DeclareMathOperator*{\sign}{sign}

\ifdefined \var 
  \renewcommand{\var}{\mathsf{var}}
\else
  \newcommand{\var}{\mathsf{var}}
\fi

\ifdefined \abs \else
 \newcommand{\abs}[1]{\lvert#1\rvert}
\fi

\ifdefined \norm \else
 \newcommand{\norm}[1]{\lVert#1\rVert}
\fi

\ifdefined \set 
  \renewcommand{\set}[1]{\left\{#1\right\}}
\else
  \newcommand{\set}[1]{\left\{#1\right\}}
\fi

\ifdefined \poly 

\else
  \newcommand{\poly}{\emph{poly}}
\fi

\DeclareMathOperator*{\argmin}{arg\,min}
\DeclareMathOperator*{\argmax}{arg\,max}

\usepackage{stackengine}

\usepackage[normalem]{ulem}

\newtheorem{observation}{Observation}
\newenvironment{proofnote}{\noindent {\textbf{Proof Note:}}}{\hfill$\blacksquare$}

\newcommand{\cuben}{\{0,1\}^n}
\newcommand{\reals}{\mathbb{R}}

\newcommand{\zb}{\bar{0}}
\newcommand{\lzb}[1]{\overline{0#1}}
\newcommand{\rzb}[1]{\overline{#10}}

\newcommand{\zok}{\{0,1\}^k}
\newcommand{\zonk}{\{0,1\}^{n-k}}

\newcommand{\basis}[1]{\phi_{#1}}
\newcommand{\coeff}[1]{{\hat{f}}_{#1}}
\newcommand{\approxcoeff}[1]{{\tilde{f}}_{#1}}

\newcommand{\coeffgen}[2]{{\hat{#1}}_{#2}}
\newcommand{\approxcoeffgen}[2]{{\tilde{#1}}_{#2}}

\def\fS{\hat{f}_{\mathcal{S}}}

\def\ps{\phi_\mathcal{S}}
\def\pS{\ps}

\def\pa{\operatorname{pa}}
\def\A{A}

\def\C{c}
\def\B{b}
\def\D{D}
\def\Z{Z}

\def\tmax{t^\ast}
\def\smax{s^\ast}

\def\opt{\emph{opt}}
\def\L1{{L_1}}

\def\ltp{l_P}
\def\ltphat{l_{\hat{P}}}
\def\dtv{d_{\text{TV}}}

\DeclareMathOperator{\err}{err}

\usepackage[nolist,nohyperlinks]{acronym}

\newacro{ML}[ML]{machine learning}
\newacro{IID}[i.i.d.]{independent and identically distributed} 
\newacro{PAC}[PAC]{\textit{probably approximately correct}}
\newacro{VC}[VC]{Vapnik–Chervonenkis}
\newacro{ERM}[ERM]{\textit{empirical risk minimization}}
\newacro{BN}[BN]{Bayesian network}
\newacro{DAG}[DAG]{directed acyclic graph}

\def\cite{\citep}

\newenvironment{proof}%
{%
\par\noindent{{\bf Proof: } }%
}%
{\hfill$\blacksquare$\par}

\begin{document}
\title{Learning DNF through Generalized Fourier Representations}
 \author{{}\footnote{Equal Contribution.} \ \ \ and {Roni Khardon}\footnotemark[1]  
 \\ {\small \texttt {\{mheidar|rkhardon\}@iu.edu}} \\
  Department of Computer Sciences, Indiana University, Bloomington, IN, USA
}

 \author{{Mohsen Heidari}\footnote{The authors contributed equally to this work.} \ \ \ and {Roni Khardon}\footnotemark[1]  
 \\ {\small \texttt {\{mheidar|rkhardon\}@iu.edu}} \\
  Department of Computer Sciences, Indiana University, Bloomington, IN, USA
}
\date{}
\maketitle
\begin{abstract}%
The Boolean Fourier representation has been widely used in learning theory, particularly for learning Disjunctive Normal Form (DNF) under uniform and product distributions. Extending these results to non-product distributions has remained a longstanding open problem.

We address this challenge by introducing a generalized Fourier representation that enables learning under a broad class of non-product distributions. Our approach represents any distribution $D$ as a Bayesian network (BN) and derives a corresponding Fourier expansion. We show that standard Fourier-based learning techniques using membership queries to identify heavy coefficients can be adapted to this generalized representation with minor modifications.

We prove that the $L_1$ spectral norm of conjunctions remains bounded under this expansion for difference-bounded tree BNs, significantly generalizing the known result for uniform distributions; matching lower bounds demonstrate the necessity of these constraints. Using these results, we establish the learnability of DNF and the agnostic learnability of decision trees under such distributions. Finally, we present an algorithm for learning difference-bounded tree BN distributions, extending our results to settings where the distribution is unknown.
\end{abstract}

\section{Introduction}

\newcommand{\TBDF}[1]{[TBD citations from Feldman: [#1]}
\newcommand{\TBD}[1]{[citations to add: [#1]}
\newenvironment{informalthm}[1]{\par\noindent{\textbf{Informal Version of {#1}}}\itshape}{}

The problem of learning Disjunctive Normal Form (DNF) expressions from examples has been a major open problem
since its introduction by \citet{Valiant1984}.
Significant progress has been made by
considering subclasses of expressions
(e.g., \cite{Valiant1985,BshoutyTamon96,SakaiMuroka00})
or making distribution assumptions \cite{Verbeurgt90,Servedio04},
and the best-known algorithm for the general problem is not polynomial time
\cite{KlivansS04}.
A potentially less demanding model allows for an additional source of information through membership queries (MQ).
In this model, in addition to random examples, the learner can ask for the label of the target function on any input.
\citet{Valiant1984} gave an efficient MQ learning algorithm for Monotone DNF. Since then, several subclasses of DNF have been shown to be learnable in this model
(e.g.\ \citep{Bshouty95,Kushilevitz96,ABKKPR1998,HellersteinKSS12}).

\citet{AngluinK95} have shown that (under cryptographic assumptions) the general distribution free case for general DNF
is not easier with MQ. On the other hand, positive results have been obtained for specific distributions.
\citet{Jackson1997}
gave the first polynomial time MQ learning algorithm for DNF over $c$-bounded product distributions.
This result was based on the Fourier representation of functions over the Boolean cube
\cite{linial1993constant} and combines the algorithm by \citeauthor{KM1993} (henceforth KM algorithm) for finding the heavy Fourier coefficient of a boolean function \cite{Goldreich1989,KM1993} with Boosting.
The approach was elaborated and improved by several authors
\cite{BshoutyJT04,Feldman07,Kalai2009,Feldman2012,Gopalan2008}.
However, to date, the Fourier approach has been largely limited to product distributions and
implications for learnability are restricted to such distributions.
In this paper, we provide a significant extension of these results to a broad class of distributions.
To achieve this the paper makes several distinct contributions.

First,
we develop a novel generalized Fourier representation induced by any distribution $D$,
{\em by using the Bayesian Network (BN) representation of $D$}.
A BN specifies a distribution using a \ac{DAG} and conditional probability distributions where each node is conditioned on its parents
\cite{Pearl1989,KollerFriedman2009}.
The generalized Fourier expansion constructs basis functions $\phi_S$ for $S\subseteq \{1,\ldots,n\}$ using the graph structure and conditional probabilities that specify the BN, yielding an orthonormal basis so that for any function we have $f(\bfx)=\sum_S \fS \phi_S(\bfx)$.
While the new basis preserves some important properties, unlike the standard construction, it does not impose sparsity.
That is, if $f$ depends only on a subset of variables $T$ and $S\setminus T\not = \emptyset$, the value of the coefficient $\fS$ may not be zero.

The second contribution is showing that
the KM algorithm can be extended for the new basis and with high probability it recovers all the heavy coefficients, $|\fS|\geq \theta$, of a function $f$.
The algorithm does not require inference with the BN (e.g. calculating marginal or conditional probabilities) that can be computationally hard but only requires forward sampling which is always feasible.
Our third contribution is in analyzing the Fourier representation of conjunctions $g$ with $d$ literals, specifically providing bounds for its spectral norm
$L_1(g)=\sum_S |\hat{g}_S|$.
This type of analysis is easy for the uniform case 
(where $L_1=1$ \cite{BlumFJKMR94,Khardon94}) 
and the product case 
(where $L_1=O(2^{d/2})$ \cite{Feldman2012}), 
but is nontrivial for general distributions due to the fact that sparsity does not hold.
We show that
the values of the coefficients are determined in a combinatorial manner by the values of the corresponding BN parameters.
We then derive bounds
for a broad class of difference bounded tree BN distributions, where the value of a parent can change the conditional probability of a child by at most $\alpha<0.5$ (noting that product distributions satisfy this with $\alpha=0$).
In particular, 
in chain BNs we have 
$L_1(g)=O((\frac{2}{1-2\alpha})^{d})$ and for 
tree BN we have $L_1(g)=O((\frac{2}{1-2\alpha})^{2d})$.
The upper bounds are complemented by showing that
without boundedness or without a tree structure the spectral norm can be exponentially large even for $d=1$. 
As a byproduct,
our analysis also provides an \emph{exact value} for the spectral norm in the product case, improving on previous bounds. 
Finally, in addition to tree distributions, we analyze the spectral norm of conjunctions under 
$k$-junta distributions \cite{AliakbarpourBR16} showing that $L_1(g)=O(2^{(k+d)/2})$.

Our fourth contribution includes the main results of this paper, showing learnability under the corresponding families of distributions. We show that the extended KM algorithm can be directly used to learn disjoint DNF (and decision trees), and that the algorithm PTF-construct of \citet{Feldman2012} can be used with a slight modification (removing the filtering of large degree coefficients, which is only suitable with sparsity) to learn DNF, in the PAC model with membership queries. In addition, we show that the KM-gradient descent method of \cite{Gopalan2008} can learn disjoint DNF (and decision trees) in the agnostic learning model with membership queries.

Finally, note that the discussion so far assumed that a BN representation of the distribution is given to the learner.
While learning general distributions is computationally hard, it is well known that tree BN are learnable
\cite{Chow1968,Hoffgen93,BhattacharyyaGPTV23}. 
Our fifth contribution is a learnability analysis for difference-bounded tree distributions. 
We develop two variants of the base algorithm, one for the realizable case where the target is known to be a bounded tree distribution and a slightly more complex algorithm for the unrealizable case, showing polynomial learnability in both cases. 
Combined with the results above this shows that the learnability results hold even when the distribution is not known in advance.

\paragraph{Summary of the contributions:}
To summarize, the paper develops a generalized Fourier basis and shows that major algorithmic tools from learning theory
can be used with this basis.
Using these and an analysis of the spectral norm for conjunctions, the paper shows the learnability of DNF and agnostic learnability of decision trees under difference-bounded tree distributions,
significantly extending previous results to a broad class of distributions.
We emphasize that the basis and the extended KM algorithm are valid for any distribution and they do not require a tree structure or boundedness.
These conditions are required only to establish spectral norm bounds used for learnability results.
More specifically, our main contributions include:
\begin{itemize}
 \item A new Fourier basis for any distribution $D$ by using the BN representation of $D$ (Corollary~\ref{cor:orthonormalbasis}).
 \item An extension of the KM algorithm to any distribution using the new basis (Theorem \ref{thm:EKM}), a proof of learnability of disjoint DNFs using KM algorithm (Corollary \ref{cor:KM DNF}), learnability of DNFs using the PTFconstruct algorithm of \citet{Feldman2012} (Corollary \ref{cor:FC15}) and agnostic learnability of decision tree using KM-gradient descent of \citet{Gopalan2008}  (Corollary \ref{cor:DT}).
 \item An exact Fourier expansion of conjunctions under chain BNs (Lemma \ref{lem:Fourier exact}) and an upper bound on the spectral norm of $d$ literal conjunctions for difference-bounded chains (Theorem \ref{thm:L1 bound chain}). %
 \item An exact characterization of the spectral norm of $d$ literal conjunctions under product distributions, and a tighter upper bound on the spectral norm (Proposition \ref{prop:product}). 
 \item An upper bound for the spectral norm of $d$ literal conjunctions under difference bounded tree BNs (Theorem \ref{thm:tree2}) and $k$-junta distributions (Lemma \ref{lm:kjunta}). %
 \item An exponential lower bounds on the spectral norm of $1$-literal conjunctions, when the difference-bounded condition is violated, or without the tree BN structure (Corollary~\ref{cor:LBchain} and Theorem \ref{thm:expLBantitree}). 
 \item
 Proof of learnability of difference-bounded tree distributions in the realizable case (Corollary~\ref{lm:directed-c1}) and unrealizable case (Corollary~\ref{lm:directed-c1 unrealizable}). 
\end{itemize}

\paragraph{Organization of the paper:} We start with basic definitions and the preliminaries in Section \ref{sec:prelim}. 
The new BN induced Fourier basis is introduced in Section \ref{sec:basis} and the extension of the KM algorithm to arbitrary distributions is presented in Section \ref{sec:KM}. 
Sections~\ref{sec:chain}, \ref{sec:tree} and \ref{sec:lb} 
develop upper bounds and lower bounds for the spectral norm under chain and tree BNs,
and Section~\ref{sec:Kjunta} develops bounds for $k$-junta distributions.
The implications of our results for the learnability of DNFs and agnostic learnability of decision trees are discussed in Section \ref{sec:learnDNF}. 
Learning difference bounded tree distributions is developed in Section~\ref{sec:learnTrees}.
Lastly, 
Section \ref{sec:conclusion} concludes with a summary and some questions for future work.

\section{Preliminaries}\label{sec:prelim}
\paragraph{Notations.} We use $[n]$ to denote the sequence $[1,2,\ldots,n]$, $[a,b)$ to denote the sequence $a,a+1,\ldots,b-1$, and similarly we define $(a,b)$ and $[a,b)$ for sequences $a+1,\ldots,b-1$, and $a,a+1,\ldots,b-1$, respectively. Capital letters are used for random variables and lowercase letters to denote their assignments. Given a vector $\bfx \in \{0,1\}^n$ and $\CS\subseteq [n]$, define $\bfx_\CS :=(x_i)_{i\in \CS}$ as the restriction of $\bfx$ to the coordinates $\CS$.

Learning in this paper is defined based on the well-known \ac{PAC} model \cite{Valiant1984,Kearns1994}. In this framework, the learner finds an approximation $h$ to an unknown target Boolean function $f$ given oracle access to $f$ in the form of labeled examples $(\bfx, f(\bfx))$ with $\bfx$ generated based on an unknown distribution $D$. 
With membership queries (MQ),
in addition to random examples, the algorithm can ask for the label $f(\bfx)$ of any example $\bfx$ of its choice.
The objective is to output a hypothesis $h$ that is close to the target function $f$ in terms of its predictions.
More formally:
\begin{definition}[Learnability]
  $\mathcal{F}$ is learnable with MQ if there is an algorithm such that, for every $\epsilon, \delta > 0$, and every target function
  $f\in\mathcal{F}$ where  $f: \{0, 1\}^n \to \{0, 1\}$, given oracle and MQ access to $f$, 
 with probability $1 - \delta$ the algorithm
 produces a hypothesis $h : \{0,1\}^n \to \{0, 1\}$ such that
  \[
 P_{\bfX \sim D}(h(X) \neq f(X)) \leq \epsilon.
  \] 
\end{definition}
The key to our analysis is incorporating the representation of the distribution $D$ as a BN into the representation of functions. 
Since BN is a universal representation this does not restrict the set of distributions under consideration.
\paragraph{Bayesian Network.}
A BN is specified by a directed acyclic graph and a set of conditional probability tables (or functions).
Let $G=(V,E)$ be the graph where nodes correspond to the individual random variables. 
Then the joint probability distribution can be written as a product of individual probability distributions of each node conditioned on its parent variables: 
\begin{align*}
P(X_1,\cdots, X_n) = \prod_{v \in V} P(X_v| X_{\pa(v)}),
\end{align*} 
where $\pa(v)$ is the set of the parents of node $v$ in $G$, and where in this paper the variables are binary, i.e., $X_i\in\{0,1\}$.

For some of the results, we will need to restrict the class of distributions $D$. For any node $v$, let $\mu_{v, x_{\pa(v)}}$ and $\sigma_{v, x_{\pa(v)}}$ denote the conditional expectation and standard deviation of $X_v$ given its parents' realization $x_{\pa(v)}$, respectively. Note that the values $\{\mu_{v, x_{\pa(v)}}\}$ are exactly the parameters that specify the BN representation of $D$.

\begin{definition}\label{def:D bounded}
A distribution $D$ is $c-${\em bounded} for $c\in (0,1)$ if for all $v\in V$ and any assignments $x, x'$ we have $c\leq P(X_v=x | x_{\pa(v)}=x')\leq 1-c$.
Moreover, $D$ is $\alpha$-{\em difference-bounded} if %
for all $v$ and for any two assignments $x_{\pa(v)}, y_{\pa(v)}$ to $\pa(v)$, 
we have
$\abs{\mu_{v, x_{\pa(v)}} - \mu_{v, y_{\pa(v)}}}\leq \alpha$ and $\abs{\sigma_{v, x_{\pa(v)}} - \sigma_{v, y_{\pa(v)}}}\leq \alpha$. 
\end{definition}
Based on this definition, a distribution is a {\em difference-bounded tree BN}, if it can be expressed as a BN where each node has at most a single parent and the conditional probability tables in the specification are $c$-bounded and $\alpha$-difference bounded. The difference boundedness implies a limitation on the influence of the parent nodes on the first two moments of a child node. 

Our results generalize previous works that established the learnability of functions using the Fourier representation under uniform and bounded product distributions 
\cite{linial1993constant,furst1991improved,KM1993,Jackson1997,Kalai2008,Feldman2009}. 
These distributions are captured by BNs with an empty edge set. The boundedness condition holds and the differences are zero and hold trivially. 

\paragraph{Fourier Expansion on the Boolean Cube with Product Distributions.} 
We briefly discuss prior constructions of the Fourier basis; see \cite{ODonnell2014,Wolf2008} for a review.
In the following, 
given any distribution $D$, define the induced inner product between any pair of functions $f,g: \{0,1\}^n\rightarrow \RR$
by
$\<f, g\>=\EE_D[f(\bfX) g(\bfX)]$ . 
The expansion relies on a set of basis functions defined as $\psi_{i}(\bfx)=(-1)^{x_i}=1-2x_i$ 
where $x_i\in\{0,1\}$ and $i\in [n]$. The orthonormal Fourier basis for the uniform Boolean Fourier is defined as $\psi_{S}(\bfx)=\prod_{i\in S} \psi_{x_i}(\bfx)$, for all subsets $S\subseteq [n]$ and $\bfx\in \{0,1\}^n$.
Any function $f$ on the Boolean cube admits the following decomposition 
\begin{equation*}
 f(\bfx) = \sum_{\mathcal{S}\subseteq [n]} \fS~\psi_{\mathcal{S}}(\bfx), \quad \forall \bfx\in \{0,1\}^n,
\end{equation*}
where $\fS$ are called the Fourier coefficients of $f$ and are calculated as
$$ \fS = \<f, \psi_{\mathcal{S}}\>=\EE_{\bfx\sim \mbox{Uniform}}[f(\bfX) \psi_\CS(\bfX)]= \frac{1}{2^n}\sum_{\bfx} f(\bfx) \psi_{\mathcal{S}}(\bfx).$$ 
This expansion relies on the restriction that the input variables are uniformly distributed over the Boolean cube. 
The orthonormal basis for product distributions, generalizing this, is given by $\psi_{i}(\bfx)=\frac{\mu_i-x_i}{\sigma_i}$, where $\mu_i$ and $\sigma_i$ are the expectation and the standard deviation of $X_i$, and 
as above we have $\psi_{S}(\bfx)=\prod_{i\in S} \psi_{x_i}(\bfx)$ and 
 $\fS = \EE_D[f(\bfX)\psi_\CS(\bfX)]$. This formulation reduces to the one for the uniform distribution, where $\mu_i=0.5$ and $\sigma_i=\sqrt{\mu_i(1-\mu_i)}=0.5$. 
In our development below for the general case, we use the negation of these basis functions in order to simplify the presentation. However, our basis can be seen as a direct extension of the product case.

The above expansion for non-uniform product distributions is not technically a Fourier transform as $\psi_\CS$ are not group characters and therefore some properties are not satisfied (see related discussion by \citet{ODonnell2014}).
However, certain nice properties of the Fourier transform are satisfied such as the Plancherel formula. Therefore, we use the term {\em Fourier expansion} to distinguish it from the Fourier transform. 

There are other forms of orthogonal decomposition including the Hoeffding-Sobel decomposition \citep{HoeffdingDecomp,Sobol1993,Chastaing2012} and its generalization \citep{Chastaing2012}. However, such decompositions are basis-free, making their learning implications unclear. Moreover, extensions of the Fourier expansion to non-product distributions have been studied \cite{HeidariICML2021,Heidari2022}. When the distribution is not known, such works rely on an empirical orthogonalization process that is applied on the above Fourier expansion. Although these Fourier expansions are suitable for machine learning problems such as feature selection, they may not be applicable to prove upper bounds on learning of Boolean functions, especially conjunctions. The reason is that the basis functions are unknown prior to the given samples. In our work, we generalize the Fourier expansion to non-product distributions described by \acp{BN}.

\section{BN Induced Fourier Basis}\label{sec:basis}
In what follows, we define a distribution dependent Fourier basis for functions on the Boolean cube. 
The key is to define the basis using not just the distribution but using a specific representation of the distribution as a BN,
i.e., the \ac{DAG} $G$ and associated conditional probabilities. 
For a BN $G$, with $|V|=n$ we can identify the nodes with their indices, i.e., $V=\{1,\ldots,n\}$.

For any node $v\in V$ in a BN $G$, define 
 \begin{align*}
\phi_{v}(\bfx) := \frac{x_v- \mu_{v,x_{\pa(v)}}}{\sigma_{v, x_{\pa(v)}}},
\end{align*}
for all $\bfx\in \{0,1\}^n$. 
Then, the Boolean Fourier basis on the BN $G$ is:  

\begin{definition}[BN Induced Fourier Basis] %
The Boolean Fourier basis for a BN given by $G$ and the associated parameters is defined as 
\begin{align}
\phi_{\CS}(\bfx) := \prod_{v\in \CS} \phi_{v}(\bfx) = \prod_{v\in \CS} \frac{x_v- \mu_{v,x_{\pa(v)}}}{\sigma_{v, x_{\pa(v)}}}
\end{align}
 for all $\CS \subseteq V$ and $\bfx\in \set{0,1}^n$.  
 Note that for $S=\phi$, the product is empty and $\phi_{\CS}(\bfx)=1$.
\end{definition}
To simplify the presentation, we assume $G$ is known and omit any notation showing an explicit dependence on $G$. 
The next lemma shows that this is indeed an orthonormal basis and develops some more useful properties of the basis.

\begin{lemma}\label{lem:pS properties}
The following holds for any $\CS, \CT\subseteq [n]$ and $\bfx\in \{0,1\}^n$.
\begin{enumerate}[(a)]
\item 
For $S\not = \phi$,
$\EE[\ps(\bfX)]=0$, and $\EE[\ps(\bfX)|\bfx_{\pa(\CS)}]=0$, where $\pa(\CS)$ is the set of parents of all $v\in \CS$.
\item $\EE[\ps^2(\bfX)]=1$, and $\EE[\ps^2(\bfX)|x_{\pa(\CS)}]=1$.
\item $\pS \phi_{\CT} = \phi_{\CS\cap \CT}^2 \phi_{\CS\Delta \CT}$.
\item $\EE[\ps(\bfX) \phi_{\CT}(\bfX)] = 0$ and $\EE[\ps(\bfX) \phi_{\CT}(\bfX) |\bfx_{\pa(\CS\cup \CT)} ] = 0$ for $\CT\neq \CS$.
\end{enumerate}
\end{lemma}
\begin{proof} 
We start by establishing that individual basis functions are normalized.
We have
\begin{align*}
\EE[\phi_{v}(\bfX)] &=
\EE_{X_{\pa(v)}} \qty[\EE_{v} \qty[\phi_{v}(\bfX)|X_{\pa(v)}]] = 
\EE_{X_{\pa(v)}} \qty[\EE_{v} \qty[\frac{X_v- \mu_{v,X_{\pa(v)}}}{\sigma_{v, X_{\pa(v)}}}|X_{\pa(v)}]]\\
&= \EE_{\pa(v)}[0]=0.
\end{align*} 
Similarly, 
\begin{align*}
\EE[\phi_{v}(\bfX)^2] &=
\EE_{X_{\pa(v)}} \qty[\EE_{v} \qty[\phi_{v}(\bfX)^2|X_{\pa(v)}]] = 
\EE_{X_{\pa(v)}} \qty[\EE_{v} \qty[(\frac{X_v- \mu_{v,X_{\pa(v)}}}{\sigma_{v, X_{\pa(v)}}})^2|X_{\pa(v)}]]\\
& = \EE_{\pa(v)}[1]=1.
\end{align*} 
Let $\CS=\{{j_1}, \ldots, {j_k}\}$ where $j_1,j_2,\ldots,j_k$ satisfy the BN ordering. Then
\begin{align}
\nonumber
&
\EE[\ps(\bfX)|x_{\pa(\CS)}]
\\ = &
\label{eq:ExpBasisSingle}
\EE_{X_{j_1}} \qty[\phi_{{j_1}}(\bfX) \ 
\EE_{X_{j_2}} [\phi_{{j_2}}(\bfX) \ 
\ldots 
\EE_{X_{j_k}} [\phi_{{j_k}}(\bfX)
|x_{\pa(X_{j_k})}] 
\ldots
|x_{\pa(X_{j_2})}] 
|x_{\pa(X_{j_1})}] 
\\ = &
\nonumber
\EE_{X_{j_1}} \qty[\phi_{{j_1}}(\bfX) \ 
\EE_{X_{j_2}} [\phi_{{j_2}}(\bfX) \ 
\ldots 
0
\ldots
|x_{\pa({j_2})}] |x_{\pa(X_{j_1})} ]
\\ = & 0.
\nonumber
\end{align}
The same sequence of equations for $\phi^2$ yields 1, and taking the expectation over $x_{\pa(\CS)}$ returns the same constant. 
This establishes (a) and (b).
Note from these equations that if we perform expectations in reverse order (from $k$ to $1$) then a $\phi$ term yields a zero for the entire expectation and a $\phi^2$ term contributes a multiplier of 1, i.e., the value is taken from the next expectation.

\bigskip
We next observe that (c) holds by definition:
\begin{align*}
\pS \phi_{\CT} = \prod_{v\in \CS} \phi_{v}(\bfx) \prod_{v\in \CT} \phi_{v}(\bfx) =
\phi_{\CS\cap \CT}(\bfx)^2 \phi_{\CS\Delta \CT}(\bfx).
\end{align*}
For (d) we have
\begin{align}
\label{eq:ExpBasisProduct}
\EE[\ps(\bfX) \phi_{\CT}(\bfX) |x_{\pa(\CS\cup \CT)} ] 
=
\EE\qty[  \prod_{v\in \CS\cap \CT} \phi_{v}(\bfX)^2 \prod_{v\in \CS\Delta \CT} \phi_{v}(\bfX) |x_{\pa(\CS\cup \CT)} ]. 
\end{align}
Let $\CS\cup \CT=\{{j_1}, \ldots,{j_k}\}$ where $j_1, \ldots,j_k$ satisfy the BN ordering. 
Then we can order the expectations in \eqref{eq:ExpBasisProduct} in the same manner as the expectations in \eqref{eq:ExpBasisSingle}. 
As above, if is $\CS\Delta \CT$ not empty then the expectation on its last variable in the ordering yields a zero and the entire expectation is 0. 
On the other hand variables in $\CS\cap \CT$ yield a 1, so if $\CS\Delta \CT$ is empty the expectation is 1. 
Finally, taking expectation over $X_{\pa(\CS\cup \CT)}$ maintains the constant.
\end{proof}

Note that unlike the uniform case where $\pS \phi_{\CT} = \phi_{\CS\Delta \CT}$ we have the weaker property (c).
Properties (b,d) show that $\pS$ functions are orthonormal 
w.r.t.\ inner product 
$\<f, g\>=\EE_D[f(\bfX) g(\bfX)]$.
Moreover, they form a basis for the space of all functions on the Boolean cube: 

\begin{corollary}
\label{cor:orthonormalbasis}
Given a probability distribution $D$ on $\{0,1\}^n$ described by a \ac{BN}, the set of $\ps, \CS\subseteq [n]$ forms an orthonormal basis. 
For all functions $f:\cuben \rightarrow \RR$,
\begin{align*}
f(\bfx) = \sum_{\CS} \fS \ps(\bfx),
\end{align*}
where  $\fS=\EE_D[f(\bfX)\pS(\bfX)]$ is the Fourier coefficient of $f$.
\end{corollary}

The standard properties of this expansion are summarized below. These statements are derived from the orthonormality of the basis. Hence, we omit the proofs. 
\begin{fact}\label{fact:fourier1}
For any bounded pair of functions $f,g: \pmm^d\mapsto \RR$, the following statements hold: 
\begin{itemize}
\item Plancherel Identity: $\EE[f(\bfX)g(\bfX)]=\sum_{\mathcal{S}\subseteq [d]} \fS \hat{g}_{\mathcal{S}}$. 
\item Parseval’s identity $\EE_D[f^2] = \sum_{\mathcal{S}\subseteq [d]} \fS^2$.
\end{itemize}
\end{fact}

Following previous work, the following definition is instrumental for analysis of the learnability:
\begin{definition}\label{def:spectral norm}
    The spectral norm of a function $f$ under a distribution $D$ is the sum of the absolute values of its Fourier coefficients under $D$, $L_1(f):=\sum_\CS |\fS|$. 
\end{definition}

\section{Extending KM to Arbitrary Distributions}\label{sec:KM}
In this section, we show that
the KM algorithm \cite{KM1993} can be generalized to recover the large coefficients of a function $f$, 
for any 
distribution $D$, given a BN representation for it and the corresponding basis. 
The algorithm is based on a recursive procedure that identifies whether a subset of coefficients includes any large coefficients and in this way identifies all the large coefficients. 
The main new development in our work is given in the following lemma, which shows how the construction of the main tool and its proof of correctness can be adapted to the new basis.

To introduce the generalization, we extend the notation to specify sets using binary strings. We represent a set $\CS\subseteq [n]$ with a binary vector $\gamma \in \{0,1\}^n$ where $\gamma_i=1$ if $i\in \CS$, and $\gamma_i=0$ otherwise. With this notation strings can refer to sets in a 1-1 manner in a natural way. 
Moreover, the Fourier expansion is written as 
\begin{align*}
    f(\bfx) = \sum_{\gamma \in \{0,1\}^n} \hat{f}_{\gamma} \phi_\gamma(\bfx).
\end{align*}
Let $\alpha\in\zok$ and $\beta\in \zonk$. We use $\beta\alpha$ to denote the concatenation of the two strings which is of length $n$. Moreover, to avoid confusion, our strings, serving as subscripts or arguments to our function $f$, basis functions or coefficients, will always be of length $n$. To achieve this we pad a string $\alpha\in\zok$ on the left with $\bar{0}=0^{n-k}$ to get $\lzb{\alpha}=\zb\alpha$ where we use the overline to emphasize the concatenation. Similarly, we pad $\beta\in\zonk$ on the right with $\bar{0}=0^{k}$ to get $\rzb{\beta}=\beta\zb$. When using this notation the length of $\zb$ should be clear from the context. 

Let $\alpha\in\zok$  and define the function
\begin{align}
g_\alpha(u)=\sum_{\beta\in\zonk}\ \coeff{\beta\alpha}\ \basis{\rzb{\beta}}(\rzb{u}),
\end{align}
for every $u\in \zonk$. This function has the portion of the Fourier spectrum of $f$ that contains $\alpha$. 
The next lemma which is a generalization of Lemma~3.2 of \citet{KM1993} shows that $g_\alpha(u)$ can be computed as an expectation:
\begin{lemma}
\label{lm:KMgalpha}
For any $\alpha\in\zok$ and $u\in \zonk$, let $\bfY=(X_{n-k+1},\ldots, X_{n})$ be the last $k$ variables of $\bfX^n$, then
\begin{align}
g_\alpha(u) = 
\EE_{(\bfY|(X_1,\ldots,X_{n-k}) = u)} 
[f(u\bfY) \basis{\lzb{\alpha}}({u\bfY})].
\end{align}
\end{lemma}
\begin{proof}
Starting from the RHS, by the Fourier expansion of  $f(u\bfY)$ we have:
\begin{align}
\nonumber
 \EE_{\bfY|u} 
    [f(u\bfY) \basis{\lzb{\alpha}}({u \bfY})]
 = %
 \EE_{\bfY|u} 
    [ \sum_{a_1}\sum_{a_2} \ \coeff{a_1 a_2} \basis{a_1 a_2}(u\bfY) 
    \basis{\lzb{\alpha}}({u\bfY})],
\end{align}
where $a_1\in\zonk$ and $a_2\in\zok$. Next note, similar to Lemma~\ref{lem:pS properties}, that
\begin{align}
\nonumber
\basis{a_1 a_2}(u\bfY) \basis{\lzb{\alpha}}({u\bfY})
& =
(\prod_{j\in A_1} \basis{j}({u\bfY}))
(\prod_{j\in A_2} \basis{j}({u\bfY})^2)
(\prod_{j\in A_3} \basis{j}({u\bfY}))
\\
\nonumber
& =
\basis{\rzb{a_1}}({u\bfY})
(\prod_{j\in A_2} \basis{j}({u\bfY})^2)
(\prod_{j\in A_3} \basis{j}({u\bfY}))
\end{align}
where 
$A_1=\rzb{a_1}$
(i.e., $A_1$ is the set corresponding to $\rzb{a_1}$),  
$A_2=\lzb{a_2}\cap \lzb{\alpha}$
and
$A_3=\lzb{a_2}\Delta \lzb{\alpha}$,
implying that
\begin{align}
\nonumber
RHS
& =
\sum_{a_1}\sum_{a_2} \ \coeff{a_1 a_2} 
\EE_{\bfY|u} 
    [
    \basis{\rzb{a_1}}({u\bfY})
    (\prod_{j\in A_2} \basis{j}(u{\bfY})^2)
    (\prod_{j\in A_3} \basis{j}(u{\bfY}))
    ]
    \\
\nonumber
& =
\sum_{a_1}\sum_{a_2} \ \coeff{a_1 a_2} 
    \basis{\rzb{a_1}}(\rzb{u})
\EE_{\bfY|u} 
    [
    (\prod_{j\in A_2} \basis{j}(u{\bfY})^2)
    (\prod_{j\in A_3} \basis{j}(u{\bfY}))
    ].
\end{align}
The last equality holds because indices in $A_1$ (as well as their ancestors) are restricted to the first $n-k$ variables. Therefore $\basis{\rzb{a_1}}({u\bfY})=    \basis{\rzb{a_1}}(\rzb{u})$ and we can pull this term out of the expectation.
We can now proceed in evaluating the expectation over variables in $A_2,A_3$ in reverse lexicographical order over the BN (i.e., from children to parents) as a sequence of conditional expectations.
In that sequential computation, 
for indices in $A_2$, $\EE_{X_j|X_{\pa(j)}}[\basis{j}({u\bfY})^2| x_{\pa(j)}] =1$ and 
for indices in $A_3$, $\EE_{X_j|X_{\pa(j)}}[\basis{j}({u\bfY})| x_{\pa(j)}]=0$.
That is, the expectation is zero when $A_3$ is not empty, i.e., when $a_2\not = \alpha$ and it is 1 when $a_2 = \alpha$.
Therefore, as claimed,
\begin{align}
\nonumber
RHS 
& =
\sum_{a_1}\ \coeff{a_1 \alpha} 
    \basis{\rzb{a_1}}(\rzb{u})=g_\alpha(u).
\end{align}
\end{proof}

The above lemma is crucial in extending the results of \citet{KM1993} to general BN distributions.  In particular, the following results  follow along the same lines as in the original paper \cite{KM1993} or its generalization to the product case \cite{Bellare1991,Jackson1997}. We include a sketch of the ideas for completeness.

\begin{lemma}[cf. Lemma~3.4 of \cite{KM1993}]
\label{lem:KM g alpha} \ 
\\
(1)
Given any $\theta>0$, the number of sets $\CS$ such that $|\coeff{\CS}|\geq \theta$ is bounded by $1/\theta^2$.
\\
(2)
For all $\alpha \in \zok$,
$$\EE_{U\sim  D(X_{1},\ldots,X_{n-k})} [g_\alpha^2(U)]=\sum_{\beta\in\zonk} \coeff{\beta\alpha}^2.$$
\\
(3)
If $\exists \beta\in \zonk$, such that $|\coeff{\beta\alpha}|\geq \theta$ then $\EE_U[g_\alpha^2(U)]\geq \theta^2$.
\\
(4)
For a fixed $k$, 
the number of $\alpha$ values such that $\EE_{U}[g_\alpha^2(U)]\geq \theta^2$ is bounded by $1/\theta^2$.
\end{lemma}
The lemma follows from Fact \ref{fact:fourier1} implying that for a Boolean function $f$, $1\geq \EE_D[{f}^2]=\sum_S \coeff{\CS}^2$ and the following observation $$\sum_\alpha \EE_{U}[g_\alpha^2(U)] = \sum_\CS \coeff{\CS}^2.$$

\begin{algorithm}[t]
    \caption{Extended KM Algorithm}
    \label{alg:KM}
    \DontPrintSemicolon
    \SetKwProg{Fn}{}{:}{} 
    \SetKwFunction{KM}{KM}
    \SetKwFunction{Coef}{Coef}
    \SetCommentSty{textrm}
    \SetKwComment{Comment}{$\triangleright$ }{}
    \Comment{Assumes variable ordering $\{1,\ldots,n\}$ satisfies ordering in BN representation of $D$}%
    \Fn{\KM{$D, f, \theta$}}{
        $\CA \leftarrow $\Coef{$\emptyset$,$0$}\;
        \KwRet $\CA$
    }
    \Fn{\Coef{$\alpha$,$k$}}{ 
    \If{$G_\alpha \geq \theta^2/2$}
    { \hspace{3.5cm} \Comment{$G_\alpha$ from (\ref{eq:Galpha}) approximates $\EE[g^2_\alpha]$ within $\theta^2/4$.} %
    \If{$k=n$}{     
        \KwRet $\alpha$\;
    }
        \KwRet \Coef($0\alpha$,$k+1$) $\cup$ \Coef(1$\alpha$,$k+1$)\; \mbox{        }   
        \hspace{3.3cm} \Comment{Note left concatenation on $\alpha$. } %
    }
    \Else{
        \KwRet $\emptyset$
    }
    }
\end{algorithm}

\subsection{Extended KM Algorithm}
The KM algorithm uses the facts above to find the heavy Fourier coefficients of $f$ under the uniform distribution \cite{KM1993}.  
The Extended KM Algorithm
(see Algorithm \ref{alg:KM}) 
is a slight elaboration that controls the ordering of variables in the construction of $\alpha$
and modifies the form of sampling in estimating $\EE_{U} [g_\alpha^2(U)]$. 

For a distribution $D$ described by a BN, Lemma \ref{lem:KM g alpha} implies that if we can evaluate 
\begin{align*}
\EE_{U} [g_\alpha^2(U)] =
\EE_{U} 
\EE_{\bfY_1|U,\bfY_2|U} [f(U\bfY_1) f(U\bfY_2) \basis{\lzb{\alpha}}({U\bfY_1}) \basis{\lzb{\alpha}}({U\bfY_2})],
 \end{align*}
we can recursively find all the large coefficients of $f$ and that the number of such coefficients is not too large. 
In addition, once the coefficient set is chosen, the algorithm has to estimate the coefficients which are themselves expectation $\fS=\EE_D[f(\bfX)\phi_\CS(\bfX)]$.
In both of these expressions, the basis functions are bounded by $(1/\sigma)^{|\CS|}$ 
where $\sigma=\sqrt{\mu(1-\mu)}$ is the minimal standard deviation. 
However, since this is exponential in $|\CS|$, we need a more refined argument than in the case of the uniform distribution.
In previous work on product distributions, \citet{Jackson1997} relied on the fact that $|\CS|$ is logarithmic to obtain a polynomial bound. 
\citet{Kalai2009} developed an alternative form for the estimates that still allows for the use of Hoeffding's bound even for large $|\CS|$, but their construction relies on independence and is not easily applicable for the general case. 

We therefore use a different route,
by bounding the variance and using the median-of-means (MoM) estimator \citep{Nemirovskij1983,Alon1999,Jerrum1986}. 
Toward that we introduce three random variables $Z_1, Z_2, Z_3$ so that 
\begin{align*}
\EE[Z_1] &= \fS =\EE_D[f(\bfX)\phi_S(\bfX)],\\
\EE[Z_2| u] & = g_\alpha(u) = \EE_{\bfY|u} [f(u\bfY) \basis{\lzb{\alpha}}({u\bfY})],\\
\EE[Z_3| u]  &= g^2_\alpha(u)  = \EE_{\bfY_1|u,\bfY_2|u} [f(u\bfY_1) f(u\bfY_2) \basis{\lzb{\alpha}}({u\bfY_1}) \basis{\lzb{\alpha}}({u\bfY_2})], \\
\EE[Z_3]  &= \EE_{U}[g^2_\alpha(U)]  = \EE_{U,\bfY_1|u,\bfY_2|u} [f(U\bfY_1) f(U\bfY_2) \basis{\lzb{\alpha}}({U\bfY_1}) \basis{\lzb{\alpha}}({U\bfY_2})],
\end{align*}
where we have identified the arguments of the expectations with $Z_1, Z_2|u, Z_3|u, Z_3$ respectively.
We have the following observation.
\begin{observation}
(1) $\var(Z_1)\leq 1$,
(2) $\var(Z_3)\leq \frac{5}{4}$.
\end{observation}
\begin{proof}
Recall that $|f(\cdot)|\leq 1$.
We have $\var(Z_1)\leq \EE[Z_1^2]\leq \EE[\phi_S(X)^2]=1$ which proves (1).
We next show that $\var(Z_2|u)\leq 1$. 
This follows from
$\var(Z_2|u)\leq \EE[Z_2^2|u]\leq \EE[\basis{\lzb{\alpha}}({u\bfY})^2]=1$
where the last equality follows because, conditioned on $u$, the function $\basis{\lzb{\alpha}}({u\bfY})$ is a normalized basis function in the reduced $k$-dimensional space. 
To see this observe that conditioning on a root variable in a BN produces another BN with one less variable and the same conditional distributions (and basis functions) for other variables and since $u$ is an ancestor set, we can condition on all its variables in this manner.
The bound for $Z_2$ also implies  $g^2_\alpha(u) = \EE[Z_2|u]^2 \leq \EE[Z_2^2|u]\leq 1$, where we used Jensen's inequality.

Next, using a similar reasoning we 
show that $\var(Z_3|u)\leq 1$.
This follows from
$\var(Z_3|u)\leq \EE[Z_3^2|u]\leq$ $\EE_{\bfY_1|u} [\basis{\lzb{\alpha}}({u\bfY_1})^2]$ $\EE_{\bfY_2|u} [\basis{\lzb{\alpha}}({u\bfY_2})^2]=1$.

Finally, using the law of total variance we get:
$$
\var(Z_3) = \EE_U[\var(Z_3|U)] + \var_U(\EE[Z_3 |U]) \leq 1 + \var_U(g^2_\alpha(U))\leq \frac{5}{4},
$$
where in the last step we used $0\leq g^2_\alpha(u) \leq 1$ and Popoviciu's inequality to bound the variance of $g^2_\alpha(U)$.
\end{proof}

The MoM estimator
operates by dividing the samples into  multiple blocks of 
equal size, computing the mean within each block, and then returning the median of these means.
More formally,
we sample $m_1$ values for $U$, and conditioned on each $u^i$ sample a pair of values, $\bfY^{i}_1,\bfY^{i}_2$, independently. We calculate
\begin{equation*}
    Z^{i}_3  = f(u^{i}\bfY^{i}_1) f(u^{i}\bfY^{i}_2) \basis{\lzb{\alpha}}({u^{i}\bfY^{i}_1}) \basis{\lzb{\alpha}}({u^{i}\bfY^{i}_2}),
\end{equation*}
for each $i\in [m_1]$ and group $Z^{i}_3$'s into $k$ blocks of identical size, where $k$ is determined in the proof of Lemma \ref{lem:Galpha MOM}. 
For $j\in [k]$, 
let $G^j_\alpha$ be the mean of $Z^{i}_3$ in the $j$th block. The output of this estimation is 
\begin{equation}\label{eq:Galpha}
    G_\alpha = \text{Median}(G^1_\alpha, \cdots, G^k_\alpha).
\end{equation} 

Note that this differs from prior work where sampling is done in two stages, first estimating $Z_3|u$ and then estimating $G_\alpha$ and where Hoeffding's bound is applicable. 
Instead, we use MoM as the main estimator which yields an exponential concentration bound because the variance of $Z_3$ is bounded \citep{Lugosi2019}.  %

What is crucial for our case is that this only requires forward sampling in the BN: $u$ is sampled sequentially from the roots of the BN and $Y_1,Y_2$ are sampled conditional on $u$. That is, the process can be performed in so called ancestral sampling \cite{KollerFriedman2009}
where variables are sampled in their ordering from the BN and
$X_i$ is sampled from the conditional probability table (CPT) $p(X_i|X_{\pa(i)})$. 
Since all variables are binary, 
sampling can be done in polynomial time in the size of the CPT and hence polynomial in the representation size of the BN. 
Hence the algorithm is the same as the original KM algorithm, but it specifically constrains the sampling order over variables using the BN in this process. We have therefore argued that:

\begin{proposition}
The sampling required by the extended KM algorithm can be done in time polynomial in the size of the BN specification (number of nodes and size of CPTs).
\end{proposition}

With this in place, we proceed with the analysis of this estimator. 
\begin{lemma}\label{lem:Galpha MOM}
   Given $\delta \in (0, \tfrac{1}{2})$, with probability at least $1-\delta$, $|G_\alpha-\EE_{U} [g_\alpha^2(U)]| \leq \tfrac{\theta^2}{4}$ when 
   $m_1 = O(\tfrac{1}{\theta^4}\ln\tfrac{1}{\delta})$. 
\end{lemma}
\begin{proof}
    The standard analysis of MoM (see for example Proposition 12 of \citep{Lerasle2019}), with $k$ blocks of total $m_1$ samples, implies that 
    when $(\frac{1}{2}- \frac{k}{m_1} \frac{Var(Z_3)}{\epsilon^2})>0$
    \begin{equation*}
            P(\abs{G_\alpha-\EE_{U} [g_\alpha^2(U)]} > \epsilon) \leq \exp{-2k\qty(\frac{1}{2}- \frac{k}{m_1} \frac{Var(Z_3)}{\epsilon^2})^2}.
    \end{equation*}
    The RHS is smaller than $\delta$, by setting $k=\lceil 8\ln\tfrac{1}{\delta}\rceil,$ and $m_1=\tfrac{4k}{\epsilon^2}Var(Z_3)$. The lemma follows from $\epsilon=\tfrac{\theta^2}{4}$ and  the fact that $\var(Z_3)\leq \tfrac{5}{4}$.
\end{proof}

Therefore, with probability at least $1-\delta$, if  $\EE_U[g_\alpha^2(U)]\geq \theta^2$ then 
$G_\alpha \geq \tfrac{3}{4}\theta^2 \geq \tfrac{1}{2}\theta^2$. 
On the other hand if $\EE_U[g_\alpha^2(U)] <  \frac{1}{4} \theta^2$ then $G_\alpha < \theta^2/2$.
 As in \cite{KM1993}
the algorithm can then recurse if $G_\alpha\geq \theta^2/2$. With that, the algorithm outputs 
all $\alpha$ with 
$|\hat{f}_{\alpha}|\geq \theta$ and no $\alpha$ with  $|\hat{f}_{\alpha}|< \theta/2$. 
Therefore, the number of $\alpha$'s produced is bounded by $4/\theta^2$. 

In.a similar manner $\hat{f}_{\alpha}$ can be estimated using MoM estimator.
\begin{lemma}
    Each individual coefficient $\hat{f}_{\alpha}$ can be estimated to the required accuracy $\gamma$ with probability at least $1-\delta$, given $m_2=O(\tfrac{1}{\gamma^2}\ln \tfrac{1}{\delta})$ samples.
\end{lemma}
\begin{proof}
    We generate $m_2$ samples of $Z_1 = f(\bfX)\phi_S(\bfX)$ and apply MoM with 
    $k=\lceil 8\ln\tfrac{1}{\delta}\rceil,$ and $m_1=\tfrac{4k}{\epsilon^2}Var(Z_1)$.
    The claim follows from the fact that $\var(Z_1)\leq 1$.
\end{proof}

We have  established that each individual estimate of $G_\alpha$ and each individual estimate of a coefficient succeed with probability at least $1-\delta$. 
In total we have $\leq \frac{4n}{\theta^2}$ estimates  of $G_\alpha$ and $\leq \frac{4}{\theta^2}$ coefficient estimates. 
Scaling $\delta$ appropriately and taking a union bound,
the overall query complexity of the algorithm is  $${O}\qty(\frac{1}{\theta^{4}}\log \frac{n}{\delta\theta^2}+ \frac{1}{\gamma^2}\log \frac{1}{\delta\theta^2}).$$
The run time  has an additional poly$(|BN|)$ factor to generate samples and $O(n)$ factor to evaluate $\phi()$.
We have therefore established the following theorem.

\begin{theorem} [cf. Theorem~3.10 of \cite{KM1993}]
\label{thm:EKM}
Consider any distribution $D$ specified by a BN and its corresponding Fourier basis, and any Boolean function $f$. 
Algorithm $\mbox{KM}(D,f,\theta,\gamma,\delta)$ is given access to the BN representation of $D$ and a MQ oracle for $f$ and three accuracy parameters $\theta,\gamma,\delta$.
The algorithm runs in time polynomial in $n$,$1/\theta$,$1/\gamma$,$1/\delta$ and with probability at least $1-\delta$ returns a list of sets $\CA = \{\CS\}$, estimates of the corresponding coefficients $\approxcoeff{\CS}$, and a hypothesis
$h(x) =\sum_{\CS\in \CA} \approxcoeff{\CS} \basis{\CS}(x)$
such that  (1) $\CA$ includes all $\CS$ such that $|\coeff{\CS}|\geq \theta$, (2) $|\CA|\leq \frac{4}{\theta^2}$, (3)
for all $\CS\in \CA$, $|\approxcoeff{\CS}-\coeff{\CS}|\leq\gamma$. 
\end{theorem}

Several approaches have been developed to improve the query complexity in the uniform and product cases
\cite{BshoutyJT04,Kalai2009}. We leave exploration of these or other approaches that avoid using Chebyshev's bounds to future work. 

As in prior work, the KM algorithm leads to learnability results for DNF and we develop these ideas in Section~\ref{sec:learnDNF}.
A key requirement in the analysis is that the spectral norm of conjunctions is bounded. 
Hence, the next few sections develop upper and lower bounds on the spectral norm.

\section{Fourier Expansion of Conjunctions for Chain BN}\label{sec:chain}
Our first step in bounding  the spectral norm of conjunctions is to restrict our attention to linear chain BNs. The next section studies tree BNs. 
In a chain each node has a single parent so that w.l.o.g.\ we can rename the variables so that 
the parent of any node $i$ is node $i-1$. In that case, the BN induced basis is described by 
\begin{align*}
    \phi_{\CS}(\bfx) := \prod_{i\in \CS} \phi_{i}(\bfx) = \prod_{i\in \CS} \frac{x_i- \mu_{i,x_{i-1}}}{\sigma_{i, x_{i-1}}}. 
\end{align*}

Consider the conjunction $f(\bfx) : = \bigwedge_{i \in \CT_1} x_i \bigwedge_{j\in \CT_0} \bar{x}_j$ where $\CT_0$ and $\CT_1$ are disjoint subsets of $[n]$. 
In what follows we find the Fourier coefficients of this function. 
Given any $\CS\subseteq [n]$, we are interested in computing
\begin{align*}
\fS:=\<\ps, f\>&=\EE_D[\ps(\bfX) f(\bfX)].
\end{align*}

We introduce a series of notations using which we present a closed form expression for  $\fS$. First, we use $\Phi_{i}$ to denote 
$\phi_{i}(\bfX)$ 
which is understood as a random variable depending on $X_i$ and its parent $X_{i-1}$. 
Let $\CT = \CT_0\cup \CT_1$ and define the following random variable
\begin{align}\label{eq:Z_i}
\Z_{i} = \begin{cases} 
X_i & ~if~i \in \CT_1\backslash\CS\\
1-X_i & ~if~i \in \CT_0\backslash\CS\\
\Phi_i & ~if~i \in \CS\backslash\CT\\
\Phi_i X_i & ~if~i \in \CS\cap\CT_1\\
\Phi_i (1-X_i) & ~if~i \in \CS\cap\CT_0\\
1 & ~if~i\notin \CT\cup \CS.
\end{cases}
\end{align}

With this notation, it is not difficult to see that  
\begin{align}
\fS =\EE_D\qty\Big[\prod_{j\in [n]} Z_j]
=
\EE_{X_1}\qty[Z_1 \EE_{X_2}\qty[Z_2 \cdots \EE_{X_{n}}\qty[Z_n \Big | X_{n-1}] \cdots \Big | X_1] ],
\label{eq:chain-f-hat-general}
\end{align}
where we have used the fact that $X_1 \rightarrow X_2 \rightarrow \cdots \rightarrow X_n$ form a chain and that $Z_i$ is a function of $X_i, X_{i-1}$.

\subsection{Basic Iterative Forms} 

To analyze the expressions appearing in the iterative expectation, note that $\EE_{X_i}[X_i | X_{i-1}]=\mu_{i,X_{i-1}}$, 
and by  Lemma \ref{lem:pS properties}, $\EE_{X_i}[\Phi_i | X_{i-1}]=0$. In addition we have:
\begin{lemma} 
\label{lm:phix}
$\EE_{X_i}[\Phi_i X_i | X_{i-1}]=\sigma_{i,X_{i-1}}$.
\end{lemma}
\begin{proof} 
Since $X_i$ is binary we have $X_i=X_i^2$ and
\begin{align*}
\EE[\Phi_i X_i | X_{i-1}]
&
=
\EE[\phi_i(X) X_i | X_{i-1}]
=
\EE\qty[\frac{X_i-\mu_{i,X_{i-1}}}{\sigma_{i,X_{i-1}}} X_i | X_{i-1}]
\\
&
=
\EE\qty[X_i\frac{1- \mu_{i,X_{i-1}}}{\sigma_{i,X_{i-1}}}  | X_{i-1}]
=
\mu_{i,X_{i-1}} \frac{1- \mu_{i,X_{i-1}}}{\sigma_{i,X_{i-1}}}  
= 
\sigma_{i,X_{i-1}}.
\end{align*}
\end{proof}

This shows that in the iterative expectation we may get a term with $\mu$ or with $\sigma$.
The following definition provides the key to our analysis as it allows us to abstract 
the various cases in the same form, and by doing so to capture
the combinatorial structure of parameters $\fS$:

\begin{definition}[Recursive Form]
\label{def:A}
    For any $\CT_0, \CT_1$ and $\CS$ define
\begin{align*}
 \A^{\CT_0, \CT_1,\CS}_i(x_{i-1}) = 
 \begin{cases} 
\mu_{i,x_{i-1}} & ~if~i \in \CT_1\backslash\CS\\
1-\mu_{i,x_{i-1}} & ~if~i \in \CT_0\backslash\CS\\
\sigma_{i,x_{i-1}} & ~if~i \in \CS\cap\CT_1\\
-\sigma_{i,x_{i-1}} & ~if~i \in \CS\cap\CT_0\\
\sigma_{i,x_{i-1}} & ~if~i \in \CS\backslash\CT\\
\mu_{i,x_{i-1}} & ~if~i \not \in \CT\cup\CS.
\end{cases}
\end{align*}
For shorthand notation, we drop the $\CT, \CS$ dependence  in the notation and simply write $\A_i(x_{i-1})$. Moreover, with $A_i$ we denote the random variable $\A_i(X_{i-1})$ which is a function of $X_{i-1}$.  Note that $A_i$ satisfies the following identity
\begin{align}
\label{eq: Ai to X_i}
A_i = \A_i(X_{i-1}) = \A_i(0) + X_{i-1}(\A_i(1)-\A_i(0)).
 \end{align}
  \end{definition}

The next lemma analyzes the most basic terms $\EE[Z_i]$ that may appear in \eqref{eq:chain-f-hat-general}:

\begin{lemma}\label{lem:exp of Z}
The following holds:
\begin{align*}
\EE_{X_i}[Z_i | X_{i-1}]=
\begin{cases}
1 & ~if~i \not\in \CS\cup \CT\\
0 & ~if~i \in \CS\backslash \CT\\
\A^{\CT_0, \CT_1, \CS}_i(X_{i-1}) & ~if~i \in \CT.
\end{cases}
\end{align*}
\end{lemma}
\begin{proof} 
The cases can be verified using $\EE[X_i | x_{i-1}]=\mu_{i,X_{i-1}}$ which holds by definition,
$\EE[\Phi_i | X_{i-1}]=0$ which was shown in Lemma~\ref{lem:pS properties}, and 
$\EE[\Phi_i X_i | X_{i-1}]=\sigma_{i,X_{i-1}}$
which was shown in Lemma~\ref{lm:phix}.
\end{proof}

Hence, we see that when evaluating \eqref{eq:chain-f-hat-general} we may inherit a $A_i$ term from the child and need to evaluate $\EE \qty[Z_{i-1} A_{i}]$. The following lemma develops a compact form for this expectation:

\begin{lemma}\label{lem:Z recursion}
For any distribution for $X_{i-1} | X_{i-2}$,
$\EE_{X_{i-1}}\qty[Z_{i-1} A_{i}| X_{i-2}] =b_i A_{i-1}  + c_i,$
where 
\begin{align*}
\B_{i} = \begin{cases}
\A_i(1) & ~if~i-1\in \CT_1\\
\A_i(0) & ~if~i-1\in \CT_0\\
\A_i(1)-\A_i(0) & ~if~i-1\notin \CT_0\cup \CT_1,
\end{cases}
\end{align*}
and
\begin{align*}
\C_{i} = \begin{cases}
\A_i(0)& ~if~i-1\notin \CS \cup \CT\\
0 & ~otherwise.
\end{cases}
\end{align*}
\end{lemma}

\begin{proof}
Note that 
\begin{align*}
\EE_{X_{i-1}}\qty[Z_{i-1} A_{i}| X_{i-2}] = \A_i(0)\EE_{X_{i-1}}\qty[Z_{i-1}| X_{i-2}]  + \EE_{X_{i-1}}\qty[Z_{i-1}X_{i-1}| X_{i-2}](\A_i(1)-\A_i(0)).
\end{align*}
Below we use Lemma~\ref{lem:exp of Z} and the fact that for binary variable $X$ we have $X^2=X$ to analyze this expression. 
We have the following cases:

\noindent\textbf{Case 1:} Suppose $i-1 \notin \CT_0\cup\CT_1 \cup \CS$. Then $Z_{i-1} = 1$ and the  expectation equals 
\begin{align*}
\A_i(0) + \EE_{X_{i-1}}\qty[X_{i-1}| X_{i-2}] (\A_i(1)-\A_i(0)) & = \A_i(0)+ \mu_{i-1,X_{i-2}} (\A_i(1)-\A_i(0)) 
 \\ & = \A_i(0)+A_{i-1}(\A_i(1)-\A_i(0)).
\end{align*}

\noindent\textbf{Case 2.1:} Suppose $i-1 \in \CT_1 \backslash \CS$. Then $Z_{i-1} = X_{i-1}$ and the expectation equals 
\begin{align*}
\A_i(0)\A_{i-1}+\EE_{X_{i-1}}\qty[X_{i-1}^2| X_{i-2}](\A_i(1)-\A_i(0)) &= \A_i(0)\A_{i-1}+A_{i-1}(\A_i(1)-\A_i(0))\\
&=\A_i(1)A_{i-1}.
\end{align*}

\noindent\textbf{Case 2.2:} Suppose $i-1 \in \CT_0 \backslash \CS$. Then $Z_{i-1}= 1-X_{i-1}$ and 
the expectation equals 
\begin{align*}
\A_i(0)\A_{i-1}+\EE_{X_{i-1}}\qty[X_{i-1}(1-X_{i-1})| X_{i-2}](\A_i(1)-\A_i(0)) &= \A_i(0)\A_{i-1}.
\end{align*}

\noindent\textbf{Case 3:} Suppose $i-1 \in \CS \backslash \CT$. Then $Z_{i-1} = \Phi_{i-1}$ which zeros out the first expectation. 
Using Lemma~\ref{lm:phix} the second term equals
\begin{align*}
\EE_{X_{i-1}}\qty[\Phi_{i-1}X_{i-1}| X_{i-2}](\A_i(1)-\A_i(0)) &
= \sigma_{i-1,X_{i-2}} (\A_i(1)-\A_i(0))
= \A_{i-1}(\A_i(1)-\A_i(0)).
\end{align*}

\noindent\textbf{Case 4.1:} Suppose $i-1 \in \CS \cap \CT_1$.  Then $Z_{i-1} = \Phi_{i-1}X_{i-1}$ and 
using 
Lemma~\ref{lm:phix}
the expectation equals
\begin{align*}
\A_i(0)\A_{i-1} + \EE_{X_{i-1}}\qty[\Phi_{i-1}X_{i-1}^2| X_{i-2}] (\A_i(1)-\A_i(0))&= \A_i(0)\A_{i-1} +\A_{i-1}(\A_i(1)-\A_i(0))\\
&=\A_i(1)A_{i-1}.
\end{align*}

\noindent\textbf{Case 4.2:} Suppose $i-1 \in \CS \cap \CT_0$.  Then $Z_{i-1} = \Phi_{i-1}(1-X_{i-1})$ and the expectation equals
\begin{align*}
\A_i(0)\A_{i-1} + \EE_{X_{i-1}}\qty[\Phi_{i-1}(1-X_{i-1})X_{i-1}| X_{i-2}] (\A_i(1)-\A_i(0))&= \A_i(0)\A_{i-1}.
\end{align*}
Putting together all the cases, we have that 
\begin{align*}
\EE_{X_{i-1}}\qty[Z_{i-1} A_{i}| X_{i-2}] = \begin{cases} \A_i(0)+A_{i-1}(\A_i(1)-\A_i(0)) & if~i-1 \notin \CT \cup \CS\\
\A_i(1)A_{i-1} & if~i-1 \in \CT_1 \backslash \CS\\
\A_i(0)A_{i-1} & if~i-1 \in \CT_0 \backslash \CS\\
 \A_{i-1}(\A_i(1)-\A_i(0)) & if~i-1 \in \CS \backslash \CT\\
\A_i(1)A_{i-1} & if~i-1 \in \CS \cap \CT_1\\
\A_i(0)A_{i-1} & if~i-1 \in \CS \cap \CT_0.
\end{cases}
\end{align*}
\end{proof}

The previous lemma allows us to recurse over all the iterative expectations to compute $\fS$.
The additive term $\C_i$ yields hierarchically structured expressions with a constant term and a linear term.
However, 
note that $\C_i$ will be canceled if the next expectation $j<i$ where $Z_j\not=1$ is $j\in \CS\setminus\CT$. 
That is $\EE[\Phi_{j}(c_i+b_i A_{i-1})] = \EE[\Phi_{j} b_i A_{i-1}]$
because $\EE[\Phi_{j}]=0$.  
This is crucial in understanding the structure of $\fS$.

\subsection{The Fourier Coefficients of Conjunctions}
We first present the following example to illustrate the combinatorial structure of the Fourier coefficients.
\begin{example}
Consider a chain with $n=10$ variables and the monotone conjunctions $f=X_2 X_6 X_{10}$ and $g=X_1 X_6 X_{10}$. We calculate $\coeffgen{f}{S}$ and $\coeffgen{g}{S}$ for $S=\{3,6,8\}$.
To facilitate the presentation, for variable $i$ we denote $D_{i,\mu}=\mu_{i,1}-\mu_{i,0}$ and $D_{i,\sigma}=\sigma_{i,1}-\sigma_{i,0}$.
We implicitly use the facts $\EE[X_i]=\mu_{i,X_{i-1}}$, $\EE[\phi_i]=0$, $\EE[\phi_i X_i]=\sigma_{i,X_{i-1}}$,
$\mu_{i,0}+D_{i,\mu}=\mu_{i,1}$, $\sigma_{i,0}+D_{i,\sigma}=\sigma_{i,1}$ and 
$(a+ b X_i) X_i=(a+b)X_i$ as derived above.
Starting with $f$ we have:
\begin{align*}
\coeffgen{f}{S} 
& 
= \EE[X_2 \phi_3 \phi_6 X_6 \phi_8 X_{10}] 
\\ & 
= \EE[X_2 \phi_3 \phi_6 X_6 \phi_8 (\mu_{10,0}+D_{10,\mu} X_9)]
\\ & 
= \EE[X_2 \phi_3 \phi_6 X_6 \phi_8 (\mu_{10,0}+D_{10,\mu}\mu_{9,0} + D_{10,\mu}D_{9,\mu} X_8)]
\\ & 
= (D_{10,\mu}D_{9,\mu}) \EE[X_2 \phi_3 \phi_6 X_6 \phi_8  X_8]
\\ & 
= (D_{10,\mu}D_{9,\mu}) \EE[X_2 \phi_3 \phi_6 X_6 (\sigma_{8,0}+D_{8,\sigma} X_7)]
\\ & 
= (D_{10,\mu}D_{9,\mu}) \EE[X_2 \phi_3 \phi_6 X_6 (\sigma_{8,0}+D_{8,\sigma} \mu_{7,0} +D_{8,\sigma} D_{7,\mu} X_6)]
\\ & 
= (D_{10,\mu}D_{9,\mu}) \EE[X_2 \phi_3 \phi_6 (\sigma_{8,0}+D_{8,\sigma} \mu_{7,0} +D_{8,\sigma} D_{7,\mu})X_6 ]
\\ & 
= [(D_{10,\mu}D_{9,\mu}) (\sigma_{8,0}+D_{8,\sigma} \mu_{7,1})] \EE[X_2 \phi_3 \phi_6 X_6 ]
\\ & 
= [(D_{10,\mu}D_{9,\mu}) (\sigma_{8,0}+D_{8,\sigma} \mu_{7,1})] \EE[X_2 \phi_3  (\sigma_{6,0}+D_{6,\sigma} X_5)]
\\ & 
= [(D_{10,\mu}D_{9,\mu}) (\sigma_{8,0}+D_{8,\sigma} \mu_{7,1})] 
\EE[X_2 \phi_3  (\sigma_{6,0}+D_{6,\sigma} \mu_{5,0} +D_{6,\sigma} D_{5,\mu} X_4)]
\\ & 
= ([D_{10,\mu}D_{9,\mu}) (\sigma_{8,0}+D_{8,\sigma} \mu_{7,1})] 
(D_{6,\sigma} D_{5,\mu}) \EE[X_2 \phi_3  X_4)]
\\ & 
= [(D_{10,\mu}D_{9,\mu}) (\sigma_{8,0}+D_{8,\sigma} \mu_{7,1})] 
(D_{6,\sigma} D_{5,\mu}) \EE[X_2 \phi_3  (\mu_{4,0}+D_{4,\mu} X_3))]
\\ & 
= [(D_{10,\mu}D_{9,\mu}) (\sigma_{8,0}+D_{8,\sigma} \mu_{7,1})] 
(D_{6,\sigma} D_{5,\mu} D_{4,\mu} ) 
\EE[X_2 \phi_3  X_3]
\\ & 
= [(D_{10,\mu}D_{9,\mu}) (\sigma_{8,0}+D_{8,\sigma} \mu_{7,1})] 
(D_{6,\sigma} D_{5,\mu} D_{4,\mu} ) 
\EE[X_2 (\sigma_{3,0}+D_{3,\sigma} X_2)]
\\ & 
= [(D_{10,\mu}D_{9,\mu}) (\sigma_{8,0}+D_{8,\sigma} \mu_{7,1})] 
[(D_{6,\sigma} D_{5,\mu} D_{4,\mu} ) (\sigma_{3,1})]
\EE[X_2]
\\ & 
= [(D_{10,\mu}D_{9,\mu}) (\sigma_{8,0}+D_{8,\sigma} \mu_{7,1})] 
[(D_{6,\sigma} D_{5,\mu} D_{4,\mu} ) (\sigma_{3,1})]
\EE[(\mu_{2,0}+D_{2,\mu} X_1)]
\\ & 
= [(D_{10,\mu}D_{9,\mu}) (\sigma_{8,0}+D_{8,\sigma} \mu_{7,1})] 
[(D_{6,\sigma} D_{5,\mu} D_{4,\mu} ) (\sigma_{3,1})]
[()(\mu_{2,0}+D_{2,\mu} \mu_1)].
\end{align*}
In these equations we identified 3 segments in $[\ldots]$ where each segment includes a $D$ component and a hierarchical component. 
The segments $[7-10],[3-6],[1-2]$ are split at indices in the conjunction $f$ and the split into $D$ component corresponds to the (least) $S$ index in the corresponding region. 
For example, the least $S$ index in $[7-10]$ is 8 and this segment splits into a $D$ portion $[9,10]$ and hierarchical portion $[7,8]$.
In this example the segment $[1-2]$ does not have a $S$ element and correspondingly no $D$ component, which we emphasize with the empty parenthesis.

Turning to the function $g$, all the steps up to the expectation over $X_2$ are the same (replacing $X_2$ with $X_1$ in the conjunction). Continuing from there we have

\begin{align*}
\coeffgen{g}{S} 
& 
= [(D_{10,\mu}D_{9,\mu}) (\sigma_{8,0}+D_{8,\sigma} \mu_{7,1})] 
(D_{6,\sigma} D_{5,\mu} D_{4,\mu} ) 
\EE[X_1 (\sigma_{3,0}+D_{3,\sigma} X_2)]
\\ & 
= [(D_{10,\mu}D_{9,\mu}) (\sigma_{8,0}+D_{8,\sigma} \mu_{7,1})] 
(D_{6,\sigma} D_{5,\mu} D_{4,\mu} ) 
\EE[X_1 (\sigma_{3,0}+D_{3,\sigma} \mu_{2,0}+D_{3,\sigma} D_{2,\mu} X_1)]
\\ & 
= [(D_{10,\mu}D_{9,\mu}) (\sigma_{8,0}+D_{8,\sigma} \mu_{7,1})] 
[(D_{6,\sigma} D_{5,\mu} D_{4,\mu} ) (\sigma_{3,0}+D_{3,\sigma} \mu_{2,1})]
\EE[X_1]
\\ & 
= [(D_{10,\mu}D_{9,\mu}) (\sigma_{8,0}+D_{8,\sigma} \mu_{7,1})] 
[(D_{6,\sigma} D_{5,\mu} D_{4,\mu} ) (\sigma_{3,0}+D_{3,\sigma} \mu_{2,1})]
[()(\mu_1)].
\end{align*}
We see that the same structure arises when $i=1$ is in the conjunction. In this case the last $D$ component is always empty and the hierarchical component is just $\mu_1$.

\end{example}

Our proofs below show that this type of structure 
where each segment is split into a 
$D$ component and a hierarchical component
always arises and therefore leads to a bounded L1 norm. 

Consider an auxiliary node $0$ at the top of the chain BN that is not connected to any other nodes. Therefore, $\EE[X_1|X_0=0] = \EE[X_1|X_0=1]$; implying that $A_1(1)=A_1(0)$; hence $A_1=A_1(0)$.

\begin{definition}\label{def:A' D'}
For any $i\in [n]$ define $\D_i := A_i(1)-A_i(0)$. Given any pair of  subsets $\CS$ and $\CT$, let   $a_1<\cdots<a_m$ denote the the ordered elements of $\CS\cup\CT$. For $i\in [m]$ and any  $r\in [a_{i-1}, a_{i}]$ define
\begin{align}\label{eq:A' D'}
D'_{r, a_i} &= \prod_{r\leq \ell \leq a_{i}} D_\ell, & &  A'_{r, a_i}(0) = \sum_{r \leq \ell \leq  a_{i}} D'_{(\ell+1, a_i)}  A_{\ell}(0),  %
\end{align}
and by convention $D'_{r, a_i}=1$ for any $r>a_i$. Let $A'_{r, a_i}(1) = A'_{r, a_i}(0) + D'_{r, a_i}$ and define the random variable 
$A'_{r, a_i}:= A'_{r, a_i}(0) + X_{r-1}(A'_{r, a_i}(1) -A'_{r, a_i}(0) ) = A'_{r, a_i}(0) + X_{r-1}D'_{r, a_i}$.
Note that when $r>a_i$ then $A'_{r,a_i}(0)=0$ and  $A'_{r,a_i}(1)=D'_{r, a_i}=1$.  

\end{definition}
As a special case of the above definition $D'_{a_i, a_i}=D_{a_i}$ and 
$A'_{a_i, a_i}=A_{a_i}$.  
We highlight that $A', A'(1), A'(0)$  and $D'$ satisfy the same relations as  $A, A(1), A(0)$  and $D$.
Moreover, observe that 
\begin{lemma}\label{A'alt}
$A'_{(r,a_i)} = A'_{(r+1,a_i)}(0)+A_{r}D'_{(r+1,a_i)}.$ 
\end{lemma}
\begin{proof}
For the left term note that
$A'_{r, a_i}=  A'_{r, a_i}(0) + X_{r-1}D'_{r, a_i} = A'_{r+1, a_i}(0) + A_r(0) D'_{r+1, a_i}  + X_{r-1}D'_{r, a_i}$.
On the other hand expanding $A_{r}$ in the right term we have
$$A_{r}D'_{(r+1,a_i)} = A_r(0) D'_{r+1, a_i}  + X_{r-1}D_r D'_{r+1, a_i} = A_r(0) D'_{r+1, a_i}  + X_{r-1} D'_{r, a_i}.$$
\end{proof}

As in the example, we split the chain into segments and analyze the expectation calculating $\fS$ over segments. 
The next two lemmas prepare the ground by analyzing individual segments with boundaries at $T$ nodes. 
\begin{lemma}[Single $\CS\cup\CT$ segment]\label{lem:single S T segment}
Consider a chain of nodes  $X_0\rightarrow \cdots \rightarrow X_{n+1}$ that form a segment in the sense that  $\CS, \CT\subseteq \{1\}$. Then,  
\begin{align*}
\EE_{X_{1},\ldots,X_{n}}\qty[Z_1 A'_{(n+1,n+1)}(X_n)| X_{0}] = \begin{cases}
\A'_{(1, n+1)}(X_0) & ~if~ 1\notin \CT\cup \CS\\
\D'_{(2, n+1)} A_{1}(X_0) & ~if~ 1\in \CS\backslash \CT\\
\A'_{(2, n+1)}(1) A_{1}(X_0) & ~if~ 1\in \CT_1\\
\A'_{(2, n+1)}(0)A_{1}(X_0) & ~if~ 1\in \CT_0.
\end{cases}
\end{align*}
\end{lemma}

\begin{proof}
Assume $n>1$ as the lemma follows from Lemma \ref{lem:Z recursion} when $n=1$.  By the law of total expectation, the  expectation in the lemma iterates as
\begin{align*}
\EE_{X_1}\qty[Z_1 \EE_{X_2}\qty[ \cdots \EE_{X_{n}}\qty[A'_{(n+1,n+1)} \Big | X_{n-1}]. \cdots \Big | X_1] \Big | X_0].
 \end{align*} 

By an inductive argument we prove that 
\begin{align}\label{eq:empty segment induction}
\EE_{X_2}\qty[ \cdots \EE_{X_{n}}\qty[A'_{(n+1,n+1)} \Big | X_{n-1}]. \cdots \Big | X_1]=A'_{(2,n+1)}.
 \end{align} 
Observe that $A'_{(n+1,n+1)}(X_n) = A_{n+1}(X_n)$. From Lemma \ref{lem:Z recursion}, as $n\notin \CS\cup\CT$, the innermost expectation equals
\begin{align*}
\EE_{X_{n}}\qty[ A'_{(n+1,n+1)}| X_{n-1}] = D_{n+1}A_n + A_{n+1}(0) = A'_{(n,n+1)},
\end{align*}
where the last equality follows from Lemma~\ref{A'alt}.

Assuming that $\EE_{X_m, \cdots, X_n}\qty[A'_{(n+1,n+1)} \Big | X_{m-1}]=A'_{(m,n+1)}$ holds for a fixed $2 < m \leq n$, for $m-1$ we have 
\begin{align*}
\EE_{X_{m-1}, \cdots, X_n}\qty[A'_{(n+1,n+1)} \Big | X_{m-1}]&=\EE_{X_{m-1}}\qty[ A'_{(m,n+1)}|X_{m-2}]\\
&=A'_{(m,n+1)}(0) + D'_{(m,n+1)}\EE_{X_{m-1}}\qty[ X_{m-1}|X_{m-2}]\\
&=A'_{(m,n+1)}(0) + D'_{(m,n+1)}A_{m-1} = A'_{(m-1,n+1)},
\end{align*}
where the last step holds because $m-1\not\in\CS\cup\CT$ and $\EE[X_{m-1}| X_{m-2} ]=\mu_{m-1,X_{m-2}}=A_{m-1}$. 
Therefore, with this inductive argument, we proved \eqref{eq:empty segment induction}. 
Hence it remains to compute the outermost expectation. From Lemma \ref{lem:exp of Z} and \ref{lem:Z recursion}, have that
   \begin{align*}
   \EE_{X_1}\qty[Z_1 A'_{(2,n+1)}\Big | X_0]  &\stackrel{(a)}{=} A'_{(2,n+1)}(0) \EE_{X_1}\qty[Z_1\big | X_0]+ D'_{(2,n+1)}\EE_{X_{1}}\qty[Z_1 X_{1}|X_{0}]\\\numberthis\label{eq:single segment outer exp}
   & \stackrel{(b)}{=} A'_{(2,n+1)}(0) (\one{1\notin \CT\cup \CS} +   A_1\one{1 \in \CT}) +D'_{(2,n+1)}A_1 \one{1\notin \CT_0},
   \end{align*}
where (a) follows from Definition \ref{def:A' D'}.
For (b) note that when $1\in \CT_0$ we have $\EE_{X_{1}}\qty[Z_1 X_{1}|X_{0}]=0$ because $X_1(1-X_1)=0$.
We then have  $\EE_{X_{1}}\qty[Z_1 X_{1}|X_{0}] = A_1 \one{1\notin \CT_0}$ that holds from the definition of $A_1$ and
   \begin{align*}
    \EE_{X_{1}}\qty[Z_1 X_{1}|X_{0}] = \mu_{1, X_0}\one{1\notin \CT\cup \CS}+\mu_{1, X_0}\one{1\in \CT_1 \setminus \CS} + \sigma_{1, X_0}\one{1\in \CS\setminus \CT} +  \sigma_{1, X_0}\one{1\in \CS\cap \CT_1}.
   \end{align*}
 Separating \eqref{eq:single segment outer exp} for the different memberships of $1$ gives the desired expression. 
\begin{align*}
    \EE_{X_1}\qty[Z_1 A'_{(2,n+1)}\Big | X_0]= \begin{cases} D'_{(2,n+1)}A_1+A'_{(2,n+1)}(0)  = A'_{(1,n+1)} & if~1 \notin \CT \cup \CS\\
        A'_{(2,n+1)}(0) A_1+ D'_{(2,n+1)}A_1 = A'_{(2,n+1)}(1)A_1 & if~1 \in \CT_1 \backslash \CS\\
        A'_{(2,n+1)}(0)A_1  & if~1 \in \CT_0 \backslash \CS\\
        D'_{(2,n+1)}A_1 & if~1 \in \CS \backslash \CT\\
        A'_{(2,n+1)}(0) A_1+ D'_{(2,n+1)}A_1 =A'_{(2,n+1)}(1)A_1   & if~1 \in \CS \cap \CT_1\\
        A'_{(2,n+1)}(0)A_1  & if~1 \in \CS \cap \CT_0.
    \end{cases}
    \end{align*}

\end{proof}
\begin{remark}
The lemma implies an alternative definition of $A'$ that is insightful. Consider $A'_{(r, a_i)}$ as in Definition \ref{def:A' D'} and an empty segment  $[r,a_i)$, then Lemma \ref{lem:single S T segment} (with $1\leftarrow r, n\leftarrow a_i-1$) implies that 
\begin{align}\label{eq:A'_r as cond exp}
A'_{(r, a_i)}(X_{r-1}) = \EE_{X_r,\ldots, X_{a_i-1}}\qty[A_{a_i} | X_{r-1}].
\end{align}
\end{remark}

\begin{lemma}[Single $\CT$-segment, multiple $\CS$-segments]\label{lem:single T multiple S segments}
Consider a chain $X_0 \rightarrow \cdots \rightarrow X_{n+1}$ that potentially has multiple $\CS$ nodes inside a $\CT$ segment,  
where 
$\CT\subseteq \{1,n\}$ and 
$\CS\subseteq [n]$.  

Let  $\Gamma = \EE\qty[ \prod_{1\leq j\leq n} Z_j A'_{(n+1,n+1)}\Big| X_{0}]$. 
If $\CS\setminus\CT$ is not empty, define $s_{\circ} := \min \CS\setminus \CT$. 
Then,
\begin{align*}
    \Gamma = 
\begin{cases}
(a):~if~ \CT=\emptyset~\text{then} & \D'_{(s_\circ+1, n+1)} \A'_{(1, s_\circ)}(X_0),\\
(b):~if~ \CT=\{1\}~\text{then}     & \D'_{(s_\circ+1, n+1)} \A'_{(2, s_\circ)}(y_1) A_1(X_0), \\
(c):~if~ \CT=\{n\}~\text{then}     & A_{n+1}(y_n) \D'_{(s_\circ+1, n)} \A'_{(1, s_\circ)}(X_0),\\
(d):~if~ \CT=\{1, n\}~\text{then}  & A_{n+1}(y_n) \D'_{(s_\circ+1, n)} \A'_{(2, s_\circ)}(y_1) A_1(X_0).\\
\end{cases}
\end{align*}
where  $y_n =\mathbf{1}_{(n\in \CT_1)}$, and $y_1 =\mathbf{1}_{(1\in \CT_1)}$. If $\CS\setminus\CT$ is empty, then 
\begin{align*}
    \Gamma = 
\begin{cases}
(e):~if~ \CT=\emptyset~\text{then} & \A'_{(1, n+1)}(X_0) ,\\
(f):~if~ \CT=\{1\}~\text{then}     & \A'_{(2, n+1)}(y_1) A_1(X_0), \\
(g):~if~ \CT=\{n\}~\text{then}     & A_{n+1}(y_n) \A'_{(1, n)}(X_0),\\
(h):~if~ \CT=\{1, n\}~\text{then}  &  A_{n+1}(y_n) \A'_{(2, n)}(y_1) A_1(X_0).\\
\end{cases}
\end{align*}

\end{lemma}
\begin{proof}
The proof follows  by the law of iterative expectations and by breaking the chain into segments and applying Lemma \ref{lem:single S T segment} on each segment. Let $\Gamma$ be the expectation of interest as in the lemma's statement. If $\CS\setminus \CT$ is not empty, let $s_\circ = a_1<\cdots < a_k$ be the ordered elements of $\CS\setminus \CT$ with $k=|\CS\setminus \CT|$.  We consider four cases depending on $\CT$ elements.

\textbf{Case 1 ($\CT=\emptyset$):} In this case, the segments are $[1, a_1)$, $[a_i, a_{i+1})$ for $i=1,\ldots, k-1$, and $[a_k, n+1)$. %
Starting from  the tail segment $[a_k, n+1)$, from Lemma \ref{lem:single S T segment}, since $a_k\in \CS\setminus\CT$ the contribution  is 
\begin{align*}
\Gamma_{[a_k, n+1)} = \EE_{X_{a_k},\ldots, X_{n}}\qty[ Z_{a_k} A'_{(n+1,n+1)}(X_n)| X_{a_k-1}] = D'_{(a_k+1,n+1)} A_{a_k}(X_{a_k-1}).
 \end{align*}
 Note that $A_{a_k}=A'_{(a_k, a_k)}$, implying that the input to the next segment is an $A'$ term. From this point for any $\CS$-segment $[a_i, a_{i+1})$, the contribution is 
    \begin{align}\label{eq:cont a_i a_i+1}
\Gamma_{[a_i, a_{i+1})} =    \EE_{X_{a_i},\ldots, X_{a_{i+1}-1}}\qty[ Z_{a_i}  A'_{(a_{i+1}, a_{i+1})}| X_{a_i-1}] = D'_{(a_i+1, a_{i+1})} A'_{(a_i, a_i)},
    \end{align}
    where we used the fact that $A'_{(a_i, a_i)}=A_{a_i}$.   Continue this argument until the head segment, if $s_\circ=1$, the head segment is $[a_1, a_2)$ that has been covered and contributing $A'_{(1, 1)}$. If $s_\circ>1$, then the head segment is $[1, s_\circ)$. The input to this segment is $A'_{(a_1, a_1)}$. Again from Lemma \ref{lem:single S T segment}, as $1\notin \CS\cup\CT$, the contribution is 
    \begin{align}\label{eq:cont 1 s_0}
   \Gamma_{[1, s_\circ)} = \EE_{X_{1},\ldots,X_{s_{\circ}-1}}\qty[ Z_1 A'_{(s_{\circ}, s_{\circ})}| X_{0}] = A'_{(1, s_{\circ})}(X_0).
    \end{align}
Notice that the notation $ A'_{(1, s_{\circ})}$ in \eqref{eq:cont 1 s_0} is consistent with the case where $s_\circ=1$.  Lastly, multiplying  all the contributions gives the desired expression:
    \begin{align}\label{eq:1T mS case 1}
\Gamma_{\text{case~1}} =    D'_{(a_k+1,n+1)} \prod_{i=1}^{k-1} D'_{(a_i+1, a_{i+1})} A'_{(1, s_{\circ})}(X_0) = D'_{(s_\circ+1,n+1)}A'_{(1, s_{\circ})}(X_0),
    \end{align}
 where the last equality holds as  products of consecutive $D'$ terms gives another $D'$ with the total interval. 
Now consider  $\CS\setminus\CT=\emptyset$. Since $\CT$ is empty then so is $\CS$. Hence, we have an empty chain. From Lemma \ref{lem:single S T segment}, the contribution is  $A'_{(1, n+1)}(X_0)$. %

\textbf{Case 2 ($\CT=\{1\}$):} 
We first assume $\CS\setminus\CT$ is not empty. 
There are one $\CT$ node and potentially multiple $\CS$ nodes in the chain. The segments are $[1, a_1)$, $[a_i, a_{i+1})$ for $i=1,\ldots, k-1$, and $[a_k, n+1)$. The contributions of the segments are the same as the previous case except the last segment $[1, a_1)$.  If $a_1=1$, the last segment is empty and we inherit $A'_{(1,1)}=A_1(X_0)$. Otherwise, from Lemma \ref{lem:single S T segment}, the contribution is 
    \begin{align}\label{eq:case 2 cont 1 s_0}
   \Gamma_{[1, s_\circ)} = \EE_{X_{1},\ldots,X_{s_{\circ}-1}}\qty[ Z_1 A'_{(s_{\circ}, s_{\circ})}| X_{0}] = A'_{(2, s_{\circ})}(y_1)A_1(X_0),
    \end{align}
    where $y_1=\mathbf{1}_{(1\in \CT_1)}$. Note that the above equation reduces to $A_1(X_0)$ when $a_1=1$ and hence it is consistent with the empty segment case. This is because, when $a_1=1$, then  $s_\circ=1$ and $A'_{(2, s_{\circ})}=1$.
    Therefore, multiplying  all the contributions gives the desired expression:
    \begin{align}\label{eq:1T mS case 2}
\Gamma_{\text{case~2}} =    D'_{(s_\circ+1,n+1)}  A'_{(2, s_{\circ})}(y_1)A_1(X_0).
    \end{align}
If $\CS\setminus\CT$ is empty, then we have a single $\CT$ segment $[1,n+1]$. From Lemma \ref{lem:single S T segment} the contribution is $A'_{(2, n+1)}(y_1)A_1(X_0)$. %

\textbf{Case 3 ($\CT=\{n\}$):}  Assuming that $\CS\setminus\CT$ is not empty, the segments are $[1, a_1)$, $[a_i, a_{i+1})$ for $i=1,\ldots, k-1$, and $[a_k, n)$ and $[n, n+1)$. Starting from the tail segment, from   Lemma \ref{lem:Z recursion}, 
\begin{align*}
\Gamma_{[n, n+1)} = \EE_{X_{n}}\qty[ Z_n A'_{(n+1,n+1)}(X_n)| X_{n-1}] = b_{n+1} A_n = A_{n+1}(y_n) A'_{(n,n)}(X_{n-1}), 
\end{align*}
where $y_n =\mathbf{1}_{(n\in \CT_1)}$. The $\CS$-segments in between contribute $\Gamma_{[a_i, a_{i+1})} $ as in \eqref{eq:cont a_i a_i+1}.  As for the last segment $[1, s_{\circ})$, if $s_\circ=1$, this segment is empty and hence produces $A'_{(1,1)}(X_0)$. Otherwise from Lemma \ref{lem:single S T segment}, as $1\notin \CS\cup \CT$, the last contribution is $A'_{(1, s_{\circ})}(X_0)$. As a result, the overall contribution is 
\begin{align}\label{eq:1T mS case 3}
    \Gamma_{\text{case~3}} = A_{n+1}(y_n)   D'_{(s_\circ+1,n)}  A'_{(1, s_{\circ})}(X_0).
\end{align}
If $\CS\setminus\CT$ is empty, then  we have two segments: an empty segment $[1, n)$ and a $\CT$ (or $\CS\cap\CT$) segment $[n, n+1)$. For the second the contribution is $A_{n+1}(y_n)$, and for the first, it is $A'_{(1,n)}(X_0)$ per  Lemma \ref{lem:single S T segment}.  %

\textbf{Case 4 ($\CT=\{1,n\}$):} Assuming that $\CS\setminus\CT$ is not empty, the segments are  $[1, a_1)$, $[a_i, a_{i+1})$ for $i=1,\ldots, k-1$, and $[a_k, n)$ and $[n, n+1)$. Starting from the tail segment, from   Lemma \ref{lem:Z recursion}, 
\begin{align*}
\Gamma_{[n, n+1)} = \EE_{X_{n}}\qty[ Z_n A'_{(n+1,n+1)}(X_n)| X_{n-1}] = b_{n+1} A_n = A_{n+1}(y_n) A'_{(n,n)}(X_{n-1}), 
\end{align*}
where $y_n =\mathbf{1}_{(n\in \CT_1)}$.

The next segment is a $\CS$ segment $[a_k, n)$ and, from Lemma \ref{lem:single S T segment},  
\begin{align*}
\Gamma_{[a_k, n)} = \EE_{X_{a_k}\ldots, X_{n-1}}\qty[ Z_{a_k} A'_{(n,n)}(X_{n-1})| X_{a_k-1}] = D'_{(a_k+1,n)} A'_{(a_k, a_k)}.
    \end{align*}  
    Again, for any $\CS$-segment $[a_i, a_{i+1})$ the contribution is given by $\Gamma_{[a_i, a_{i+1})},$ as in \eqref{eq:cont a_i a_i+1}.  The last segment is $[1, s_{\circ})$ contributing $\Gamma_{[1, s_\circ)}$ as in 
    \eqref{eq:case 2 cont 1 s_0}. 
    Combining all the contributions gives
    \begin{align}\nonumber
    \Gamma_{\text{case~4}} &= A_{n+1}(y_n) D'_{(a_k+1,n)} \prod_{i=1}^{k-1} D'_{(a_i+1, a_{i+1})} A'_{(2, s_{\circ})}(y_1)A_1(X_0)\\\label{eq:1T mS case 4}
    & = A_{n+1}(y_n)D'_{(s_\circ+1,n)}A'_{(2, s_{\circ})}(y_1)A_1(X_0).
    \end{align}
   If $\CS\setminus\CT$ is empty, we have two $\CT$ (or $\CS\cap\CT$) segments $[1, n)$ and $[n, n+1)$. For the second the contribution is $A_{n+1}(y_n)$, and for the first, it is $A'_{(2,n)}(y_1)A_1(X_0)$ per  Lemma \ref{lem:single S T segment}. %
\end{proof}

The lemma implies that $\CS$ segments inside a $\CT$ segment contribute a $D'$ term which is a product of $D$'s from the the minimum $\CS$-only node to the end $\CT$ node. 
We are finally ready to express $\fS$.
\begin{lemma}[Fourier coefficients of conjunctions]\label{lem:Fourier exact}
Consider a conjunction $f$ with $\CT$ being the set of literals.  Let $\CT_0\subset \CT$ be the negated literals and $\CT_1=\CT\backslash \CT_0$ be the set of literals without negation.  Suppose $\CS$ is any subset of $[n]$. If $\max \CS > \max\CT$, then $\fS=0$; otherwise
\begin{align*}
\fS = \prod_{a_{j}\in \CT_0\cup \{0\}} A'_{(a_{j}+1, a_{j+1})}(0) \prod_{a_{k}\in \CT_1} A'_{(a_{k}+1, a_{k+1})}(1) \prod_{a_{l}\in \CS\backslash \CT} D'_{(a_{l}+1, a_{l+1})}, 
\end{align*}
where    $0<a_1<a_2 < \cdots$ are the  ordered elements of $\CS\cup\CT\cup\{0\}$. 
 
\end{lemma}
\begin{proof}
The main idea is to break the chain into the $\CT$ segments with boundaries at $\CT$ nodes. As we show the contribution of each segment is independent of the rest of the chain and we can analyze each one separately using Lemma \ref{lem:single T multiple S segments}. 

We proceed by adding an auxiliary node $X_0$ as a parent of $X_1$, i.e., the chain is, $X_0\rightarrow X_1 \rightarrow \cdots \rightarrow X_n$ but we assign $p(X_1|X_0)=p(X_1)$. 
This allows us to unify the presentation noting that conditioning on $X_0$ has no effect on the probability of other variables.
The Fourier coefficient is calculated as the following iterative expectation
\begin{align}\label{eq:fs exp}
\fS = \EE_{X_1}\qty[Z_1 \cdots \EE_{X_{n}}\qty[Z_{n} \Big | X_{n-1}] \cdots \Big | X_0],
\end{align}
where $Z_i$ depends on the $i$th node and is defined as in \eqref{eq:Z_i}. 

Let $\tmax=\max(\CT)$ and $\smax=\max(\CS)$ be the largest index in each of the two sets. If $\smax > \tmax$ then $\fS=0$. 
This holds,
because, due to the order of the nodes in the chain and the law of total expectation:
\begin{align*}
\fS = \EE_{X_1}\qty[Z_1 \cdots \EE_{X_{\smax}}\qty[Z_{\smax} \Big | X_{\smax-1}] \cdots \Big | X_0].
\end{align*}
Since the only term dependent on $X_{\smax}$ is $\Phi_{\smax}$, then the inner most expectation equals 
\begin{align*}
\EE_{X_{\smax}}\qty[Z_{\smax} \Big | X_{\smax-1}] = \EE_{X_{\smax}}\qty[\Phi_{\smax} \Big | X_{\smax-1}] =0,
\end{align*}
where we used Lemma \ref{lem:pS properties} implying that $\fS=0$. %

Suppose $\smax \leq \tmax$, and let $t_1 <\cdots, <t_d=\tmax$ be the ordered elements of $\CT$  with $d=|\CT|$.  
These nodes create $d-1$,  $\CT$-segments $[t_{i-1}, t_{i})$, a head segment $[1, t_1)$ and a tail segment $[\tmax, n]$.  
Let $t_0:=0$,
$\CS_i=\CS\cap [t_{i-1}, t_i)$ for any $i\in [d]$ 
and note that since $\CS\subseteq [n]$, we have $\CS_1\subseteq [1, t_1)$.
If $\CS_i\setminus \CT$ is not empty, define $h_i  := \min \CS_i \setminus \CT$; otherwise set $h_i = t_i$.
 For any $i\leq j$, let $\bfZ_i^j=\prod_{i\leq l\leq j}Z_l$. 
Then, $\fS$ breaks into expectations over each segment:
\begin{align*}
\fS = \EE_{[X_1, X_{t_1})}\qty[\bfZ_1^{t_1-1} \cdots \EE_{[X_{t_{d-1}}, X_{\tmax})}\qty[ \bfZ_{t_{d-1}}^{\tmax-1} \EE_{[X_{\tmax}, n]}\qty[ Z_{\tmax} \Big | X_{\tmax-1}]\Big | X_{t_{d-1}-1}] \cdots \Big | X_0] .
\end{align*}
Starting from the tail, we use an inductive argument to calculate $\fS$. We show that each segment  $[t_{i-1}, t_{i})$ inherits $A'_{(t_{i}, t_{i})}(X_{t_{i}-1})$ as the input from the previous inner segment and contributes a variable $A'_{(t_{i-1}, t_{i-1})}(X_{t_{i-1}-1})$ multiplied by a constant. As a result, $\fS$ is the product of all the constants.  To show this argument, we start from the innermost segment that contributes the following based on Lemma \ref{lem:exp of Z}:
\begin{align*}
\EE_{[X_{\tmax}, n]}\qty[ Z_{\tmax} \Big | X_{\tmax-1}] = A_{\tmax} = A'_{(\tmax, \tmax)}.
\end{align*}
By the induction assumption, suppose that the $i$th segment $[t_{i-1}, t_{i})$ inherits $A'_{(t_{i}, t_{i})}$ as the input from the previous inner segment.  Its contribution is given by the $i$th conditional expectation:
\begin{align*}
\EE_{[X_{t_{i-1}}, X_{t_{i}})}\qty[ \bfZ_{t_{i-1}}^{t_{i}-1} A'_{(t_{i}, t_{i})}(X_{t_i-1}) \Big | X_{t_{i-1}-1}].
\end{align*}
This expectation is calculated using Lemma \ref{lem:single T multiple S segments} independently of the outer segments. From Lemma \ref{lem:single T multiple S segments}, (with $\CT=\{t_{i-1}\}$), the contribution of this segment is
\begin{align*}
\D'_{(h_i+1, t_i)} \A'_{(t_{i-1}+1, h_i)}(y_{i-1}) A_{t_{i-1}}(X_{t_{i-1}-1}),
\end{align*}
where $y_{i-1} = \mathbf{1}_{(t_{i-1} \in \CT_1)}$. Notice that the notation above is consistent when $\CS_i\setminus \CT$ is empty and, as a result, $h_i=t_i$. To sum up, each  segment $[t_{i-1}, t_{i})$ produces constants $\D'$ and $A'$ followed by a variable $A'_{(t_{i-1}, t_{i-1})}$ fed to the next segment; hence the induction holds. The induction ends with the head segment $[1, t_1)$. If $t_1>1$, this segment  inherits $A'_{(t_{1}, t_{1})}$ and contributes 
\begin{align*}
\EE_{[X_1, X_{t_1})}\qty[\bfZ_1^{t_1-1} A_{t_{1}}\Big| X_0] = \D'_{(h_1+1, t_1)} \A'_{(1, h_1)}(X_0),
\end{align*}
where we used Lemma \ref{lem:single T multiple S segments} (with $\CT=\emptyset$).  Since, $X_0$ is an auxiliary node independent of other nodes,  $A_1(1)=A_1(0)$. Hence, we obtain that $\A'_{(1, h_1)}(X_0)= \A'_{(1, h_1)}(0)=\A'_{(1, h_1)}(1)$, which implies that we can replace $X_0$ with $0$ in the right-hand side, as if $0\in \CT_0$.  If $t_1=1$, then the output of the previous segment is $A'_{(1, 1)}(X_0) = A_1(X_0) = A_1(0)$. Since, the head segment in this case is empty, then $h_1=1$ and as a result the contribution is $\D'_{(2, 1)} \A'_{(1, 1)}(X_0) = A_1(X_0)=A_1(0)$. 

Lastly, multiplying the constants produced by the segments gives $\fS$:  %
\begin{align}\label{eq:fS_alternative}
\fS = \prod_{i=1}^d \D'_{(h_i+1, t_i)} \A'_{(t_{i-1}+1, h_i)}(y_{i-1}), 
\end{align}
with the convention that $y_0=0$.  The proof is complete by rewriting the above equation using the $a_i$'s as in the statement of the lemma. 
\end{proof}

The above lemma helps us derive an upper bound on each $|\fS|$ and finally obtain a bound on the spectral norm of conjunctions.   

\subsection{Upper bound on the  spectral norm of conjunctions under chain BNs}
We start by bounding the values of $A$ and $D$ terms.
In the following lemmas the terminology of Definition~\ref{def:A' D'}
that defines $a_i$, $A'$ and $D'$ is used implicitly.

\begin{lemma}[Bounding $A'$ terms]
\label{lem:Abound2}
Suppose the distribution $D$ is a $c$-bounded chain %
 for a constant $c\in(0,\frac{1}{2})$.  Then, 
 for all $j$, $\abs{A_j(0)}\leq 1-c$,  $\abs{A_j(1)}\leq 1-c$, 
 and for  $A'_{(r, a_i)}$ as in Definition \ref{def:A' D'} with any $a_i$ and  $r\in [a_{i-1}, a_{i}]$  we have that
 \begin{align}\label{lem:bound hierarchical terms}
  \abs{A'_{(r, a_i)}(0)}\leq 1-c ,& & \abs{A'_{(r, a_i)}(1)} \leq 1-c.
  \end{align} 
  \end{lemma}
\begin{proof} 
By definition $A_{i,x_{i-1}}$ is one of $\mu_{i,x_{i-1}}$, $1-\mu_{i,x_{i-1}}$, $\sigma_{i,x_{i-1}}$, or $-\sigma_{i,x_{i-1}}$. Note that $\mu_{i,x_{i-1}} =  \EE[X_i | X_{i-1}=x_{i-1}]$; implying that $\mu_{i,0}, \mu_{i,1} \in [c, 1-c]$ and the same holds for their complements. Furthermore,
\begin{align*}
\sigma_{i,x_{i-1}} := \sqrt{\var(X_i|X_{i-1}=x_{i-1})}.
\end{align*}
 The variance of a Bernoulli random variable with bias $p$ is $\sigma^2={p(1-p)}$. Therefore, the maximum value of  $\sigma$ is $\frac{1}{2}$ which occurs when $p=\frac{1}{2}$. Moreover, given that $p\in [c, 1-c]$, as  $\sigma^2$ is a quadratic function of $p$, one can easily check that the minimum value occurs when $p=c$ or $1-c$. Therefore, $\sqrt{c(1-c)} \leq \sigma \leq \frac{1}{2},$ when $p\in [c, 1-c]$. With this observation,  $\sigma_{i,0}, \sigma_{i,1} \in [\sqrt{c(1-c)}, \frac{1}{2}]$. Therefore, we have $\abs{A_i} \leq 1-c$, which proves the first claim of the lemma. 
As for the second part, from \eqref{eq:A'_r as cond exp}, $A'_{r, a_i} = \EE\qty[A_{a_{i}} | X_{r-1}]$, and since $\abs{A_{a_{i}}}\leq 1-c$,  the expectation is bounded by the same range. 
\end{proof}

\begin{lemma}[Bounding $D$ terms with weak conditions]
\label{lem:Dbound2}
Suppose   $D$ is a $c$-bounded   chain %
and  
$$\abs\Big{\EE[X_i|X_{i-1}=1]-\EE[X_i|X_{i-1}=0]} \leq \frac{1}{2}-c$$
 for a constant $c\in(0,\frac{1}{2})$.  Then, 
 $|D_i|\leq \alpha_1 := \frac{1}{2}-c$ which means $D$ is $\alpha_1$- difference bounded. 
 \end{lemma}
 \begin{proof}
 When  $i\notin \CS$, then $$|D_i| =|A_{i, 1} - A_{i, 0}| = \abs{\mu_{i,1} -\mu_{i,0}} \leq \frac{1}{2}-c$$ by the second assumption in the lemma.  
When, $i\in \CS$, then $A_i$ is a function of $\sigma_{i,x_{i-1}} $ and 
\begin{align*}
|D_i| = \abs{\sigma_{i,0}-\sigma_{i,1}} \leq \frac{1}{2} - \sqrt{c(1-c)} \leq \frac{1}{2} -c,
\end{align*}
where we used the fact that $\sigma_{i,0}, \sigma_{i,1} \in [\sqrt{c(1-c)}, \frac{1}{2}]$
and $c\leq \sqrt{c(1-c)}$ for $c\leq 0.5$.
\end{proof}

The previous lemma gives a bound on $D_i$ that depends on $c$ that will affect the complexity bound for our algorithm. To generalize the case of product distribution (where $D$ is zero) without such dependence we state a lemma with a stronger condition and conclusion where $D$ is bounded away from 0.5 by a constant.

\begin{lemma}[Bounding $D$ terms with stronger conditions]
\label{lem:Dbound3}
Suppose $D$ is a $c$-bounded chain for a constant $c\in(0,\frac{1}{2})$, and  
$$\abs{\EE[X_i|X_{i-1}=1]-\EE[X_i|X_{i-1}=0]} \leq 0.1.$$
 Then, 
 $|D_i|\leq \alpha_2 := 0.4$, implying that $D$ is $\alpha_2$-difference bounded. 
 \end{lemma}
 \begin{proof}
When  $i\notin \CS$, then as above $|D_i| =\abs{\mu_{i,1} -\mu_{i,0}} \leq 0.1.$
When, $i\in \CS$, then $A_i$ is a function of $\sigma_{i,x_{i-1}}$.
Let $D_\mu,D_\sigma$ denote the difference in $\mu,\sigma$ values.
Recalling that 
$\sigma(x)=\sqrt{x(1-x)}$ and calculating 
$\sigma'(x)=\frac{1-2x}{2\sqrt{x(1-x)}}$
and 
$\sigma''(x)= \frac{-4x(1-x)-0.5(1-2x)^2}{4x(1-x)\sqrt{x(1-x)}}<0$
we observe that $\sigma$ is concave and $\sigma'$ is decreasing as a function of $x\in(0,1/2)$. 
This implies that the largest possible segment in $\sigma$ values is obtained when $\mu_1=c$ and $\mu_2=c+D_\mu$.
Now if $c>0.1$ then $D_\sigma\leq \sigma(0.5)-\sigma(0.1)=0.2$.
On the other hand,  if $c\leq 0.1$ then $D_\sigma\leq \sigma(c+D_\mu)-\sigma(c)\leq \sigma(c+D_\mu)\leq \sigma(0.2)=0.4.$
\end{proof}

\noindent
Now, since by definition  $D'_{r, a_i} =  \prod_{r\leq \ell \leq a_{i}} D_\ell$ we can bound $D'$ terms:
\begin{lemma}[Bounding $D'$ terms]
\label{lem:Dbound4}
For any $\alpha$-difference bounded chain and any $r\in (a_{i-1}, a_{i}]$,  $|D'_{r, a_i}|\leq  \alpha^{a_{i}-r+1}$.
 \end{lemma}
\begin{theorem}\label{thm:L1 bound chain}
Suppose the BN distribution is a $c$-bounded chain %
 for a constant $c\in(0,\frac{1}{2})$, and $\abs{D_{\sigma, i}}=\abs{\sigma_{i,1}-\sigma_{i,0}}\leq D_\sigma,$ and $\abs{D_{\mu, i}}=\abs{\mu_{i,1}-\mu_{i,0}}\leq D_\mu$ with $D_\sigma + D_\mu <1$. Then, the spectral norm of any conjunction $f$ with $d$ literals is bounded by 
$\L1(f) \leq \qty(\frac{(2-D_\sigma-D_\mu)(1-c)}{1-D_\sigma - D_\mu} )^{d}$. %
\end{theorem}

\begin{proof}
From Lemma \ref{lem:Fourier exact}, and particularly \eqref{eq:fS_alternative} we have the expression for $\fS$. Recall the notation $0=t_0< t_1 <\cdots, <t_d=\tmax$ with $d=|\CT|$ denoting the ordered elements of $\CT\cup \{0\}$. %
Define $\CS_i=\CS\cap [t_{i-1}, t_i)$ for any $i\in [d]$. 
If $\CS_i\setminus \CT$ is not empty, let $h_i  := \min \CS_i\setminus \CT$; otherwise set $h_i=t_i$.
Then, using \eqref{eq:fS_alternative}, Lemma~\ref{lem:Abound2}
and by the theorem's assumption, 
for all $\CS \subseteq [\tmax]$
the Fourier coefficient is bounded as 
\begin{align}\label{eq:fS bound chain}
    |\fS| 
    = \prod_{i=1}^{d} \abs{\D'_{(h_i+1, t_i)}} \abs{\A'_{(t_{i-1}+1, h_i)}(y_{i-1})}
    \leq  \qty\big(1-c)^{d}  \prod_{i=1}^d \abs{\D'_{(h_i+1, t_i)}}.
\end{align}
If $\CS\nsubseteq [\tmax]$, then $\fS=0$. Therefore, summing over $|\fS|$ for all $\CS\subseteq [\tmax]$ gives the $L_1$ norm, that is bounded as 
 \begin{align}\label{eq:L1 bound 1}
    \sum_{\CS} \abs{\fS} &\leq  \qty\big(1-c)^{d} \sum_{\CS\subseteq [\tmax]} \prod_{i} \abs{\D'_{(h_i+1, t_i)}}.
 \end{align}

By definition, $\cup_i \CS_i$ covers the set  $\set{1, \cdots, \tmax-1}$ %
and the summation on the right hand side of \eqref{eq:L1 bound 1} is replaced by %
  \begin{align*}
    \qty\big(1-c)^{d} \sum_{\CS_1\subseteq [1,t_1)} \cdots \sum_{\CS_{d}\subseteq [t_{d-1}, t_d) }  \sum_{\CS_{d+1}\in \{t_d\}}\prod_{i} \abs{\D'_{(h_i+1, t_i)}},
   \end{align*}
   where we used the fact that $0\notin \CS$ for the first summation.     
We would like to proceed by interchanging the summations and the product. For that we need to show that the variable $\D'_{(h_i+1, t_i)}$ only  depends on the $i$th summation. However, this variable depends on the $i$th and $(i+1)$th summations. Because it is a function of  $D_{t_i}$ which itself depends on whether  $t_i\in \CS_{i+1}$.   To address this issue, we slightly change the summations
(note the switch from right open set to left open set):
\begin{align*}
    \qty\big(1-c)^{d} \sum_{\CS'_1\subseteq (0, t_1]}  \cdots \sum_{\CS'_{d}\subseteq (t_{d-1}, t_d] }\prod_{i} \abs{\D'_{(h_i+1, t_i)}}.
   \end{align*}
Now, we can interchange the product with the summation, because the $i$th term appearing in the product depends only on the $i$th summation. Note that $h_i$ only depends on $\CS'_i$. Because $\CS_i\setminus \CT = \CS'_i\setminus \CT$.  As a result, the above quantity equals the following 
\begin{align*}
\qty\big(1-c)^{d} \prod_{i=1}^{d} \qty\Big(\sum_{\CS'_i\subseteq (t_{i-1}, t_i]}   \abs{\D'_{(h_i+1, t_i)}}).
\end{align*}
By conditioning on the value of $h_i$, the $i$th summation equals 
\begin{align}\label{eq:L1 bound 2}
    \sum_{k = t_{i-1}+1}^{t_i} \sum_{\substack{\CS'_i: h_i =k}}\abs{\D'_{(k+1, t_i)}}.
\end{align}
By definition  $D'=1$ when $h_i= t_i$. This happens when $\CS'_i =\{t_i\}$ or is empty. Otherwise, if  $h_i=k<t_i$ we have 
\begin{align*}
    \abs{\D'_{(k+1, t_i)}} &= \prod_{j=k+1}^{t_i} \abs{D_{j}} =  \prod_{\substack{k < j \leq t_i\\ j \in \CS'_i}} D_{\sigma, j} \prod_{\substack{k < j \leq t_i\\ j \notin \CS'_i}} D_{\mu, j}\\
    & \leq D^{|\CS'_i|-1}_\sigma D_\mu^{t_i-k - |\CS_i|+1},
\end{align*}
where we get $|\CS'_i|-1$ in the first exponent because $k\in \CS'_i$ but $D'$ starts with $k+1$. Moreover, we used  Definition \ref{def:A} implying that   we get $D_{\sigma, j}$ when $j\in \CS$, and $D_{\mu, j}$ otherwise.  
Therefore, by separating $k=t_i$, and  introducing another set $\CA\subseteq (k, t_i]$ such that $\CS'_i = \CA \cup \{k\}$, the summation in \eqref{eq:L1 bound 2} is simplified as
\begin{align}
    \eqref{eq:L1 bound 2} 
    &\leq 2 + \sum_{k = t_{i-1}+1}^{t_i-1} \sum_{\substack{\CA \subseteq (k, t_i]}}D^{|\CA|}_\sigma D_\mu^{t_i-k - |\CA|}
    \label{eq:chainWoA}
    \\
    & =  2+\sum_{k = t_{i-1}+1}^{t_i-1} \sum_{r=0}^{t_i-k}{\binom{t_i-k}{r}}  D^{r}_\sigma D_\mu^{t_i-k-r} \nonumber \\
    &=  2+\sum_{k = t_{i-1}+1}^{t_i-1} (D_\sigma+D_\mu)^{t_i-k} \nonumber \\
    & = 2+\sum_{k'=1}^{t_i-t_{i-1}-1}(D_\sigma+D_\mu)^{k'} \nonumber \\
    & = 1+\sum_{k'=0}^{t_i-t_{i-1}-1}(D_\sigma+D_\mu)^{k'}, \nonumber
\end{align}
where the second equality holds by counting the number of subsets $\CA$ of size $r$, the  third equality follows from the binomial theorem, and the fourth by change of the variable $k$ to $k'$.
The last summation is a geometric sum with the base  $D_\sigma+D_\mu$ which is less than one by assumption.  Therefore,  by increasing the range of the summation to $\infty$, the following inequality is obtained:
\begin{align}\label{eq:L1 bound 4}
    \sum_{\CS'_i\subseteq (t_{i-1}, t_i]}   \abs{\D'_{(h_i+1, t_i)}}  \leq  1+\frac{1}{1-D_\sigma-D_\mu} = \frac{2-D_\sigma-D_\mu}{1-D_\sigma-D_\mu}.  %
\end{align}

Now combining \eqref{eq:L1 bound 1}-\eqref{eq:L1 bound 4} gives the following $L_1$ bound
\begin{align*}
\sum_{\CS} \abs{\fS} &\leq \qty\big(1-c)^{d} \prod_{i=1}^{d} \qty(\frac{2-D_\sigma-D_\mu
}{1-D_\sigma-D_\mu}) \\
&=  \qty(\frac{(2-D_\sigma-D_\mu)(1-c)}{1-D_\sigma-D_\mu} )^{d}.
\end{align*}
\end{proof}

We therefore have the following bound where (i) captures more cases but
(ii) is a strict generalization of the product case because the base of the polynomial is not a function of $c$.

\begin{corollary}
    \noindent (i) Under the conditions of Lemma~\ref{lem:Dbound2} (weak conditions) $D_\mu + D_\sigma \leq  1-2c$  and $\sum_{\CS} \abs{\fS} \leq (\frac{(1+2c)(1-c)}{2c})^{d}$.\\
    \noindent (ii)
    Under the conditions of Lemma~\ref{lem:Dbound3} (strong conditions) $\D_\mu +D_\sigma \leq 0.5$ and $\sum_{\CS} \abs{\fS} \leq (3(1-c))^{d}$.
\end{corollary}

\begin{remark}
The bound in Theorem~\ref{thm:L1 bound chain} can be made tighter by a slightly more refined analysis
that accounts for $A'$ terms explicitly instead of bounding them by $1-c$.
In particular, we can distribute the $A'$ terms to their individual segments and in 
\eqref{eq:chainWoA} we can account for the $A'$ term in each element. 
Recall that $A_k$ terms and top level elements in $A_k'$ terms are based on $\sigma$ (when $k\in S$) or $\mu$ (when $k\not \in S$);
cf.\ the terms 
 $(\sigma_{8,0}+D_{8,\sigma} \mu_{7,1})$
and $(\mu_{2,0}+D_{2,\mu} \mu_1)$
in the example above.
Referring to the two cases when $S\setminus T$ is empty, 
when $S=\{\}$ we have $A'_{(t_{i-1}+1,t_i),\mu}$ and when  $S=\{t_i\}$ we have $A'_{(t_{i-1}+1,t_i),\sigma}$ where we have added the annotation for $\mu$ or $\sigma$ to distinguish these cases. Similarly when $S\setminus T$ is not empty and its least index is $k$ we get
$A'_{(t_{i-1}+1,k),\sigma}$. 
Then, proceeding as in  \eqref{eq:L1 bound 2} and \eqref{eq:chainWoA} we obtain an exact computation of the spectral norm and a corresponding bound:
\begin{align}
\label{eq:spectrumExact}
\sum_{\CS} \abs{\fS} 
& = 
\prod_{i=1}^{d} 
\qty\Big(
A'_{(t_{i-1}+1,t_i),\mu} + |A'_{(t_{i-1}+1,t_i),\sigma}|+ \sum_{k = t_{i-1}+1}^{t_i-1} |A'_{(t_{i-1}+1,k),\sigma}| 
\sum_{\substack{\CS'_i: h_i =k}}\abs{\D'_{(k+1, t_i)}}
)
\\
& \leq
\prod_{i=1}^{d} 
\qty\Big(
A'_{(t_{i-1}+1,t_i),\mu} + |A'_{(t_{i-1}+1,t_i),\sigma}|+ \sum_{k = t_{i-1}+1}^{t_i-1} |A'_{(t_{i-1}+1,k),\sigma}| \sum_{\substack{\CA \subseteq (k, t_i]}}D^{|\CA|}_\sigma D_\mu^{t_i-k - |\CA|}
).
\nonumber
\end{align}
Bounding $A'$ terms as $A'_\mu\leq 1-c<1$ and   $|A'_\sigma|\leq 0.5$ yields a final bound of $\qty(\frac{(1.5-D_\sigma-D_\mu)}{1-D_\sigma-D_\mu} )^{d}$.
\end{remark}
The analysis of \eqref{eq:spectrumExact} is especially useful for the case of product distributions. In this case we can view the distribution as a chain (with an arbitrary ordering) where for all $i$ we have $\mu_{i,0}=\mu_{i,1}=\mu_i$.
In this case all the differences are zero, $D_\mu=D_\sigma=0$, and hence the internal sum in \eqref{eq:spectrumExact} is zero. Moreover, due to the same reason, 
the $A'$ terms simplify to their leading terms 
(e.g., $(\sigma_{8,0}+D_{8,\sigma} \mu_{7,1})$
simplifies to  $\sigma_{8,0}$)
that is 
$A'_{\mu,(t_{i-1}+1,t_i)}$ is $\mu_{t_i}$ or $1-\mu_{t_i}$ and $A'_{\sigma,(t_{i-1}+1,k)}=\sigma_{t_i}$. 
This gives an exact value for the spectral norm of product distributions. Noting that $\sigma_i=\sqrt{\mu_i(1-\mu_i)}$ and 
using the fact that $x+\sqrt{x(1-x)}\leq 1.21$ for any $x\in [0,1]$, we obtain 
a slightly tighter bound than the bound $(\sqrt{2})^d$ stated by \cite{Feldman2012}.
\begin{proposition}
\label{prop:product}
Let $f$ be a conjunction of $d$ literals 
with positive literals 
$\CT_1=\{i_1,i_2,\ldots,i_{d_1}\}$ 
and negative literals 
$\CT_0=\{j_1,j_2,\ldots,j_{d_2}\}$ where $d=d_1+d_2$.
Then, for any product distribution,
the spectral norm of $f$ is given by  
\begin{align}
\sum_{\CS} \abs{\fS} & 
= 
\qty(\prod_{i\in \CT_1} (\mu_i+\sigma_i))\qty(\prod_{j\in \CT_0} ((1-\mu_j)+\sigma_j)) \leq 1.21^d. 
\end{align}
\end{proposition}

\section{Fourier Expansion of Conjunctions for Tree BNs}\label{sec:tree}
In this section, we extend our results to Tree BNs. 
To simplify the presentation we use the generic condition $D_i\leq \alpha$ instead of the more detailed analysis using $D_\mu, D_\sigma$ that was used for chains. That is we assume $\alpha$-difference bounded tree BNs. 
We start with an additional development for chains that facilitates the analysis for trees.

Given the pair $\CS, \CT\subseteq[n]$, recall the notation $0=t_0< t_1 <\cdots, <t_d=\tmax$ with $d=|\CT|$ denoting the ordered elements of $\CT\cup \{0\}$ and $\CS_i=\CS\cap [t_{i-1}, t_i)$ for any $i\in [d]$. If $\CS_i\setminus \CT$ is not empty, let $h_i  := \min \CS_i\setminus \CT$; otherwise $h_i =t_i$.

A function $f:\set{0,1}\rightarrow \RR$ is said to be \textit{bounded single-sided} if  $\max_x |f(x)|\leq 1$ and $f(0)$ and $f(1)$ have the same sign. A trivial example is the constant function $f(x)=1$.
\begin{lemma}[Generic branch]\label{lem:genericbranch}
  Suppose that $X_0\rightarrow \cdots \rightarrow X_n$ form a $\alpha$-difference bounded chain and $f$ is a bounded single-sided function.   Then for any $\CS, \CT\subseteq [n]$ and the corresponding random variables $Z_i, i\in [n]$, we have that 
  \begin{align*}
    \EE_{X_1,\cdots, X_n}\qty[\prod_{i=1}^n Z_{i} ~~ f(X_n) \Big | X_{0}] &=  b'_{\CS, \CT}~ g(X_{0}),
  \end{align*}
  for some bounded single-sided function $g$ and a constant  $ b'_{\CS, \CT}$ bounded as 
  \begin{align*}
    |b'_{\CS, \CT}| \leq \begin{cases}
      |f(1)-f(0)| \alpha^{n-h_{d+1}} \prod_{i=1}^d \alpha^{t_i-h_i} & ~\text{if}~ \max \CS>\max \CT\\
      \prod_{i=1}^d \alpha^{t_i-h_i} & ~\text{otherwise}
    \end{cases}  
  \end{align*}
  where $h_{d+1}  := \min \CS\cap (\tmax, n]$ if $\max \CS > \max \CT$; otherwise $h_{d+1}=n+1$. 
\end{lemma}
\begin{proof}
  The proof is similar to Lemma \ref{lem:Fourier exact}. Note that $f(X_n)$ behaves as an $A'$ term because it is bounded and single-sided. Therefore, we can imagine an auxiliary node $n+1$ with  $A'_{(n+1,n+1)}=A_{n+1} = f(X_n)$ and $D_{n+1}:= f(1)-f(0)$. %
In the following,
let $t_{d+1}:=n+1$ and extend the notation of $\CS_i$ and $h_i$ to index $d+1$.
  With this notation we have one head segment $[1,t_1)$, $d-1$ segments $[t_{i-1}, t_i)$, and a tail segment $[t_d, n+1)$.

 We repeatedly use Lemma \ref{lem:single T multiple S segments} to find the contribution of each segment. 
  Starting from the tail segment, the contribution is given by the lemma's case (b) or (f):
  \begin{align*}
    \EE_{[X_{t_d},   X_n]}\qty[\bfZ_{t_d}^n f(X_n) \Big | X_{t_d-1} ] = D'_{(h_{d+1}+1, n+1)}A'_{(t_d+1, h_{d+1})}(y_{d}) A_{t_d}(X_{t_d-1}),
  \end{align*}
  where $y_{d}=\mathbf{1}_{(t_d \in \CT_1)}$.
  Note that the notation above covers the case where $\CS_{d+1}$ is empty, for which $h_{d+1}=n+1$ and by convention $D'$ term equals 1. 
  By an inductive argument similar to the proof of Lemma \ref{lem:Fourier exact}, the contribution of any intermediate segment $[t_{i-1}, t_i)$ is given by Lemma \ref{lem:single T multiple S segments} case (b) or (f):
  \begin{align*}
  \D'_{(h_i+1, t_i)}\A'_{(t_{i-1}+1, h_i)}(y_{i-1})  A_{t_{i-1}}(X_{t_{i-1}-1}),
  \end{align*}
  where $y_{i-1} = \mathbf{1}_{(t_{i-1} \in \CT_1)}$. 

  As for the head segment, we use Lemma \ref{lem:single T multiple S segments} case (a) or (e): %
  $$\EE_{[X_1, X_{t_1})}\qty[\bfZ_1^{t_1-1} A_{t_{1}}\Big| X_0] = \D'_{(h_1+1, t_1)} \A'_{(1, h_1)}(X_0).$$ 
  Note that the above expression covers the case where the head segment is empty ($t_1=1$), because $h_1=1$ and the $D'$ term is 1. 
  
  Therefore, the total contribution is given by multiplying the contributions of all segments: 
  \begin{align*}
     \A'_{(1, h_1)}(X_0) \D'_{(h_1+1, t_1)}\prod_{i=2}^{d+1} \D'_{(h_i+1, t_i)} \A'_{(t_{i-1}+1, h_i)}(y_{i-1}). 
  \end{align*}
Let $g(X_0)  =  \A'_{(1, h_1)}(X_{0})$. We need to show that $g$ is a bounded and single-sided function. If $h_1<n+1$, then $g$ is an $A'$ term independent of $f$. In that case, from \eqref{eq:A'_r as cond exp}, $g(x)= \EE[A_{h_1}|X_0=x]$ for any $x\in \{0,1\}$. Therefore, $g$ is bounded as $A_{h_1}$ is bounded from Lemma \ref{lem:Abound2}. Moreover, $g$ is single-sided as $A_{h_1}(0), A_{h_1}(1)$ have the same sign depending whether $h_1 \in \CS\cap \CT_0$ or not (see Definition \ref{def:A}). If $h_1=n+1$, we conclude that $d=0$ and $\CT, \CS$ are empty sets. From \eqref{eq:A'_r as cond exp}
\begin{align*}
  g(x) = \EE[A_{n+1}(X_n)|X_0=x] = \EE[f(X_n)|X_0=x]. 
\end{align*}
In that case, $g$ is bounded and single sided because of $f$.

  It remains to prove the upper bound on $ b'_{\CS, \CT}$ which is defined as 
  \begin{align*}
    b'_{\CS, \CT} = \D'_{(h_1+1, t_1)} \prod_{i=2}^{d+1} \D'_{(h_i+1, t_i)} \A'_{(t_{i-1}+1, h_i)}(y_{i-1}). 
  \end{align*}
  The absolute value of each $A'$ term above is bounded by $1$. This is because of \eqref{eq:A'_r as cond exp} and the facts that $A(0)$, $A(1)$ and $f$ are bounded.  

  For $1\leq i<d+1$, from Lemma \ref{lem:Dbound4} the corresponding  $D'$ term is bounded by $\alpha^{t_i-h_i}$. 
  For $i=d+1$, when, $h_{d+1}=n+1$,  $\D'_{(h_{d+1}+1, n+1)}=\D'_{(n+2, n+1)}=1$
  and  
 $\abs{b'_{\CS, \CT}}\leq \prod_{i=1}^{d} \alpha^{t_i-h_i}$.
 When $h_{d+1}<n+1$, $\max \CS > \max \CT$ and the corresponding $D'$ term is 
  $$\D'_{(h_{d+1}+1, n+1)}=(f(1)-f(0))\D'_{(h_{d+1}+1, n)}.$$
  The $D'$ term is a standard one, implying that its absolute value is bounded by $\alpha^{n-h_{d+1}}$ 
  and
  \begin{align*}
    \abs{b'_{\CS, \CT}}\leq \abs{f(1)-f(0)}\alpha^{(n-h_{d+1})} \prod_{i=1}^{d} \alpha^{t_i-h_i}.
  \end{align*}
\end{proof}

\begin{lemma}[L1 bound for a generic branch]\label{lem:L1genericbranch}
  In the setup of Lemma \ref{lem:genericbranch}, the L1 norm of $b'_{\CS, \CT}$ over all subsets   $\CS\subseteq [n]$ is bounded as
  \begin{align*}
    \sum_\CS \abs{b'_{\CS, \CT}} \leq \qty(1+\frac{\abs{f(1)-f(0)}}{1-2\alpha}) \qty(\frac{2-2\alpha}{1-2\alpha})^{|\CT|}.
  \end{align*}
  \end{lemma}
\begin{proof}
  Let $b_{\CS_{d+1}} = \abs{f(1)-f(0)}\alpha^{n-h_{d+1}}$ if  $\max \CS > \max \CT$, and $b_{\CS_{d+1}}=1$ otherwise. Using the upper bound in Lemma \ref{lem:genericbranch}  and by decomposing the L1 summation over subsets on each segment we have 
  \begin{align*}
    \sum_\CS |b'_{\CS, \CT}| \leq       \sum_{i=1}^{d+1} \sum_{\CS_i \subseteq \CS\cap [t_{i-1}, t_i)}  b_{\CS_{d+1}} \prod_{i=1}^{d} \alpha^{t_i-h_i}.
  \end{align*}
  We next change the summation ranges as follows
  \begin{align*}
    \sum_{i=1}^{d+1} \sum_{\CS'_i \subseteq (t_{i-1}, t_i]}  b_{\CS_{d+1}} \prod_{i=1}^{d} \alpha^{t_i-h_i}.
  \end{align*}
  Note that $\CS_i'\setminus \CT = \CS_i\setminus \CT$. Therefore, $h_i$ only depends on the $i$th summation as it is a function of $\CS'_i \setminus \CT$. Therefore, the summations and the product are exchangeable because each summation on $\CS_i$ runs over independent terms. The L1 norm is upper bounded by the following
  \begin{align*}
    \qty\Big(\sum_{\CS'_{d+1} \subseteq (t_{d}, n]} b_{\CS_{d+1}})\prod_{i=1}^{d} \qty\Big(\sum_{\CS'_i \subseteq (t_{i-1}, t_i]} \alpha^{t_i-h_i}),  
  \end{align*}
 When $h_i=t_i$, then $\CS'_i=\{t_{i}\}$ or is empty. Therefore, by separating this case, and conditioning on the value of $h_i$, the $i$th summation equals
  \begin{align*}
    \sum_{\CS'_i \subseteq (t_{i-1}, t_i]} \alpha^{t_i-h_i} & = 2+ \sum_{k = t_{i-1}+1}^{t_i-1} \sum_{\substack{\CA \subseteq (k, t_i]}}\alpha^{t_i-k}\\
    &\stackrel{(a)}{=} 2+\sum_{k = t_{i-1}+1}^{t_i-1} (2\alpha)^{t_i-k}\\
    & = 2+\sum_{k'=1}^{t_i-t_{i-1}-1}(2\alpha)^{k'}\\
    & \stackrel{(b)}{\leq} 2+\frac{2\alpha}{1-2\alpha}=\frac{2-2\alpha}{1-2\alpha},
  \end{align*}
where (a) holds by counting the number of subsets of $(k, t_i]$ and (b) follows by increasing the upper range of the geometric sum to $\infty$ and the assumption $2\alpha<1$.  

As for the summation over  $\CS'_{d+1}$, if $t_d=n$, then it is equal 1. Otherwise, if $t_d<n$, separating the two cases of whether $\CS'_{d+1}$ is empty or not leads to the following
\begin{align*}
  1 + \abs{f(1)-f(0)} \sum_{\substack{\CS'_{d+1}\subseteq (t_d, n]\\ \CS_{d+1}\setminus \CT\neq \emptyset}} \alpha^{n-h_{d+1}}.
\end{align*}
Note that $h_{d+1}\leq n$ when  $\CS'_{d+1} \neq \emptyset$. %
By conditioning on the value of $h_{d+1}=k$,
the summation equals 
\begin{align}\nonumber
  \sum_{k = t_{d}+1}^{n} \sum_{\substack{\CA \subseteq (k, n]}}\alpha^{n-k}  = \sum_{k = t_{d}+1}^{n}(2\alpha)^{n-k} =  \sum_{k'=0}^{n-t_{d}-1} \qty\big(2\alpha)^{k'} \leq \frac{1}{1-2\alpha},
\end{align}
where we used the fact that the last summation is a geometric sum with the base  $2\alpha\leq 1$. 

Putting all the arguments together, the L1 bound is 
\begin{align*}
  \sum_\CS \abs{b'_{\CS, \CT}}\leq  \qty(1+\frac{\abs{f(1)-f(0)}}{1-2\alpha}) \qty(\frac{2-2\alpha}{1-2\alpha})^d. 
\end{align*}
\end{proof}

\begin{figure}
\centering 
\includegraphics[scale=1]{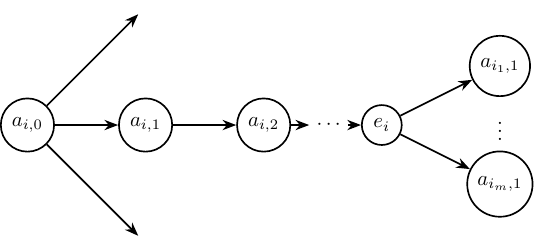}
\caption{A generic branch inside a tree BN. The branch is the chain of the nodes $a_{i,1},\cdots, e_i$.}
\label{fig:tree}
\end{figure}

\begin{theorem}\label{thm:tree2}
Under any $\alpha$-difference bounded tree BN,  the spectral norm of  any conjunction $f$ with $d$ literals is bounded by 
$\L1(f) \leq \qty(\frac{2-2\alpha}{1-2\alpha})^{2d}$.
\end{theorem}

\begin{proof}
The proof builds upon the argument for chains by viewing the  BN tree as a collection of connected branches (chains). More precisely, a branch is a set of nodes that form a chain starting from an expansion node and ending with a leaf or another expansion node. Figure \ref{fig:tree} shows a generic branch with nodes $a_{i,1}, a_{i,2}, \cdots, e_i$ in a generic tree. Notice that the parent of the branch $a_{i,0}$ is not counted toward the branch. To avoid double counting of nodes, we exclude the top expansion node from each branch. 

Let $X_1$ be the root node in the tree. We treat $X_1$ as an expansion point even if it is connected to only one branch. For that we add an auxiliary node $X_0$ that is independent of every other node and create $\ClC_0$  the auxiliary root branch $X_0\rightarrow X_1$. We will use this branch for presentation convenience.  Let $\ClC_0, \cdots, \ClC_k$ be all the branches of the tree that are topologically ordered with $\ClC_0$ at the root.    Let $$Z_{\ClC_i}:= \prod_{j\in \ClC_i}Z_j$$ be the product of the $Z$-variables belonging to the $i$th branch.  Also let  $Z_{\ClC_0}=Z_1$. Then, the Fourier coefficient $\fS$ is written as the iterative expectations over each branch: 
\begin{align*}
\fS = \EE_{\sim \ClC_0}\qty[Z_1 \EE_{\sim \ClC_1}\qty[ Z_{\ClC_1} \cdots \EE_{\sim \ClC_k}\qty[ Z_{\ClC_k}  \Big | X_{\pa(\ClC_{k})}] \cdots \Big | X_1] ],
\end{align*}
where the subscript $\sim \ClC_i$ is to emphasize that  the expectation is taken over all $X$-variables in the branch.

We proceed by induction in a reverse topological order starting from leaf branches to the root branch $\ClC_0$,
showing that each branch 
contributes a constant $b'_{\ClC_i}$ and outputs an $A'$-type term to the parent branch.

Any leaf branch $\ClC_i$ starts from an expansion point $a_{i, 0}$ (or the root in the trivial case) and ends with a leaf. Conditioned on its head node, it is independent of the other branches. We use Lemma \ref{lem:genericbranch} with $\CS \leftarrow \CS\cap \ClC_i$ and $\CT \leftarrow \CT\cap \ClC_i$ and the constant function $f(x)=1$ for $x=0,1$.
If $\max \CS\cap \ClC_i > \CT\cap \ClC_i$,  the contribution is zero and consequently $\fS=0$. 
Therefore, $\fS=0$ when there is a $\CS$ node after all $\CT$ nodes in a leaf branch. 
We can therefore ignore such sets $\CS$ in the calculation of the L1 spectrum.
In addition, if a leaf branch does not include any $\CT$ node then it cannot include any $\CS$ node, because otherwise $\fS=0$. Moreover, if such branches do not have any $\CS$ or $\CT$ nodes, then the expectation is 1. 
We can therefore ignore such branches in $G$ in the computation of the L1 spectrum.
Assuming these do not occur,
 from Lemma \ref{lem:genericbranch}, the contribution of the  leaf branch $\ClC_i$ is 
 \begin{align}\label{eq:branch contribution 1}
 \EE_{\sim \ClC_i}\qty[Z_{\ClC_i} \Big | X_{\pa(\ClC_{i})}] &= b'_{\ClC_i} g_k(X_{\pa(\ClC_{i})}),
 \end{align}
where $g_i$ is a bounded single-sided function and $b'_{\ClC_i}$ a constant bounded as in the lemma. 
Now, since $g_i$ is  bounded single-sided then it behaves as an $A'$ term, because we can write 
\[
  g_i(x) = g_i(0) + x (g_i(1)-g_i(0)),\qquad x=0,1.
\]
Therefore,  a leaf branch contributes a constant $b'_{\ClC_i}$ and outputs an $A'$-type term to the parent branch.
For the inductive step, assume the claim holds for all descendants of  branch $\ClC_i$. Because the child-branches are independent of each other conditioned on the parent, their $A'$ terms will be multiplied to create the input to the parent branch.  
Let $\ClC_{i_1}, \cdots, \ClC_{i_{m_i}}$ be the child branches of $\ClC_i$, where $m_i$ is the number of child branches.
The contribution of each child-branch $\ClC_{i_ j}$ is  $b'_{\ClC_{i_j}} g_{i_j}(X_{\pa(\ClC_{i_j})})$. %
The constants $b'$ can be taken out of the expectation. Each $g_{i_j}$ is a function of the last element in $\ClC_i$, which is an expansion node denoted by $e_i$.  
Ignoring the $b$'s from the child branches, the contribution of $\ClC_i$ is calculated as
\begin{align}\label{eq:branch C_i}
   \EE_{\sim \ClC_i}\qty\bigg[Z_{\ClC_i}   \prod_j  g_{i_j}(X_{e_i})      \Big | X_{\pa(\ClC_{i})}]. 
\end{align}
 We need the following lemma.
\begin{lemma}\label{lem:product bounded singlesided}
  A finite product $\prod_i f_i$ of bounded single-sided functions  is bounded single-sided.
\end{lemma}
\begin{proof}
  It suffices to prove the statement for the product of two bounded single-sided functions, say $f_1 f_2$. The lemma follows from an inductive argument over the number of terms in the product.  Let $h(x)=f_1(x)f_2(x)$ for $x=0,1$. Clearly $\abs{h(x)}=\abs{f_1(x)}\abs{f_2(x)}\leq 1$. Moreover, $h(0)h(1) = f_1(0)f_1(1) \times f_2(0)f_2(1)$. This quantity is non-negative as $f_1$ and $f_2$ are single-sided. Hence, $h$ is bounded single-sided. 
\end{proof}

Therefore, %
the branch $\ClC_i$ takes as input a single-sided function $f = \prod_jg_{i_j}$ which is an $A'$-type term. Then, using Lemma \ref{lem:genericbranch}, this branch produces   a bounded single-sided function $g_i(X_{\pa(\ClC_{i})})$ multiplied by a constant $b'_{\ClC_i}$. With that the induction is established meaning that any branch $\ClC_i$ contributes the constant $b'_{\ClC_i}$ given as in Lemma \ref{lem:genericbranch}. %
 Lastly, the contribution of the auxiliary root branch $X_0\rightarrow X_1$ is $b'_{\ClC_0}g_0(X_0)$ for some bounded single-sided function $g_0$.
Since $X_0$ is independent of all other nodes, then we can replace $g_0(X_0)$ with $g_0(0)$. 

As a result, the expression for $\fS$ is calculated as
\begin{align}\label{eq:final coef tree}
\fS &=\prod_{i=0}^k b'_{\ClC_i}~g_{0}(0).
\end{align}
Next, we establish the L1 bound using  the analysis of Lemma \ref{lem:L1genericbranch}. Since, $\abs{g_0(0)}\leq 1$, we have 
\begin{align*}
\sum_{\CS} \abs{\fS} &\leq  \sum_{\CS} \prod_{i=0}^k \abs{b'_{\ClC_i}}\\
&=  \sum_{\CS\cap \ClC_0} \cdots \sum_{\CS\cap \ClC_k} \prod_{i=0}^k \abs{b'_{\ClC_i}}\\
&\stackrel{(a)}{=}\prod_{i=0}^k \qty\Big(\sum_{\CS\cap \ClC_i} \abs{b'_{\ClC_i}}),
\end{align*}
where for (a) we interchanged the summations with the product as they depend on disjoint sets of variables.
Next, we analyze each summation in the above equation. For each leaf branch,  we use  Lemma \ref{lem:L1genericbranch} with  $f=1$ being the constant function and derive the following L1 bound 
\begin{align*}
  \sum_{\CS\cap \ClC_i} \abs{b'_{\ClC_i}} \leq  \qty(\frac{2-2\alpha}{1-2\alpha})^{|\CT\cap \ClC_i|}.
\end{align*}
For any other branch,  from Lemma \ref{lem:L1genericbranch} each summation is bounded by  
\begin{align*}
  \sum_{\CS\cap \ClC_i} \abs{b'_{\ClC_i}} \leq    \qty(\frac{2-2\alpha}{1-2\alpha})^{|\CT\cap \ClC_i|+1},
\end{align*}
where we used the fact that $1+\frac{\abs{f(1)-f(0)}}{1-2\alpha} \leq \frac{2-2\alpha}{1-2\alpha}$. This is  because $|f(1)-f(0)|\leq 1$ as  $f$ is bounded and single-sided.
Note that these branches contribute an additional $\frac{2-2\alpha}{1-2\alpha}$ compared to the leaf branches. 
Combining these bounds we get the following L1 bound for the tree
 \begin{align*}
  \sum_{\CS} \abs{\fS} &\leq \prod_{i:~ \text{leaf}~ \ClC_i} \qty(\frac{2-2\alpha}{1-2\alpha})^{|\CT\cap \ClC_i|} \prod_{i: ~\text{non leaf}~ \ClC_i} \qty(\frac{2-2\alpha}{1-2\alpha})^{|\CT\cap \ClC_i|+1} \\
  & = \qty(\frac{2-2\alpha}{1-2\alpha})^{|\CT|+E},
 \end{align*}
where $E$ is the number of non-leaf branches. %

Finally, recall that we can ignore leaf branches without any $\CT$ nodes.
Hence the number of leaf branches is $\leq d$ and since we have a tree the number of non-leaf branches (excluding $\ClC_0$) is at most $d-1$
and $E\leq d$.
As a result, the final expression for the L1 bound is 
 \begin{align*}
  \sum_{\CS} \abs{\fS}\leq \qty(\frac{2-2\alpha}{1-2\alpha})^{2d}.
 \end{align*}
\end{proof}

The result of  Theorem \ref{thm:tree2} is readily extendable to forests. 
\begin{corollary}
  Under any $\alpha$-difference bounded forest BN, the spectral norm of any conjunction $f$ with $d$ literals is bounded by 
$\L1(f) \leq \qty(\frac{2-2\alpha}{1-2\alpha})^{2d}$.
\end{corollary}
\begin{proof}
We can convert the forest to a tree by adding an auxiliary node $X_0$, independent of every other node, and connecting it to all the roots $r_j$ in the forest.  
We can apply the analysis in Theorem \ref{thm:tree2} on this tree and the only difference is that, in the final form of the coefficient in \eqref{eq:final coef tree}, 
$g_{0}(0)$ is replaced with $\tilde{g}_0(0)=\prod_j g_{r_1}(0)$ which is also a bounded single sided function.
Therefore the same bound applies to the forest.
\end{proof}

\section{Lower Bounds on the Spectral Norm}\label{sec:lb}

Our upper bounds for chains and trees require $D_\mu+D_\sigma<1$. In this section we show that this is necessary. 
Without this condition, even for chains the spectral norm of a single literal can be exponentially large.
In addition, for general graphs, the L1 can be exponentially large even when $D_\mu+D_\sigma<1$.

\subsection{Lower Bounds for Chains}

 \begin{lemma}[Lower bound for chains]
Consider a chain BN with $n+2$ variables $X_0,\ldots,X_{n+1}$ and the function $f=X_{n+1}$.
Then there exist distributions instantiating this structure where all nodes share the same conditional distribution table,
with
$D_\mu=|\mu_{i,1}- \mu_{i,0}|>0$, $D_\sigma=|\sigma_{i,1}- \sigma_{i,0}|>0$ 
and $\L1(f) =\Omega((D_\mu + D_\sigma)^n)$.
\end{lemma}
\begin{proof}
By Lemma~\ref{lem:Fourier exact} for each $S$, $\fS$ is a product of one $A'$ term and a number of $D$ terms. 
Let ${\cal S}_0$ be the set of subsets of the variables that include $X_0$ and exclude $X_{n+1}$.
Recall that when $X_0$ is in $S$ all nodes above index zero contribute a $D_i$ term to $\fS$ and the $A$ term is the one associated with node 0.
We have that $A'_0=\sigma_{0}$ because $X_0$ is in $S$ and $|D_{n+1}|=D_\mu$ because $X_{n+1}$ is excluded from $S$. 
For $1\leq i \leq n$,
the value of $|D_i|$ is $D_\sigma$ when $i\in S$ and $D_\mu$ otherwise. 
We therefore have:
\begin{align*}
\L1(f) %
 & > \sum_{S \in {\cal S}_0} |\fS| \\
 & = \sigma_{0} D_\mu \sum_{k=1}^n {\binom{n}{k}} D_\mu^k D_\sigma^{n-k} \\
 & =  \sigma_{0} D_\mu (D_\mu + D_\sigma)^n.
\end{align*}
\end{proof}

As illustrated in the next corollary, this yields an exponential lower bound for a range of probability models
including ones where $D_\mu$ is arbitrarily close to 0.5 (although the latter requires a small value of $c$):

\begin{corollary}
\label{cor:LBchain}
(1)
Choosing 
$\mu_1=c$ , $\mu_2=0.5+c+\alpha$, where $c=0.00001$ and $\alpha=0.353$ (with $D_\mu=0.853$) yields $D_\mu+D_\sigma>1.2$
and $\L1(f) =\Omega(1.2^n)$.
\\
(2)
For any $c\leq 0.0246$, choosing
$\mu_1=c$ , $\mu_2=0.5+c+\alpha$ where $\alpha=2\sqrt{c}$ (with $D_\mu=0.5+2\sqrt{c}$) yields $D_\mu + D_\sigma>1+c$
and $\L1(f) =\Omega((1+c)^n)$.
\end{corollary}
\begin{proof}
For (1) we optimize 
$
D_\mu+D_\sigma=(0.5+\alpha) + (\sqrt{0.25-(\alpha+c)^2}-\sqrt{c(1-c))} 
$
w.r.t.\ $\alpha$ to obtain $(\alpha+c)=\sqrt{0.125}$.
Picking values to roughly match this, that is, $c=0.00001$ and $\alpha=0.353$ yields
$D_\mu+D_\sigma>1.203$.

For (2)
$
D_\mu+D_\sigma=(0.5+\alpha) + (\sqrt{0.25-(\alpha+c)^2}-\sqrt{c(1-c))} 
\geq 
(0.5+\alpha) + (\sqrt{0.25-(\alpha+c)^2}-\sqrt{c}).
$
Using $\alpha=2\sqrt{c}$ we have that
$
D_\mu+D_\sigma \geq 
0.5+\sqrt{c} + \sqrt{0.25-(2\sqrt{c}+c)^2}$.
This quantity is greater than $1+c$ when
$0.25-(2\sqrt{c}+c)^2 > (0.5+c-\sqrt{c})^2$ which is equivalent to 
$2c^2+2c\sqrt{c}+6c-\sqrt{c}<0$.
Writing $a=\sqrt{c}$ and noting that $c>0$ this holds when 
$2a^3+2a^2+6a-1<0$ which is true for $a\leq 0.157$ or $c\leq 0.0246$.
\end{proof}

\subsection{Lower Bounds for General Graphs}

We first provide a second lower bound for chains which is useful in constructing a lower bound for general graphs even when 
$D_\mu+D_\sigma <1$. Using the following observation, the lower bound is obtained on the sum of coefficients (without absolute values)
that we can calculate exactly.

\begin{observation}
$\L1(f) = \sum_S |\coeff{S}|\geq | \sum_S \coeff{S}|$. 
\end{observation}

 \begin{lemma}[Exact sum of coefficients for a chain with a single $T$ segment]
 \label{lm:exactSumChain}
Consider a chain BN with $n$ variables $X_1,\ldots,X_{n}$ and the function $f=X_{n}$
and let 
$D_{k,\mu}=\mu_{k,1}- \mu_{k,0}$, $D_{k,\sigma}=\sigma_{k,1}- \sigma_{k,0}$
(note that $D_{k,\mu}$ and $D_{k,\sigma}$ may be positive or negative).
Then 
\begin{align}
\sum_S \coeff{S} = \sum_{k=1}^{n}(\mu_{k,0}+\sigma_{k,0})\prod_{\ell=k+1}^{n}(D_{\ell,\mu}+D_{\ell,\sigma}).
\end{align}
\end{lemma}

\begin{corollary}[Lower bound for chains]
\label{cor:exactSumChain}
Under the conditions of the lemma and 
where all nodes share the same conditional distribution, i.e., 
$D_{k,\mu}=D_\mu$ and $D_{k,\sigma}=D_\sigma$, and $D_{\mu}+D_{\sigma}>0$
\begin{align}
\L1(f) \geq
(\mu_{0}+\sigma_{0})
\sum_{k=0}^{n-1} (D_{\mu}+D_{\sigma})^k
=
(\mu_{0}+\sigma_{0})\frac{1-(D_{\mu}+D_{\sigma})^n}{1-(D_{\mu}+D_{\sigma})}.
\end{align}
\end{corollary}

This expression yields 
an exponential lower bound $\Omega((D_{\mu}+D_{\sigma})^{n})$ when $D_{\mu}+D_{\sigma}>1$.
We can also use the right hand side to get a constant lower bound when $D_{\mu}+D_{\sigma}<1$. This will be useful when analyzing general graphs.

\smallskip
\begin{proof} {\textbf{(Lemma~\ref{lm:exactSumChain})}}
As above for convenience of notation we extend the chain with $X_0$ with $\mu_{1,0}=\mu_{1,1}=\mu_1$ and conditioning on $X_0$ does not affect the result. 
We first prove by induction from $j=n$ to $1$ that 
\begin{align}\nonumber
\sum_{S\subseteq \{X_j,\ldots,X_n\}} \EE_{X_j,\ldots,X_n | X_{j-1}} [\phi_{S}  X_n] 
& = 
\sum_{k=j}^{n}(\mu_{k,0}+\sigma_{k,0})\prod_{\ell=k+1}^{n}(D_{\ell,\mu}+D_{\ell,\sigma})
\\ &
+ \prod_{\ell=j}^{n}(D_{\ell,\mu}+D_{\ell,\sigma}) X_{j-1}.
\label{eq:variableportion}
\end{align}
This implies the claim in the lemma because 
$\sum_S \coeff{S} = \sum_{S\subseteq \{X_1,\ldots,X_n\}} \EE_{X_1,\ldots,X_n | X_{0}} [\phi_{S}  X_n]$
and 
for $j=1$, $D_{j,\mu}=D_{j,\sigma}=0$ and the term in \eqref{eq:variableportion} is zero.

For the base case, $j=n$ we have two cases: $S$ is empty or $S=\{X_n\}$.
When $S$ is empty, we have $$\EE_{X_n}[X_n]=\mu_{n,X_{n-1}}=\mu_{n,0}+D_{n,\mu}X_{n-1}.$$
When $S=\{X_n\}$, 
using Lemma~\ref{lm:phix},
we have $\EE_{X_n}[\phi_{X_n} X_n]=\sigma_{n,X_{n-1}}=\sigma_{n,0}+D_{n,\sigma}X_{n-1}$.
Adding the two cases we obtain $\sum_{S}\EE_{X_n}[\phi_{S}  X_n] = (\mu_{n,0}+\sigma_{n,0})+(D_{n,\mu}+D_{n,\sigma})X_{n-1}$.

Assume the claim holds for $j+1$ then we have
\begin{align}
\nonumber
 & \sum_{S\subseteq \{X_j,\ldots,X_n\}} \EE_{X_j,\ldots,X_n | X_{j-1}} [\phi_{S}  X_n]
 = 
\sum_{S_1\subseteq \{X_j\}} \sum_{S_2\subseteq \{X_{j+1},\ldots,X_n\}} \EE_{X_j} \EE_{X_{j+1},\ldots,X_n | X_{j-1}} [\phi_{S_1} \phi_{S_2} X_n]\\
\label{eq:expLBstepa}
&\quad =
\sum_{S_1\subseteq \{X_j\}} \EE_{X_j}  [\phi_{S_1} \sum_{S_2\subseteq \{X_{j+1},\ldots,X_n\}} \EE_{X_{j+1},\ldots,X_n | X_{j-1}}  [\phi_{S_2} X_n] ]
 \\
\label{eq:expLBstepb}
&\quad =
\sum_{S_1\subseteq \{X_j\}} \EE_{X_j}  [\phi_{S_1} 
\sum_{k=j+1}^{n}(\mu_{k,0}+\sigma_{k,0})\prod_{\ell=k+1}^{n}(D_{\ell,\mu}+D_{\ell,\sigma})
+ \prod_{\ell=j+1}^{n}(D_{\ell,\mu}+D_{\ell,\sigma}) X_{j}
]\\
\nonumber
& \quad =
\sum_{S_1\subseteq \{X_j\}} \EE_{X_j}  [\phi_{S_1} 
(C_1+ C_2 X_{j})
],
\end{align}
where $C_1=\sum_{k=j+1}^{n}(\mu_{k,0}+\sigma_{k,0})\prod_{\ell=k+1}^{n}(D_{\ell,\mu}+D_{\ell,\sigma})$ and $C_2=\prod_{\ell=j+1}^{n}(D_{\ell,\mu}+D_{\ell,\sigma})$. 
In \eqref{eq:expLBstepa} we can pull $\phi_{S_1}$ out of the inner expectations because its variables (and their parents) precede $X_{j+1}$ in the BN ordering. Eq~\eqref{eq:expLBstepb} holds by the inductive assumption.
Now when $S_1$ is empty $$\EE_{X_j}  [\phi_{S_1} (C_1+ C_2 X_{j})]=\EE_{X_j}  [(C_1+ C_2 X_{j})]=C_1+ C_2 \mu_{j,X_{j-1}}.$$
When $S_1=\{X_j\}$, 
using Lemma~\ref{lem:pS properties} and Lemma~\ref{lm:phix}, 
$\EE_{X_j}  [\phi_{X_j} (C_1+ C_2 X_{j})] = C_2 \sigma_{j,X_{j-1}}$.
Adding the two cases we obtain $C_1 + (\mu_{j,0}+\sigma_{j,0}) C_2 + (D_{j,\mu}+D_{j,\sigma}) C_2 X_{j-1}$, as required.
\end{proof}

\begin{corollary}[Lower bound for chains, excluding coefficient of $S=\phi$]
\label{cor:constLBchain}
Under the conditions of Corollary~\ref{cor:exactSumChain}
\begin{align}
\sum_{S\not=\phi}  |\coeff{S}| 
\geq
(\mu_{0}+\sigma_{0})\frac{1-(D_{\mu}+D_{\sigma})^n}{1-(D_{\mu}+D_{\sigma})} - \frac{\mu_{0}}{1-D_\mu}.
\end{align}
Choosing $\mu_0=0.07$, $\mu_1=0.56$ (with $D_\mu=0.49$ 
and $(D_{\mu}+D_{\sigma})\approx 0.731$)
and $n\geq 23$ so that $(D_{\mu}+D_{\sigma})^n<10^{-3}$
we have
$\sum_{S\not=\phi}  |\coeff{S}| \geq 1.07147$.
\end{corollary}
\begin{proof}
We can directly calculate $|\coeff{\phi}|$ and subtract it. 
Since there are no $S$ terms, $\coeff{\phi}$ is a single compound $A'$ term (Eq~\eqref{eq:A' D'}) which simplifies to
$\coeff{\phi}=\sum_{k=0}^{n-1} \mu_{0} D_\mu^{k} < \frac{\mu_{0}}{1-D_\mu}$.
\end{proof}

\begin{remark}
The product case also yields a bound greater than 1 when the empty set is included. 
For products with $f(\bfx)=X_k$ we have
$\sum_{S}  \coeff{S}=\mu_k+\sigma_k$ which can be $>1$. 
For example, for $\mu=0.9$ where $\sigma=0.3$ we have $\sum_{S}  \coeff{S}=1.2$ which is close to the maximum possible. 
For chains we can get a slightly larger lower bound. 
For example, choosing $\mu_0=0.29$, $\mu_1=0.78$ (with $D_\mu=0.49$ 
and $(D_{\mu}+D_{\sigma})\approx 0.450$)
and $n\geq 10$ so that $(D_{\mu}+D_{\sigma})^n<10^{-3}$
we have
$\sum_{S\not=\phi}  |\coeff{S}| > 1.35$.
\end{remark}

\begin{remark}
However, it is important to distinguish chains from the product case. 
For products we cannot exclude the empty set (that contributes $\mu_k$) and retain the $>1$ condition
whereas as shown by Corollary~\ref{cor:constLBchain} this is possible for chains. 
\end{remark}

We next show that for general BN, the spectral norm of a single literal can be exponentially large even when $D_{\mu}+D_{\sigma}<1$.
The case of nodes with multiple parents requires a more general notation. Here we focus on nodes with two parents
where $$P(X_i=1|X_jX_k=(00,01,10,11))=(\mu_{i,00},\mu_{i,01},\mu_{i,10},\mu_{i,11})$$ and
we denote $D_{i,ab}=\mu_{i,ab}-\mu_{i,00}$. With this we have
\begin{align}
A_{i|jk} = \EE[X_i | X_jX_k] = \mu_{i,00} + D_{i,10} X_j + D_{i,01} X_k + (D_{i,11} -D_{i,10} -D_{i,01} ) X_j X_k.
\end{align}
For the construction below it suffices to discuss $p(X_i=1|X_jX_k=(00,01,10,11))=(\alpha+D,\alpha,\alpha,\alpha+D)$
so that $D_{i,10} = D_{i,01} =-D$ and $(D_{i,11} -D_{i,10} -D_{i,01} )=2D$ and the conditional expectation simplifies to 
\begin{align}
\label{eq:binexpand2}
A_{i|jk} = \EE[X_i | X_jX_k] = \mu_{i,00} - D X_j - D X_k + 2D X_j X_k.
\end{align}

\newcommand{\floor}[1]{\lfloor #1 \rfloor}

With this, we consider the following graph $G^*$ where all nodes have at most one child (which we refer to as an  anti-tree), and where the graph has a single sink.
We have 3 sets of nodes. $X_1,\ldots,X_n$ are the leaves of a full binary anti-tree, whose internal nodes are $V_1,\ldots V_{n-1}$ with $V_1$ as the unique sink. For concreteness, we can set $V_1$ as the unique sink, and define the edges $V_{j}\rightarrow V_{\floor{j/2}}$
and (seeing $X_{k}$ as $V_{k+n-1}$) edges $X_{k}\rightarrow V_{\floor{(k+n-1)/2}}$.
In addition,
for $k\in[1,n]$ we have a chain of nodes $Y_{k,1}\rightarrow \ldots\rightarrow Y_{k,m-1}$ and an additional edge connecting $Y_{k,m-1}\rightarrow X_k$.
Hence the use of $X_1,\ldots,X_n$ is just for notational convenience and we can identify 
$X_k$ as $Y_{k,m}$. 
We use the setting of Corollary~\ref{cor:constLBchain} and set $m=23$ so $G^*$ has $n(m+1)-1=24n-1$ nodes.

\begin{theorem}[Exponential Lower Bound for General Graphs]
 \label{thm:expLBantitree}
Consider the anti-tree $G^*$ (with $N=24n-1$ nodes) and the function $f(\{Y\},\{V\})=V_1$, i.e., a conjunction of size 1.
There exist bounded BN distributions defined by $G^*$ such that for all nodes $i$,
$\mu_{i,pa(i)}\geq 0.01$ and 
for all assignments $\gamma_1,\gamma_2$ to the parents of node $i$,
$| \mu_{i,pa(i)=\gamma_1}-\mu_{i,pa(i)=\gamma_2}| <0.49$
and 
$ \sum_S |\coeff{S}| = \Omega(1.05^n)= 2^{\Omega(N)}$.
\end{theorem}
\begin{proof}
For the lower bound we consider sets $S$ which include at least one element from $Y_{k,1},\ldots,Y_{k,m}$ for each $k$ and  no element in $V_1,\ldots V_{n-1}$. Let ${\cal S}$ be the set of sets satisfying this condition.
Our goal is to obtain a lower bound on
\begin{align}
\sum_{S}
| \coeff{S}|
>
\sum_{S\in {\cal S}}
| \coeff{S}|
=
\sum_{S\in {\cal S}}
|
\EE_{\{Y\},\{V\}} [\phi_S V_1 ] |
=
\sum_{S\in {\cal S}}
|
\EE_{\{Y\}}[\phi_S \EE_{\{V\}} [V_1 ]]|. 
\end{align}

Now $\EE_{\{V\}} [V_1 ]$ can be derived by recursively applying \eqref{eq:binexpand2}.
This creates a multinomial expression with one distinguished term that includes both variables in each application of \eqref{eq:binexpand2}. This term is $(2D)^{n-1}\prod_{k=1}^{n} X_k$. All other terms in the multinomial expression exclude at least one $X_k$ variable.
We denote such terms generically as $ t_\ell(\bar{X}_\ell)$ where $\bar{X}_\ell$ is a strict subset of $X_1,\ldots,X_n$ and $ t_\ell$ may be positive or negative and it includes some power of $2D$.
With this we have
\begin{align}
\EE_{\{V\}} [V_1 ]=(2D)^{n-1}\prod_{k=1}^{n} X_k + \sum_\ell t_\ell(\bar{X}_\ell).
\end{align}
We next consider the expectation w.r.t.\ variables in $\{Y\}$. Due to the partial order we can start with any $k$ chain and we can choose a different expectation ordering for each term in $\EE_{\{V\}} [V_1 ]$.
For each term $t_\ell(\bar{X}_\ell)$ let $j$ be the least index for which $X_j$ does not appear in $\bar{X}_\ell$
and let $\ell^*$ be the largest index for which $Y_{j,\ell^*}\in S$.
Letting $S_{|j}=S\cap \{Y_{j,1},\ldots,Y_{j,m}\}$ and $S_{|\not=j}=S\setminus S_{|j}$
we have 
\begin{align*}
\EE_{\{Y\}}  [t_\ell(\bar{X}_\ell) ]
& =
\EE_{\{Y_k | k\not = j\}} [\phi_{S_{|\not=j}}  \EE_{\{Y_j\}} [\phi_{S_{|j}} t_\ell(\bar{X}_\ell) ]]
\\ & =
\numberthis \label{eq:genExpLBstep3}
\EE_{\{Y_k | k\not = j\}} [\phi_{S_{|\not=j}}  
\EE_{\{Y_{j,1},\ldots,Y_{\ell^*-1} \}} [\phi_{S_{|j\setminus (j,\ell^*)}} t_\ell(\bar{X}_\ell) \EE_{\{Y_{j,\ell^*},\ldots,Y_m \}} [\phi_{\ell^*}() ]]]
\\ & =
\EE_{\{Y_k | k\not = j\}} [0]=0.
\end{align*}
As above, in \eqref{eq:genExpLBstep3}
we can pull $\phi_{S_{|\not=j}} $ out of the inner expectations because its variables (and their parents) are not descendants  of the variables in the expectation. The inner expectation in \eqref{eq:genExpLBstep3} is zero due to Lemma~\ref{lem:pS properties}. Hence all terms other than the distinguished term contribute 0. Now recalling that $X_k$ is $Y_{k,m}$ and noting the independence between the $Y$ chains we have
\begin{align*}
\sum_{S\in {\cal S}}
| \coeff{S}|
 & =
(2D)^{n-1}  
\sum_{S\in {\cal S}}
|
\EE_{\{Y\}}[\phi_S  \prod_{k=1}^{n} Y_{k,m} ]|
\\ & =
(2D)^{n-1}  
\sum_{S\in {\cal S}} \prod_{k=1}^{n} |\EE_{\{Y_k\}} [\phi_{S_{|k}}  Y_{k,m} ]|
\\ & =
(2D)^{n-1}  
 \prod_{k=1}^{n} \sum_{S\subseteq Y_k, S\not=\phi} |\EE_{\{Y_k\}} [\phi_{S_{|k}}  Y_{k,m} ]|,
\end{align*}
where we can swap the order of summation and expectation because the functions in the expression are evaluated on disjoint sets of variables. 

Now using Corollary~\ref{cor:constLBchain} we have
\begin{align*}
\sum_{S}
| \coeff{S}|
\geq
\sum_{S\in {\cal S}}
| \coeff{S}|
>
(2D)^{n-1}  (1.07147)^n > (2\times 0.49 \times 1.07147)^{n}  > 1.05^n
\end{align*}
where we have chosen
$\alpha=0.01$, $D=0.49$ for the binary nodes and $Y$ nodes follow the setting in
Corollary~\ref{cor:constLBchain}, 
i.e.,  
$\mu_0=0.07$, $\mu_1=0.56$ with $D_\mu=0.49$. 
\end{proof}

\section{Fourier Expansion of Conjunctions for $k$-Junta Distributions}\label{sec:Kjunta}

In this section we consider the class of $k$-junta distributions  \cite{AliakbarpourBR16} that model distributions where, intuitively, $k$ of the variables capture the complexity of the distribution. 
More formally, 
let generalized $k$-junta distributions be distributions such that conditioned on a set $J$ of size $k$ the remaining variables have a product distribution (uniform in the original definition). 
Similarly, a depth-$d$ decision-tree (DT) distribution \cite{BlancLMT23} induces a uniform distribution in each leaf of the tree, hence it is a $k\leq 2^d$, $k$-junta distribution. \citet{AliakbarpourBR16} showed that $k$-junta distributions are learnable in time $O(n^k)$ from random examples, 
and \citet{ChenJLW21} show that the complexity can be reduced with subcube conditioning queries.
\citet{BlancLMT23,BlancLST25} 
developed distribution lifting results for depth $d$ decision-tree distributions that allow, for example, to lift 
learnability of functions under the uniform distribution to depth $d$ decision-tree distributions and can hence be applied to DNF learning.

It is interesting to compare these notions to representations of distributions using BN. First, it is clear that
a generalized $k$-junta distribution can be captured with a BN with a simple structure. Some arbitrary DAG represents the distribution 
over the set $J$, and all other variables depend on $J$ and are conditionally independent given $J$.
On the other hand, even limited BN provide flexibility that cannot be captured with shallow DT. For example consider a chain BN with $n$ variables with significant correlation in conditional expectations.
For example, we can use
$p(x_i=1|x_{i-1}=1)=0.25+c/2$ and $p(x_i=1|x_{i-1}=0)=0.75-c/2$ which is difference bounded.
If a path in a DT conditions on $d<n/2-1$ variables then the remaining variables form 
$<d+1$ segments in the chain, and by the pigeonhole principle at least one segment has more than one variable. Then the conditional distribution is not a product distribution, and a shallow decision tree cannot capture such distributions. 

As the next lemma shows, the spectral norm of conjunctions for $k$-junta distributions is bounded and hence the learnability results in this paper hold for this class.

\begin{lemma}\label{lm:kjunta}
Let $D$ be a $k$-junta distribution. Then the spectral norm of conjunctions $f$ with $d$ literals is bounded by 
$\L1(f) \leq 2^{(k+d)/2}$. 
\end{lemma}
\begin{proof}
Given a conjunction $f$ with $d$ literals using variable set $C$, the set of ancestors of $C$ in the BN is a subset of  $C\cup J$, and hence includes at most $d+k$ variables. Consider any set $S$ such that $S\setminus (C\cup J)\not=\emptyset$. 
As in the proof of Lemma~\ref{lem:pS properties}, we can compute 
$\fS=\EE_D[f(\bfX)\pS(\bfX)]$ 
by first computing the expectation over a variable $x_j\in S\setminus (C\cup J)$. Since this satisfies the BN ordering the internal expectation is 
$\EE[\phi_{x_j}(\bfX)]=0$ implying that $\hat{f}_S=0$. 
This implies that the number of non-zero coefficients of $f$ is bounded by $2^{k+d}$.
Due to the sparsity, we can use the argument for product distributions \cite{Feldman2012}.
In particular, note that $1\geq \EE[f^2]=\sum_S {\fS}^2$ and by the Cauchy-Schwartz inequality $$\sum_S |\fS| \leq 2^{(k+d)/2}\sqrt{\sum_S \fS^2}\leq 2^{(k+d)/2}.$$
\end{proof}

\section{Learning DNF and Decision Trees}
\label{sec:learnDNF}
In this section, we investigate the learnability of DNF expressions  using the BN-induced Fourier expansion.
While these results were previously proved for the uniform or product distributions, they hold more or less directly in the new setting. This section reviews some of these implications. We note that decision trees with $s$ leaves are a subset of disjoint DNF with $s$ terms. Hence while we present the results for disjoint DNF, they apply to decision trees as well.

A key requirement in this analysis is that the spectral norm 
of conjunctions is bounded. 
In this section, we assume that such a bound exists.
In particular for a distribution $D$ and conjunction $f$ with $d$ literals, we assume that $\L1(f) \leq L_1(d)$ for some function $L_1(d)$
and prove the results in a general form using that bound. 
Hence learnability holds whenever $L_1(d)$ is available, 
including for difference bounded tree distributions where $L_1(d)=O((\frac{2-2\alpha}{1-2\alpha})^{2d})$,
and $k$-junta distributions where $L_1(d) =O(2^{(k+d)/2})$. 

We need the following observation for bounded distributions:
\begin{lemma}
\label{lm:boundedprob}
For any $c$-bounded distribution $D$ and for any conjunction $f$ with $d$ literals, $c^d \leq \EE_D[f(\bfX)]\leq (1-c)^d$.
\end{lemma}
\begin{proof}
The marginal probability (for both values 0 and 1) for any root variable is bounded in $[c,1-c]$. Similarly, for any node in the BN its conditional probabilities, conditioned on any value of the parents, is bounded in $[c,1-c]$. This implies that the same condition holds if we marginalize out all the variables not in the conjunction from the BN. Therefore the joint setting of the remaining variables is bounded as claimed. 
\end{proof}

\subsection{Using KM Directly for Disjoint DNF}
We first show that the KM algorithm can be used directly to learn disjoint DNF under any distribution where the spectral norm of conjunctions is bounded.
Recall that DNFs are disjunctions of conjunctions and that in disjoint DNF the conjunctions are mutually exclusive. 
Decision trees whose node tests are individual variables are a subset of disjoint DNF, because they can be captured by the set of paths to leaves labeled 1. 

In particular, we show the following result which is a generalization of Theorem 3.10 of \citet{KM1993} to any distribution with a bounded spectral norm for conjunctions.

\newcommand{\putcorKMdisjointDNF}{
Consider any $c$-bounded distribution $D$ with its corresponding Fourier basis where $L_1(d)$ is a bound on the spectral norm of conjunctions of size $d$.
Let $f$ be any $s$-term disjoint DNF and
let $h(\bfx)$ be the output of $\mbox{KM}(D,f,\theta,\gamma,\delta)$
with 
$\epsilon'=\epsilon/3$,
$\theta=\epsilon'/2L_1$ 
and 
$\gamma^2={\epsilon'^3}/{16L_1^2}$
where $L_1 = s L_1(d)$ and $d=\log_{1-c}(\epsilon'/4s)$. 
Then with probability at least $1-\delta$, 
$\PP_D(f(\bfX)\not = \sign(h(\bfX))) \leq \epsilon$.
}

\begin{corollary}
\label{cor:KM DNF} 
\putcorKMdisjointDNF
\end{corollary}
Appendix \ref{app:km learn} presents the details of this result. The main ingredient is Lemma~\ref{lm:sparse-approx-via-L1} which shows that if $f$ can be approximated in square norm with a sparse function then it can be approximated with (estimates of) the large coefficients of $f$.

\subsection{Learning DNF Through Feldman's Algorithm}

The KM algorithm cannot be used directly to learn (non-disjoint) DNF. 
A lower bound by \citet{Mansour1995} shows that a super-polynomial set of coefficients of $f$ is needed to achieve a small square error. 
Despite this, two approaches exist to learn DNF through the Fourier basis. The first by \citet{Jackson1997} uses boosting to learn a Fourier based representation (but where the coefficients are different from the coefficients of $f$). The second, by \citet{Feldman2012}, gives a more direct algorithm with the same effect. \citet{Feldman2012}'s algorithm uses the KM algorithm as a subroutine to extract coefficients of functions so KM still plays an important role. This can be shown to succeed with the generalized basis.

In particular, we show the following result which is a generalization of Corollary 15 of \citet{Feldman2012} to any distribution with a bounded spectral norm for conjunctions.
\newcommand{\putcorFeldmanDNF}{
Consider any $c$-bounded distribution $D$ with its corresponding Fourier basis where $L_1(d)$ is a bound on the spectral norm of conjunctions of size $d$.
Let $f$ be any $s$-term DNF and
let $g(\bfx)$ be the output of $\mbox{PTFconstruct}(D,f,\gamma,\delta)$ 
where 
$\epsilon'=\epsilon/6$,
$d=\log_{1-c}(\epsilon'/4s)$,
$L_1=2 s L_1(d)+1$, 
and 
$\gamma=\frac{\epsilon'}{L_1}$.
Then with probability at least $1-\delta$, 
$\PP_D(f(\bfX)\not = g(\bfX)) \leq \epsilon$.
}
\begin{corollary} 
\label{cor:FC15}
\putcorFeldmanDNF
\end{corollary}

Appendix \ref{app:feldman dnf} presents the details of this result and reviews the main steps of \citet{Feldman2012}'s algorithm and how they can be used with the new basis and the generalized KM algorithm.

\subsection{Agnostic Learning of Disjoint DNF Through \citet{Gopalan2008}'s Algorithm}

The algorithm of \citet{Gopalan2008} is based on the observation that, under the uniform distribution, decision trees with $t$ leaves have  spectral norm at most $t$.
Therefore, one can use the class of functions with bounded spectral norm for agnostic learning of decision trees.
Their algorithm achieves this by using the KM algorithm to compute a subgradient for  the corresponding $2^n$ dimensional optimization problem. 
For general distributions, the spectral norm of Disjoint DNF (and decision trees) is not bounded. But Disjoint DNF can be approximated with Disjoint DNF whose terms are shorter than $d=\log_{1-c}(\epsilon/4s)$, that do have bounded spectral norm as derived above.
With this observation and the extended KM algorithm the analysis can be shown to apply more generally.

In particular, we show the following result which is a generalization of Theorem 20 of \citet{Gopalan2008} to any distribution with a bounded spectral norm for conjunctions.
\newcommand{\putagnosticDTcor}{
Consider any $c$-bounded distribution $D$ with its corresponding Fourier basis where $L_1(d)$ is a bound on the spectral norm of conjunctions of size $d$.
    Let $\mathcal{C}$ be the class of disjoint DNF with $s$ terms and 
    let 
    $d=\log_{1-c}(\epsilon/4s)$,
   and $t=s L_1(d)$.

    There exists an algorithm that, given $t$, $\varepsilon > 0$, and black-box
    access to any Boolean function
    $f : \{-1,1\}^n \to \{-1,1\}$, runs in time
    $\mathrm{poly}(n,t,\varepsilon^{-1})$ and outputs a hypothesis
    $h : \{-1,1\}^n \to \{-1,1\}$ such that
    \[
    \PP_{D}[h(X) \neq f(X) ] \leq \mathrm{opt}_{\mathcal{C}} + \varepsilon .
    \]
}
\begin{corollary}\label{cor:DT}
\putagnosticDTcor
\end{corollary}
Appendix \ref{app:goplan dnf} presents the details of this results and reviews the main steps of \citet{Gopalan2008}'s algorithm and how they can be used with the new basis and the generalized KM algorithm.

\section{Learning Difference bounded Tree BNs}
\label{sec:learnTrees}
In this section we remove the assumption that the distribution is known and given as a Bayesian network.
To reduce notational clutter (using $D$ for multiple purposes) in this section the distribution is denoted by $P$. 
Recall that for our learnability result we need bounded conditions on conditionals and differences between conditionals.
For concreteness, we focus on $\alpha$-difference bounded distributions $P$ as in  Definition~\ref{def:D bounded}
with $\alpha=0.5-c$.

It is well known that tree BN are learnable using the Chow-Liu algorithm \cite{Chow1968,Chow1973} and 
\citet{Hoffgen93,Dasgupta97,BhattacharyyaGPTV23}
provide finite samples guarantees for the agnostic case. 
That is, for any distribution $P$ one can learn a tree distribution which is at most $\epsilon$ away from the best tree approximation of $P$. 
However, the output of the Chow-Liu algorithm may not produce a difference bounded BN which may prevent it from being usable by our algorithm. 
In the following we show how this can be circumvented.
We provide two approaches. One for the realizable case, where the target distribution can be represented by a difference bounded tree, and a second slightly more complex algorithm for the unrealizable case when this does not hold. To facilitate the presentation we first review the algorithm and analysis for trees without constraints.

\subsection{Learning Tree BN Distributions}

\begin{definition}
    For any discrete probability distribution $P$, define $$H(P):=\sum_x -P(x)\log P(x) = \EE_{X\sim P}[-\log P(X)]$$ as the entropy function. For a joint probability distribution $P$ with marginals $P_X, P_Y$ define the mutual information as $$I(P) := \sum_{x,y}P(x,y)\log \tfrac{P(x,y)}{P_X(x)P_Y(y)}.$$  
\end{definition}

\noindent The {\textbf{Undirected Chow-Liu Algorithm with smoothing}} works as follows:
\begin{itemize}
\item Take a large sample and calculate the empirical distribution $\hat{P}_{i,j}$ for every pair of variables $(i,j)$.
\item
Construct a weighted complete undirected graph $G=(V,E)$ where $V=\{1,\ldots, n\}$, and for all $e=\{i,j\}\in E$ assign $w(i, j)=I(\hat{P}_{i,j})$.
\item
Find a maximum spanning tree of $G$ (or a minimum spanning tree with edge weights $w(i,j)$ replaced by $M-w(i,j)$ for some constant $M$).
Orient the tree using any node as root and let this tree be $\hat{T}$.
\item
Estimate conditional probabilities $P(X_i| X_{\pa(i)})$ on the edges of $\hat{T}$ using 
Laplace smoothing (i.e., adding 1 to all counts; for binary random variable $z$ when we have $k$ successes in $m$ trials the estimate is $p(z=1)=\frac{k+1}{m+2}$). 
Note that here we adopt the convention from prior work and for a root node this means that we estimate its marginal probability.
Denote the resulting tree-induced BN distribution by $\hat{P}^+$.
\end{itemize}

\newcommand{\wtp}{\mbox{wt}_P}
\newcommand{\wtph}{\mbox{wt}_{\hat{P}}}
\newcommand{\wtG}{\mbox{wt}_{G}}
\newcommand{\dkl}{d_{KL}}

For Any distribution $P$ and tree $T$, let ${P}_{{T}}$ be the distribution induced by $T$ when using exact conditional probabilities from $P$.
For any tree $T$ let $$\wtp(T)=\sum_{(i,j)\in T} I(P_{i,j})$$ be the weight of $T$ under distribution $P$.
We have the following: 
\begin{lemma}
\label{lm:CL} 
{\textbf{(\cite{Chow1968}. See also Lemma~3.3 of \citet{BhattacharyyaGPTV23})}}
For any tree $T$,
\begin{align}
\label{eq:treeweight}
\dkl(P|| P_T)=J_P -\wtp(T)
\end{align}
where $J_P :=\sum_i H(P_i)-H(P)$ is independent of $T$.
\end{lemma}

This claim and the more general statement of Lemma 3.3 of \citet{BhattacharyyaGPTV23} hold also for forests.
In these results, note that the root (or roots in forest) does not have a parent. \citet{Chow1968} implicitly use node 0 as the parent where $I(P_{i,j})$ is zero due to independence. This can be verified to be correct both for trees and forests. 
The lemma implies that maximizing $\wtp(T)$ is equivalent to minimizing $\dkl(P|| P_T)$ and hence motivates the algorithm.
The proof shows that using estimates from $\hat{P}$ with smoothing yields a sufficiently good approximation.

For finite sample analysis, let $T^*=\argmin_T \dkl(P|| P_T)$ and $\gamma=  \dkl(P|| P_{T^*})$
and
let $\hat{T}$ be the tree output by the CL algorithm, that is $\hat{T}=\argmax_T \wtph(T)$.
Note that since $\hat{P}^+$ is a tree BN distribution structured by $\hat{T}$ we have $\hat{P}^+=\hat{P}^+_{\hat{T}}$.
The following observation gives the key component of the analysis:
\begin{observation} [High Level Argument \cite{Hoffgen93,BhattacharyyaGPTV23}]
\label{lm:hatTworks} \ \\ 
If the following conditions hold 
\begin{align*}
        &(C1)\qquad     &\wtp(\hat{T})&\geq \wtp(T^*)-\epsilon/2,\\
        &(C2)\qquad     &\dkl(P|| \hat{P}^+_{\hat{T}})&\leq \dkl(P|| P_{\hat{T}})+\epsilon/2,
\end{align*}
then 
$\dkl(P|| \hat{P}^+_{\hat{T}})\leq  \dkl(P|| P_{{T^*}})+\epsilon.$
\end{observation}
\begin{proof}
If (C1) holds then by \eqref{eq:treeweight}, 
$$\dkl(P|| {P}_{\hat{T}})=J_P-\wtp(\hat{T}) \leq J_P-\wtp({T^*})+\epsilon/2 
= \dkl(P|| P_{{T^*}})+\epsilon/2.$$
Then, using (C2) we have 
$$\dkl(P|| \hat{P}^+_{\hat{T}})\leq \dkl(P|| P_{\hat{T}})+\epsilon/2
\leq \dkl(P|| P_{{T^*}})+\epsilon.$$
\end{proof}

\noindent
\textbf{Note:} The algorithm uses non smoothed probabilities to select $\hat{T}$ and smoothed probabilities to calculate $\hat{P}^+$ but the chaining of conditions in (C1), (C2) is designed to work.

\smallskip
\noindent
\citet{BhattacharyyaGPTV23}  show that the conditions needed are satisfied.
We first have:

\begin{lemma}  [$I$ is well approximated: \cite{BhattacharyyaGPTV23} Lemma 5.2]
\label{lm:BGPV-Iapprox}
When using sample size 
$$M=\Theta\qty(\frac{n^2}{\epsilon^2}\log\frac{n}{\delta}\log^2\qty(\frac{n}{\epsilon}\log\frac{n}{\delta})),$$
w.p.\ $\geq 1-\delta$, we have
(1) for any edge $|I(\hat{P}_{i,j})-I(P_{i, j})|\leq \frac{\epsilon}{4n}$, and therefore
(2) for any tree $T$,
$|\wtp(T)-\wtph(T)|\leq \epsilon/4$.
\end{lemma} 
This implies that the empirical weight maximizer on $\hat{P}$ is not far off from optimal:

\begin{lemma} [Condition C1 holds: \cite{BhattacharyyaGPTV23} Lemma 5.2]
\label{lm:BGPV-c1}
When running the CL algorithm  with sample size 
$M=\Theta\qty(\frac{n^2}{\epsilon^2}\log\frac{n}{\delta}\log^2(\frac{n}{\epsilon}\log\frac{n}{\delta}))$,
w.p.\ $\geq 1-\delta$, 
it outputs a tree $\hat{T}$, such that 
$\wtp(\hat{T})\geq \wtp(T^*)-\epsilon/2$.
\end{lemma} 
The second component is given by:

\begin{lemma} [Condition C2 holds: \cite{BhattacharyyaGPTV23} Theorem 1.4]

\label{lm:BGPV-c2}
For any $\hat{T}$,
for sample size $$M=\Theta\qty(\frac{n}{\epsilon}\log\frac{n}{\delta}\log(\frac{n}{\epsilon}\log\frac{1}{\delta})),$$ w.p.\ $\geq 1-\delta$, $\dkl(P|| \hat{P}^+_{\hat{T}})\leq \dkl(P|| P_{\hat{T}})+\epsilon/2$.
\end{lemma} 
This lemma also holds for forests where the empty set of parents behaves as in Lemma~\ref{lm:CL}.

\subsection{Learning Difference Bounded Tree Distributions: The Realizable Case}

In this section we assume the realizable case. Let $T^*$ be a difference bounded forest with constant $\alpha=(0.5-{c})$ such that $P=P_{T^*}$.
We have $\gamma=  \dkl(P|| P_{T^*})=0$.

\noindent The \textbf{Difference-restricted directed Chow-Liu Algorithm with smoothing} 
works as follows:
\begin{itemize}
\item
Take a large sample and calculate the empirical mutual information $I(\hat{P}_{i,j})$ between every pair of variables
using the empirical distribution $\hat{P}_{i,j}$.
\item
Construct a weighted complete {\em directed} graph $G=(V,E)$ where $V=\{1,\ldots, N\}$, and for all pairs $\{i, j\}$ 
assign the same weight to edges in both directions
$w(i, j)=w(j,i)=I(\hat{P}_{i,j})$.
\item
For each edge $(i, j)$, 
compute its Laplace smoothing conditional probability, and if it is not $(0.5-\frac{c}{2})$-difference bounded replace the weight with $w(i, j)=0$. 
\item
Find a maximum directed spanning tree $\tilde{T}$ of $G$ 
using the algorithm of \citet{Edmonds1967,GGST1986}
(or a minimum spanning tree with edge weights $w(i,j)$ replaced by $M-w(i,j)$).
\item
Remove 0 weight edges from $\tilde{T}$ to form a forest $\hat{T}$.
\item
Estimate conditional probabilities $p(X_i| X_{pa(X_i)})$ on the edges of $\hat{T}$ using the original sample,
with Laplace smoothing (i.e., adding 1 to all counts). 
Here too, for a root node this means that we estimate its marginal probability.
Denote the resulting forest induced BN distribution by $\hat{P}^+$.
\end{itemize}

Note that we assume that $T^*$ is difference bounded with constant $(0.5-{c})$ and that the algorithm filters edges violating boundedness for $(0.5-\frac{c}{2})$ which is less strict.

\begin{lemma} [Condition C1 holds] 
\label{lm:directed-c1-realizable}
When running the directed CL the algorithm  with sample size 
$M=\Theta(\frac{n^2}{\epsilon^2}\log\frac{n}{\delta}\log^2\qty(\frac{n}{\epsilon}\log\frac{n}{\delta}))$,
w.p.\ $\geq 1-\delta$, 
it outputs a tree $\hat{T}$, such that 
$\wtp(\hat{T})\geq \wtp(T^*)-\epsilon/2$.
\end{lemma} 
\begin{proof}
We first argue that 
with probability $\geq 1-\delta/2$
the edges of $T^*$ are not artificially assigned weight 0 in $G$.
To prove this it suffices to show that: 
\begin{align}
\label{eq:cestimate}
\mbox{ For all variables $i$ and values $b$ to their parent $j$},
|\hat{p}(x_i|X_j=b)-p(x_i|X_j=b)|\leq c/4. 
\end{align}
If this holds then 
$3c/4\leq 
\hat{p}(x_i|X_j=b) \leq
1-3c/4$
and $$|\hat{p}(x_i|X_j=1)-\hat{p}(x_i|X_j=0)|\leq |{p}(x_i|X_j=1)-{p}(x_i|X_j=0)|+c/2\leq 0.5-c/2,$$ as required.

For fixed $i$ and $b$, if we have at least $M_1=\frac{8}{c^2}\ln\frac{16n}{\delta}$ 
samples with $X_j=b$
then Hoeffding's bound implies that this holds with probability 
 $\geq 1-\delta/8n$.
If we have enough samples for all $i,b$ then \eqref{eq:cestimate} holds with probability $\geq 1-\delta/4$.

To guarantee $M_1$ samples note that for any $b$ we have $P(success)=q \geq c$ and let $\hat{q}$ be the number of successes in $M_2=2M_1/c$ trials. 
Using a 1 sided Hoeffding bound we have that if $M_2\geq \frac{2}{c^2}\ln\frac{8n}{\delta}$ 
then with probability $\geq 1-\delta/4$,
for all $i,b$ we have
$q-\hat{q}\leq c/2$ and the number of successes is at least $(c/2) M_2 =M_1$.
This implies that \eqref{eq:cestimate} holds with probability $\geq 1-\delta/2$.

Next, for any tree $T$ let $\wtG({T})$ be the sum of $T$ edge weights in $G$.
We have that
with probability $\geq 1-\delta/2$
\begin{align}
\label{eq:empiricalDCL}
\wtph(\hat{T}) = \wtG(\tilde{T})\geq  \wtG({T^*})=\wtph({T^*}).
\end{align}
The left equality holds because $\hat{T}$ only removes zero edge weights from $\tilde{T}$ and for other edges the weights are equal.
The right equality holds because $T^*$ edges are not assigned zero weight in $G$.

We next observe that  Lemma~\ref{lm:BGPV-Iapprox} implies that with probability $\geq 1-\delta/2$, $|\wtp(\hat{T})-\wtph(\hat{T})|<\epsilon/4$
and
$|\wtp({T^*})-\wtph({T^*})|<\epsilon/4$.
Combining this with \eqref{eq:empiricalDCL}
we have with probability $\geq 1-\delta$,
$\wtp(\hat{T}) \geq  \wtp({T^*})+\epsilon/2.$
Finally, note that for constant $c$ the sample bound in the statement (from Lemma~\ref{lm:BGPV-Iapprox}) is larger than $M_2$ 
and the condition for $M_2$ above holds.
\end{proof}

We next note that Lemma~\ref{lm:BGPV-c2} can be used without change on the forest $\hat{T}$ to guarantee C2
and therefore Observation~\ref{lm:hatTworks} implies

\begin{corollary} [Learning $(0.5-c)$-Difference bounded Tree BN]
\label{lm:directed-c1}
When running the directed CL the algorithm  with sample size 
$M=\Theta(\frac{n^2}{\epsilon^2}\log\frac{n}{\delta}\log\qty(\frac{n}{\epsilon}\log\frac{n}{\delta}))$
on distribution $P$ which can be represented by a $(0.5-c)$ difference bounded tree BN,
w.p.\ $\geq 1-\delta$, 
it outputs a forest $\hat{T}$ and a $(0.5-\frac{c}{2})$ difference bounded distribution $\hat{P}^+$ such that 
$\dkl(P|| \hat{P}^+_{\hat{T}})\leq  
\epsilon.$
\end{corollary} 

\subsection{Learning Difference Bounded Tree Distributions: The Unrealizable Case}

In this subsection, we extend our results to the unrealizable case where $P\neq P_T$ for any difference bounded tree $T$. The unrealizable case encapsulates two scenarios: (i) $P$ does not have any tree structure BN, and (ii) $P=P_T$ for some tree $T$ that is not difference bounded.  We present an algorithm that learns a difference bounded tree distribution that is close to the best such approximation of $P$. 
This criterion is captured by the following definition.
\begin{definition}\label{def:tree approx}
A difference-bounded tree BN $Q^*_{T^*}$ is an $\epsilon$ tree-approximation of  $P$ if 
\begin{align}\label{eq: eps approximate}
   \dkl(P|| Q^*_{T^*}) \leq \epsilon + \min_{\text{tree}~ T}\min_{\text{difference bounded}~Q} \dkl(P \| Q_{T}).
\end{align}
\end{definition}

For any pair of nodes $(i,j)$, let
$$f(P_j, P_{i|j}, Q_{i|j}):=\sum_{x\in \{0,1\}} P_j(x) \dkl\qty(P_{i|j}(\cdot | x) ~||~ Q_{i|j}(\cdot | x))-I(P_jP_{i|j}),$$ where $P_jP_{i|j}$ gives a joint distribution for variables $i,j$. 
Then, for any tree $T$, and distributions $P, Q$ we define $f(P_T, Q_T)$  to be the sum  
$$f(P_T, Q_T) := \sum_{v\in T} f(P_{\pa(v)}, P_{v|\pa(v)}, Q_{v|\pa(v)}).$$

With the above notations and recalling the definition of $J_p$, the following statement holds.
\begin{lemma}[\cite{BhattacharyyaGPTV23} Lemma 3.3]\label{lem:dkl p Q tree}
    For any tree $T$ and distributions $P, Q$
    \begin{align*}
        \dkl(P || Q_T) = J_P  +f(P_T, Q_T).
    \end{align*}
\end{lemma}
The lemma implies that the KL divergence decomposes into two components: the first term $J_p$ depends only on $P$, the second term depends on $T$ and $Q$. Next, we define the following term for any pair of nodes $(i,j)$
\begin{align}
    & \ltp(i,j):= \min_{Q_{i|j}}f(P_j, P_{i|j}, Q_{i|j}), 
    \nonumber
\\ &
\label{eq:lPoptimization}
\mbox{subject to } \qquad
    c\leq Q_{i|j}( y| x)\leq 1-c, \qquad \abs{Q_{i|j}(y | 0) - Q_{i|j}(y | 1)}\leq \alpha.
\end{align}
This definition is naturally extended to a tree: 
\begin{align*}
    \ltp(T):=\sum_{v\in T} \ltp(v, \pa(v)).
\end{align*}
With this notation the right-hand side of \eqref{eq: eps approximate} is simplified to 
\begin{align*}
    \dkl(P|| Q^*_{T^*}) \leq \epsilon + J_P + \min_{\text{tree}~ T} \ltp(T)
\end{align*}
Therefore, we consider $\ltp(T)$ as the metric for choosing the right tree structure and modify the algorithm accordingly.
Here we combine the two steps of tree selection and parameter estimation into one step where the latter is performed via  \eqref{eq:lPoptimization}.
Note that the optimization over $Q_{i|j}$ is done independently for each edge in $T$.
In addition, it is easy to check that the optimization objective in \eqref{eq:lPoptimization} is convex in its parameters 
$Q_{i|j}(y | 0),Q_{i|j}(y | 1)$, the constraints are linear and the domain is convex the therefore the optimization can be solved efficiently.

\noindent The \textbf{$l_P$-based Difference-restricted directed Chow-Liu Algorithm} works as follows:
\begin{itemize}
\item
Take a large sample and calculate the empirical   distribution $\hat{P}_{i,j}$ for any pair of variables $X_i, X_j$, including the marginals $\hat{P}_i, \hat{P}_j$.  %
\item
Construct a weighted complete {\em directed} graph $G=(V,E)$ where $V=\{1,\ldots, N\}$, and for all edges $(i,j)$ assign the weight
$\ltphat(i, j)$.
\item
Find a minimum directed spanning tree $\hat{T}$ of $G$ 
using the algorithm of \citet{Edmonds1967,GGST1986}.
\item
Return $\hat{T}$ and the  difference bounded distribution  $\hat{Q}$ induced by the optimization problem in each  $\ltphat(i,j)$.
\end{itemize}

To show that the above algorithm works, we  first need the following remark extending Observation \ref{lm:hatTworks}.
\begin{remark}
    If the following conditions hold 
    \begin{align*}
        &(A1)\qquad     &\ltp(\hat{T})&\leq \min_{\text{tree}~T} \ltp(T)+\epsilon/2\\
        &(A2)\qquad     &\dkl(P|| \hat{Q}_{\hat{T}}) & \leq \dkl(P|| {Q}^*_{\hat{T}}) +\epsilon/2
    \end{align*}
   where $Q^*_{\hat{T}}$ is the minimizer of $\ltp(\hat{T})$, then $(\hat{T}, \hat{Q})$ is 
   an $\epsilon$ tree-approximation of  $P$
\end{remark}

\begin{proof}
    \begin{align*}
        \dkl(P\| \hat{Q}_{\hat{T}}) &%
                                      {\leq} \dkl(P\| {Q}^*_{\hat{T}}) +\epsilon/2\\
                                    &= J_P + \ltp(\hat{T}) +\epsilon/2\\
                                    & \leq  J_P  + \min_{\text{tree}~ T}\ltp(T)+\epsilon
    \end{align*}
\end{proof}
It remains to show that conditions (A1) and (A2) hold when the sample size is large enough. We proceed with the following technical results:  

\begin{proposition}\label{prop:KL upper bound}
    If $Q$ is a distribution such that $Q(x)\geq c$  for any $x$, then $\dkl\qty(P \| Q)\leq \log 1/c$.
\end{proposition}
\begin{proof}
    By definition, $\dkl\qty(P \| Q) = -H(P) + \EE_{X\sim P}[\log \frac{1}{Q(X)}].$ Hence, with the non-negativity of the entropy, the KL divergence is not greater than $\max_x \log \frac{1}{Q(x)}\leq \log\frac{1}{c}$. 
\end{proof}
\begin{proposition}\label{prop:binary entropy}
    (1) For any pair of binary probability distributions $P, Q$  the binary entropy satisfies:  $|H(P)-H(Q)|\leq h_b(\dtv(P, Q))$. (2) For any $p\in [0,\tfrac{1}{2}]$,    $h_b(p)\leq 3p\log 1/p$.
\end{proposition}
\begin{proof}
    For (1), note that the highest slope of $h_b(x)$ is at $x=0, 1$. Hence the maximum absolute difference of the binary entropies is at the highest when $P$ or $Q$ are trivial probability distributions. 
    Denoting $p=P(x=1)$ and  $q=Q(x=1)$ we have $|H(P)-H(Q)|\leq |H(|p-q|)-H(0)|=H(|p-q|)=h_b(\dtv(P, Q))$.
    For (2) By definition, we can write $h_b(p) = p\log \tfrac{1-p}{p}-\log (1-p).$ From the inequality, $\log x \geq 1-\tfrac{1}{x}$ for any $x>0,$ the binary entropy is upper bounded as $ h_b(p)\leq p\log \tfrac{1}{p} + \tfrac{p}{1-p}.$ Note that $(1-x)\log\tfrac{1}{x}\geq \tfrac{1}{2}$ for any $x \in (0, 1/2)$ implying that $2 \log \tfrac{1}{p} \geq \tfrac{1}{1-p}.$ Hence, $h_b(p) \leq 3p\log 1/p$.
\end{proof}

Next, we prove a lemma which is the basis for showing that conditions A1 and A2 hold. 
\begin{lemma}\label{lem:estimated lp}
   For any tree ${T}$, and the empirical distribution $\hat{P}$ obtained from $M=\tilde{\Theta}(\tfrac{n^2}{\epsilon^2} \log^2(1/c))$ samples, $\abs{\ltphat({T})- \ltp({T})}\leq \epsilon/2$.
\end{lemma}
\begin{proof}
    We start with a pair of nodes $i,j$ as a potential edge $j\rightarrow i$ in ${T}$. Recall the definition of  $f(P_j, P_{i|j}, Q_{i|j})$ that appears in $\ltp(i,j)$. We first show that, for any possible $Q$, $f$ does not change much when $P_j$ or $P_{i|j}$ are perturbed by a small amount in terms of total variation distance. 

    \noindent\textbf{Robustness against perturbations.} From Proposition \ref{prop:KL upper bound},  $\dkl\qty(P_{i|j}(\cdot | x) ~||~ Q_{i|j}(\cdot | x))\leq \log 1/c$, when  $Q_{i|j}$   is $c$-bounded. Then, for any $\hat{P}_j$
    \begin{align*}%
        \abs{f(P_j, P_{i|j}, Q_{i|j}) - f(\hat{P}_j, P_{i|j}, Q_{i|j})} \leq 2 \log \tfrac{1}{c} ~ d_{TV}(P_j, \hat{P}_j) + \abs{I(P_jP_{i|j}) - I(\hat{P}_jP_{i|j})},
    \end{align*} 
    where the inequality follows from the triangle inequality and the definition of the total variation distance. As for the difference of the mutual information terms, note that 
    \begin{align*}
        I(P_jP_{i|j}) - I(\hat{P}_jP_{i|j}) = H(P_i)-H(\hat{P}_i) + \sum_x (\hat{P}_j(x)-P_j(x)) H(P_{i|j}(\cdot|x)).
    \end{align*}
    Therefore, the absolute difference is bounded by 
    \begin{align*}
        h_b(\dtv(P_i, \hat{P}_i)) + 2\dtv(\hat{P}_j, P_j)
    \end{align*} 
    where $\hat{P}_i$ is the marginal induced by $\hat{P}_j P_{i|j}$.
    The total variation distance is a special case of $f$-divergence that satisfies the data processing inequality. For more details, see  \cite{Ali1966,Csiszar1963,Sason2015}. 
    This inequality implies that $\dtv(P_i, \hat{P}_i)\leq \dtv(P_j, \hat{P}_j)$. Therefore, using Proposition \ref{prop:binary entropy}, 
      \begin{align}\label{eq: P_j purturb}
        \abs{f(P_j, P_{i|j}, Q_{i|j}) - f(\hat{P}_j, P_{i|j}, Q_{i|j})} \leq \dtv(\hat{P}_j, P_j) \qty(4\log\tfrac{1}{c} + 3 \log 1/\dtv(\hat{P}_j, P_j)).
    \end{align} 
    This  establishes the robustness of $f$ against perturbations of $P_j$.     Note that the argument above does not rely on the identity of $P_{i|j}$, i.e., it also applies to $\hat{P}_{i|j}$ and the same bound applies to
    $| f(P_j, \hat{P}_{i|j}, Q_{i|j}) - f(\hat{P}_j, \hat{P}_{i|j}, Q_{i|j})|$.

    Next, we show the robustness of $f$ against perturbations of $P_{i|j}$. 
    Using the identity $\dkl(P \| Q) = -H(P)-\sum_xP(x)\log Q(x)$ we have that for any $P_j, P_{i|j}, Q_{i|j}$
    \begin{align*}
        f({P}_j, {P}_{i|j}, {Q}_{i|j}) &= -\sum_x P_j(x) H(P_{i|j}(\cdot | x)) - \sum_{y,x} P_j(x)P_{i|j}(y|x) \log Q_{i|j}(y|x) - I(P_jP_{i|j})\\
        & = H(P_i) - \sum_{y,x} P_j(x)P_{i|j}(y|x) \log Q_{i|j}(y|x),
    \end{align*}
    where $P_i = \sum_x P_j(x)P_{i|j}(\cdot|x),$ and we used $I(X;Y)=H(Y)-H(Y|X)$. Note that $|\log Q_{i|j}(y|x)|\leq \log (1/c)$ for any $c$-bounded $Q_{i|j}$. Therefore,  from the triangle inequality for any $P_j, P_{i|j}$ and $c$-bounded $Q_{i|j}$
\begin{align*}
    \abs{f({P}_j, \hat{P}_{i|j}, {Q}_{i|j}) - f({P}_j, {P}_{i|j}, {Q}_{i|j})} \leq &\abs\Big{H\qty\Big(\sum_x {P}_j(x)\hat{P}_{i|j}(\cdot|x)) - H\qty\Big(\sum_x {P}_j(x)P_{i|j}(\cdot|x))}\\
    & + \sum_{y,x} {P}_j(x) \log(\tfrac{1}{c}) \abs{\hat{P}_{i|j}(y|x) - P_{i|j}(y|x)}.
\end{align*}
   The second term in the RHS is upper bounded by $2\dtv^{*} \log(1/c)$, where $\dtv^{*}:= \max_x \dtv(\hat{P}_{i|j}(\cdot|x), {P}_{i|j}(\cdot|x))$. As for the first term, we use Proposition \ref{prop:binary entropy} to upper bound the difference of the two  entropy terms based on the total variation distance. Note that 
\begin{align*}
    \dtv\qty(\sum_x {P}_j(x)\hat{P}_{i|j}(\cdot|x), \sum_x {P}_j(x)P_{i|j}(\cdot|x)) \leq \max_x \dtv\qty(\hat{P}_{i|j}(\cdot|x), {P}_{i|j}(\cdot|x))=\dtv^{*}. 
\end{align*}
Therefore, from  Proposition \ref{prop:binary entropy} part (1) and then (2),  
\begin{align}\label{eq: P_i|j purturb}
    \abs{f({P}_j, \hat{P}_{i|j}, {Q}_{i|j}) - f({P}_j, {P}_{i|j}, {Q}_{i|j})}\leq 2\dtv^{*}\log\tfrac{1}{c} + 3\dtv^{*} \log 1/\dtv^{*}.
\end{align}
 This  establishes the robustness of $f$ against perturbations of $P_{i|j}$.    Note that the argument above does not rely on the identity of $P_{j}$, i.e., it also applies to $\hat{P}_{j}$ and the same bound applies to
    $| f(\hat{P}_j, \hat{P}_{i|j}, {Q}_{i|j}) - f(\hat{P}_j, {P}_{i|j}, {Q}_{i|j})|$.

\vspace{5pt}
 \noindent\textbf{Sample complexity analysis.}
In light of \eqref{eq: P_j purturb} and \eqref{eq: P_i|j purturb}, we present a sample complexity analysis for bounded total variation distance. We show that the RHS of \eqref{eq: P_j purturb}, is bounded by $\tfrac{\epsilon}{4n},$ if 
    \begin{align}\label{eq:dtv P_j}
        \dtv(P_j, \hat{P}_j)\leq \frac{a}{\log(1/a)},%
    \end{align}
    for all $j$, where $a=\frac{\epsilon}{48 n \log (1/c)}$. If this inequality holds, then the first term $4\log\tfrac{1}{c}\dtv(\hat{P}_j, P_j)$ in RHS of \eqref{eq: P_j purturb} is bounded by $\epsilon/8n$. The second term is also bounded by $\tfrac{\epsilon}{8n}$. To see this, let $d:=\dtv(P_j, \hat{P}_j)$ so that the second term in RHS of \eqref{eq: P_j purturb} is $3\log(1/c)d\log(1/d)$. 
    Note that $d\log(1/d)$ is monotonic for $d\leq 0.5$ implying that 
    \begin{align*}
        d\log(1/d)&\leq \frac{a}{\log(1/a)}[\log\log(1/a)-\log(a)]\leq \frac{a}{\log(1/a)}[2\log(1/a)]=2a.
    \end{align*}
    Hence, the second term is bounded by $6\log(1/c) a \leq \tfrac{\epsilon}{8n}$. 

     Using Chernoff's bound and a union bound over $j$, we have that \eqref{eq:dtv P_j} holds for all $j$ w.p $(1-\delta/2)$ when $\hat{P}_j$ is estimated over 
     \begin{align*}
        M = \Theta\qty(\frac{1}{\epsilon^2} n^2 \log^2(\frac{1}{c}) \log^2(\frac{1}{\epsilon}48 n \log\big(\frac{1}{c}))\log\frac{2n}{\delta}) = \tilde{\Theta}(\tfrac{n^2}{\epsilon^2}\log^2 (1/c)),  
     \end{align*}
     samples.  With this estimation, for all $j$, RHS of \eqref{eq: P_j purturb} is upper bounded by $\tfrac{\epsilon}{4n}$ for any $c$-bounded $Q_{i|j}$.
     
     Similarly, with a union bound over $i,j$ pairs,  we can show that for all $i,j$, the RHS of  \eqref{eq: P_i|j purturb} is smaller than $\tfrac{\epsilon}{4n}$ when $\dtv^*$ is $\tilde{O}(\tfrac{\epsilon}{n\log(1/c)\log n/\epsilon})$ as in the RHS of \eqref{eq:dtv P_j}. This is ensured when $\hat{P}_{i|j}$ is estimated over $M$ samples.

We have therefore established that for all $Q$ and for all $i,j$ with probability $>1-\delta$
        \begin{align}\label{eq: P_j purturb 2}
        \abs{f(\hat{P}_j, P_{i|j}, {Q}_{i|j}) - f({P}_j, P_{i|j}, {Q}_{i|j})}\leq \frac{\epsilon}{4n},
    \end{align}
     and 
    \begin{align}\label{eq:purturb Pij 2}
        \abs{f(\hat{P}_j, \hat{P}_{i|j}, {Q}_{i|j}) - f(\hat{P}_j, {P}_{i|j}, {Q}_{i|j})} \leq  \frac{\epsilon}{4n}.
    \end{align}

 \noindent\textbf{Final stage.}    
 Let ${Q}^*_{i|j}$ be the minimizer in  $\ltp(i,j)$ and $\tilde{Q}_{i|j}$ be the difference bounded distribution  minimizing  $f(\hat{P}_j, P_{i|j}, {Q}_{i|j})$. 
     Since $Q^*_{i|j}$ is not necessarily the minimizer of $f(\hat{P}_j, P_{i|j}, {Q}_{i|j})$, then 
    \begin{align}\label{eq:f* P hat vs f*}
        f(\hat{P}_j, {P}_{i|j}, \tilde{Q}_{i|j})  \leq f(\hat{P}_j, P_{i|j}, {Q}^*_{i|j})%
    \end{align}
    On the other hand,    
    \begin{align}\label{eq:f* hat Qtilde}
       f(\hat{P}_j, \hat{P}_{i|j}, \hat{Q}_{i|j})  \leq f(\hat{P}_j, \hat{P}_{i|j}, \tilde{Q}_{i|j}),
    \end{align}
    as $\hat{Q}_{i|j}$ is 
    the minimizer of $f(\hat{P}_j, \hat{P}_{i|j}, {Q}_{i|j})$.
    Therefore, from the robustness analysis of $f$, we have the following chain of inequalities, 
    \begin{align*}
        \ltphat(i,j) &= f(\hat{P}_j, \hat{P}_{i|j}, \hat{Q}_{i|j})   \leq f(\hat{P}_j, \hat{P}_{i|j}, \tilde{Q}_{i|j})\\
        & \stackrel{(a)}{\leq}  f(\hat{P}_j, {P}_{i|j}, \tilde{Q}_{i|j}) + \frac{\epsilon}{4n}\\
        & \leq  f(\hat{P}_j, P_{i|j}, {Q}^*_{i|j})+  \frac{\epsilon}{4n}\\
        & \stackrel{(b)}{\leq} f({P}_j, P_{i|j}, {Q}^*_{i|j})+  \frac{\epsilon}{2n} = \ltp(i,j)+\frac{\epsilon}{2n},
    \end{align*}
        where (a) is due to \eqref{eq:purturb Pij 2} and (b) is due to \eqref{eq: P_j purturb 2}.

    Repeating the same  argument but for $P$ and  $\hat{P}$ interchanged gives the inequality in the reverse direction, implying that for all $i,j$ we have a bound on the absolute difference 
    $\abs{\ltp(i,j) - \ltphat(i,j)}\leq  \frac{\epsilon}{2n}.$ 
    Using 
    the additivity of $\ltp$, and the triangle inequality for the absolute difference, this bound can be generalized to any tree ${T}$ as  $\abs{\ltphat({T}) - \ltp({T})}\leq \epsilon/2,$ where we used the fact that $\hat{T}$ has at most $n$ edges. 
\end{proof}

\begin{lemma}[Condition A2 holds]
    When running the $l_P$-based directed CL algorithm with $\ltphat$ as the weights and sample size $M=\tilde{\Theta}(\tfrac{n^2}{\epsilon^2} \log^2(1/c))$, w.p. $\geq 1-\delta$, condition A2 holds, that is $\dkl(P|| \hat{Q}_{\hat{T}}) \leq \dkl(P|| {Q}^*_{\hat{T}})+\epsilon/2$.
\end{lemma}
\begin{proof}
    Recall the definition of $f(P_j, P_{i|j}, Q_{i|j})$ and  $f(P_T, Q_T)$, and that $\dkl(P|| \hat{Q}_{\hat{T}})  =  J_P + f(P_{\hat{T}}, \hat{Q}_{\hat{T}})$.
    Using a similar argument to the proof of Lemma \ref{lem:estimated lp}, we next show that given $M=\tilde{\Theta}(\tfrac{n^2}{\epsilon^2} \log^2(1/c))$ samples, 
 for all possible $Q$, we have
$\abs{f(\hat{P}_{\hat{T}}, {Q}_{\hat{T}})-f(P_{\hat{T}}, {Q}_{\hat{T}})}\leq \epsilon/4$. 
To see this note that from the triangle inequality we have that
\begin{align*}
    \abs{f(\hat{P}_{\hat{T}}, {Q}_{\hat{T}})-f(P_{\hat{T}}, {Q}_{\hat{T}})} \leq \sum_{v\in \hat{T}} \abs{f(\hat{P}_{\pa(v)}, \hat{P}_{v|\pa(v)}, {Q}_{v|\pa(v)})- f(P_{\pa(v)}, P_{v|\pa(v)}, {Q}_{v|\pa(v)})}.
\end{align*}
Then, for each $v\in \hat{T}$ the absolute difference is bounded by
\begin{align*}
    &\abs{f(\hat{P}_{\pa(v)}, \hat{P}_{v|\pa(v)}, {Q}_{v|\pa(v)})- f(\hat{P}_{\pa(v)}, P_{v|\pa(v)}, {Q}_{v|\pa(v)})}\\
    &  +\abs{f(\hat{P}_{\pa(v)}, {P}_{v|\pa(v)}, {Q}_{v|\pa(v)})- f(P_{\pa(v)}, P_{v|\pa(v)}, {Q}_{v|\pa(v)})}
\end{align*}
    The first term can be bounded by $\frac{\epsilon}{8n}$ per \eqref{eq:purturb Pij 2} by increasing $M$ by a constant factor to obtain a $1/8$ factor. Similarly, the second term can be bounded by $\frac{\epsilon}{8n}$ per \eqref{eq: P_j purturb 2}. Given that $|\hat{T}|\leq n$, then $\abs{f(\hat{P}_{\hat{T}}, \hat{Q}_{\hat{T}})-f(P_{\hat{T}}, \hat{Q}_{\hat{T}})}\leq \epsilon/4$.

    Hence, 
    as $\hat{Q}_{\hat{T}}$ is the minimizer of $\ltphat(\hat{T})$ and 
    $Q^*_{\hat{T}}$ is the minimizer of $\ltp(\hat{T})$ 
    the following inequalities hold 
    \begin{align*}
        \dkl(P|| \hat{Q}_{\hat{T}}) & =  J_P + f(P_{\hat{T}}, \hat{Q}_{\hat{T}}) \\
           			 & \leq J_P + f(\hat{P}_{\hat{T}}, \hat{Q}_{\hat{T}}) + \epsilon/4\\
                                   & \leq J_P + f(\hat{P}_{\hat{T}}, {Q}^*_{\hat{T}}) + \epsilon/4\\
                                   & \leq J_P + f({P}_{\hat{T}}, {Q}^*_{\hat{T}}) + \epsilon/2\\
                                   & = \dkl(P|| {Q}^*_{\hat{T}})+\epsilon/2.
    \end{align*}
\end{proof}

\begin{lemma}[Condition A1 holds]
    When running the $l_P$-based directed CL algorithm with $\ltphat$ as the weights 
    and sample size $M=\tilde{\Theta}(\tfrac{n^2}{\epsilon^2} \log^2(1/c))$, w.p. $\geq 1-\delta$ the algorithm outputs a tree $\hat{T}$, such that $\ltp(\hat{T}) \leq \min_{\text{tree}~T} \ltp(T)+\epsilon/2$. 
\end{lemma}
\begin{proof}
    Using Lemma \ref{lem:estimated lp}, for $\hat{T}$ and $T^*$
    \begin{align*}
        \ltp(\hat{T}) \leq \ltphat(\hat{T}) +\epsilon/2 \leq \ltphat({T}^*) +\epsilon/2 \leq \ltp({T}^*) +\epsilon.
    \end{align*}
\end{proof}

\begin{corollary} [Learning $\epsilon$ Tree-Approximation Distributions] %
\label{lm:directed-c1 unrealizable}
When running the $\ltp$-based directed CL the algorithm with parameters $c, \alpha$ and   with sample size 
$M=\tilde{\Theta}(\tfrac{n^2}{\epsilon^2} \log^2(1/c))$
on any distribution $P$,
w.p.\ $\geq 1-\delta$, 
it outputs an $\epsilon$ tree-approximation of $P$ as in Definition \ref{def:tree approx}.
\end{corollary} 

\subsection{Implications for Learning DNF}
We can use Pinsker's inequality
$d_{TV}(p_1||p_2) \leq \sqrt{0.5 \ \dkl(p_1||p_2)}$ to bound the total variation distance.
Assume we learn the BN tree distribution to accuracy $\tilde{\epsilon}$.
Then  for the realizable case w.p.\ $\geq 1-\delta$, 
$d_{TV}(P,\hat{P}^+)\leq \epsilon = \sqrt{0.5 \tilde{\epsilon}}$. Hence, w.p.\ $\geq 1-\delta$,  for any event $E$, $P(E)\leq \hat{P}^+(E) + \epsilon$.
We then learn $f(x)$ using the basis derived from $\hat{P}^+_{\hat{T}}$ and combine with the  learning results from above.
For example, 
combining with Corollary~\ref{cor:FC15} we can conclude
that with probability at least $1-2\delta$, 
$P(f(X)\not = g(X)) \leq \hat{P}^+(f(X)\not = g(X)) + \epsilon \leq 2 \epsilon$.
For the unrealizable case, the algorithm outputs $\hat{T}$ and $\hat{Q}$, so that with high probability $\dkl(P|| \hat{Q}_{\hat{T}}) \leq \tilde{\epsilon} + \opt,$ where $\opt$ is the minimum KL divergence between $P$ and any difference bounded $Q_T$ with any tree structure $T$.   Hence, $$d_{TV}(P,\hat{Q})\leq \sqrt{0.5 (\tilde{\epsilon} + \opt)} \leq \epsilon + \sqrt{0.5 \opt} $$
with $\epsilon = \sqrt{0.5 \tilde{\epsilon}}$.
We then learn $f(x)$ using the basis derived from $\hat{Q}_{\hat{T}}$ and combine with the  learning results from above. 
In this case, we will have a residual error of  $\sqrt{0.5 \opt}$, that is, we have
$P(f(X)\not = g(X)) \leq 2 \epsilon+ \sqrt{0.5 \opt}$.

\section{Discussion and Conclusion}\label{sec:conclusion}

The paper develops a generalized Fourier basis and shows that major algorithmic tools from learning theory
can be used with this basis.
Using these and an analysis of the spectral norm for conjunctions the paper shows learnability of DNF under a broad class of distributions including $k$-junta distributions and difference bounded tree distributions,
significantly extending previous results.
We emphasize that the basis and the extended KM algorithm are valid for any distribution and they do not require a tree structure or boundedness. These conditions were only required for the bounds on the spectral norm of conjunctions which are required for learnability results for decision trees and DNF.

The introduction of the generalized Fourier basis suggests many questions for future work.
Characterizing the scope of distributions where the approach is successful is one such direction.
In particular,
our results provided upper bounds for the spectral norm using conditions on the structure of the graph (i.e., a tree) and the parameters of the distribution (difference bounded) and lower bounds showing that such restrictions are needed.
It would be interesting to investigate whether some tradeoff exists so that more general graphs can be accommodated with stricter requirements on parameters.

Another question concerns the low-degree algorithm of \citet{linial1993constant}. 
Our results for disjoint DNF indirectly imply that for a suitable degree $d$ the low-degree algorithm will succeed because, by the characterization of $\fS$, the large coefficients recovered by the KM algorithm necessarily yield low degree coefficients. 
But this does not necessarily hold for 
PTFconstruct \cite{Feldman2012} 
so it is not implied for DNF.
Hence it would be interesting to characterize the large degree $L_2$ spectral norm of DNF and constant depth circuits for general distributions. 

In addition, as discussed by \citet{ODonnell2014,Wolf2008},
the standard Fourier representation for the uniform distribution has found numerous applications in algorithms and complexity analysis. 
However, much of prior work relies significantly on independence in the distribution and sparsity of coefficient structure. 
It would be interesting to explore such applications using the generalized basis.

It is interesting to discuss two recent papers that imply learnability results for DNF with MQ using different techniques
that do not directly rely on Fourier analysis. 
The first \citep{BlancLST25} develops a general distribution lifting algorithm that, 
given learnability w.r.t.\ $D'$, 
shows learnability w.r.t.\ distribution $D$ where the complexity depends on decomposing $D$ via subcubes with distribution $D'$, 
which includes decision tree and $k$-junta distributions. 
This is a more general result than the one in this paper in that it applies to any distribution-specific learning algorithm,
and it implies learnability of DNF with depth $d$ decision tree distributions in time $O(n^{O(d)})$.
However, as discussed in Section~\ref{sec:Kjunta},
tree BN distributions are much more expressive than depth $d$ decision tree distributions, even for $d=O(\log n)$ where their algorithm is quasi-polynomial time.
The second work, by \citet{AlmanNPS25}, provides a new approach for learning DNF by generating terms through random walks over positive examples,
where the random walks can be obtained with MQ. Their results hold for any distribution, but they require a structural limitation on the class of DNF expressions, where all terms must have exactly the same length, and in addition their algorithm runs in  quasi-polynomial time.
In contrast our results increase the scope of distributions where polynomial time learnability of general DNF is known for a large class of tree BN distributions, but the lower bounds shown for the spectral norm indicate that the Fourier approach may not be successful for all distributions. 
A potentially interesting future direction is investigating learnability using the ideas from these papers and our approach based on BNs.

\section*{Acknowledgment}
This work was partially supported by the NSF Grant CCF-2211423.

\newpage
\appendix

\section{Learnability Results for DNF}

The appendix provides details of the learnability results for DNF stated in the paper. 
Recall that we use the notation $L_1(f)=\sum_S |\hat{f}_S|$ for the spectral norm of function $f$ with Fourier coefficients $\hat{f}_S$.
In this section we define $\|\hat{f}\|_\infty= \max_S |\hat{f}_S|$ and similarly use $\|\hat{f}\|_1= L_1(f)=\sum_S |\hat{f}_S|$.

\subsection{Using KM Directly for Disjoint DNF}\label{app:km learn}
Analysis in prior work required either a bound on the spectral norm of the learned function $f$ \cite{KM1993}, or a square error approximation of $f$ using a sparse function \cite{Khardon94}.
For general distributions, the lemma below provides sufficient conditions using the combination of square norm approximation and L1 approximation which is implicitly assumed in some prior work \cite{Mansour1995}.


\begin{lemma}
\label{lm:sparse-approx-via-L1} 
Consider any distribution $D$ specified by a BN and its corresponding Fourier basis, and any Boolean function $f$
that can be approximated in square norm by a function $h$ with bounded spectral norm, that is, 
$\EE_D[(f(X)-h(X))^2]\leq\epsilon/4$ and $\L1(h)\leq L_1$.

Then there exists a $g$ such that 
(1)
$\L1(g) \leq L_1$,
(2) $g$ is $T=4 L_1^2/\epsilon$ sparse,
and 
(3) 
$\EE_D[(f(X)-g(X))^2]\leq \epsilon$. 

Moreover, $f$ can be approximated by approximating its large Fourier coefficients. In particular, 
let
${\cal S}=\{S \mbox{ s.t. } |\coeffgen{f}{S}|\geq \sqrt{\epsilon/T}\}$, 
and let
$\tilde{h}(x)=\sum_{S\in {\cal S^*}} \approxcoeffgen{f}{S} \basis{S}(x)$, 
for some 
${\cal S^*}$ where
${\cal S}\subseteq {\cal S^*}$, $|{\cal S^*}|\leq 4T/\epsilon$,
where $|\approxcoeffgen{f}{S}-\coeffgen{f}{S}|\leq \gamma$,
and $\gamma^2 \leq \epsilon^2/4T$.
Then 
$\PP_D(f(X)\not = \sign(\tilde{h}(X))) \leq \EE_D[(f(X)-\tilde{h}(X))^2] \leq 3 \epsilon$.
\end{lemma}

Lemma~\ref{lm:sparse-approx-via-L1} is a combination of the next two results.
The next lemma provides the key ingredient in many Fourier based learning results. It shows that if $f$ can be approximated 
{\em in square norm} with a sparse function then it can be approximated with (estimates of) the large coefficients of $f$. 
Hence, the KM algorithm can be used directly to learn such functions. 

\begin{lemma} [cf.\ Lemma~3.1 in \cite{KM1993}]
\label{lm:sparse-approx-by-f}
Consider any distribution $D$ specified by a BN and its corresponding Fourier basis, and any Boolean function $f$
that can be approximated in square norm by a function $g$ such that $g$ is $T$-sparse. 
That is,  $g(\bfx) =\sum_{S\in {\cal T}} g_S \basis{S}(\bfx)$, $|{\cal T}|=T$ and 
$\EE_D[(f(\bfX)-g(\bfX))^2] \leq \epsilon$.

Let
${\cal S}=\{S \mbox{ s.t. } |\coeffgen{f}{S}|\geq \sqrt{\epsilon/T}\}$, where by Parseval's identity
$|{\cal S}|\leq T/\epsilon$.

Let $h_1(\bfx)=\sum_{S\in {\cal S}} \coeffgen{f}{S} \basis{S}(\bfx)$,
$h_2(\bfx)=\sum_{S\in {\cal S}} \approxcoeffgen{f}{S} \basis{S}(\bfx)$, 
and $h_3(\bfx)=\sum_{S\in {\cal S^*}} \approxcoeffgen{f}{S} \basis{S}(\bfx)$, 
where
${\cal S}\subseteq {\cal S^*}$, $|{\cal S^*}|\leq 4T/\epsilon$,
where $|\approxcoeffgen{f}{S}-\coeffgen{f}{S}|\leq \gamma$,
and $\gamma^2 \leq \epsilon^2/4T$.

Then $\EE_D[(f(\bfX)-h_1(\bfX))^2] \leq 2 \epsilon$, $\EE_D[(f(\bfX)-h_2(\bfX))^2] \leq 3 \epsilon$ and 
$\PP_D(f(\bfX)\not = \sign(h_3(\bfX))) \leq \EE_D[(f(\bfX)-h_3(\bfX))^2] \leq 3 \epsilon$.
\end{lemma}
\begin{proof}
First note that $\EE_D[(f(\bfX)-g(\bfX))^2]=\sum_{S\not \in {\cal T}} \hat{f}_S^2 + \sum_{S \in {\cal T}} (\hat{f}_S-\hat{g}_S)^2\leq\epsilon$ 
implying that $\sum_{S\not \in {\cal T}} \hat{f}_S^2\leq\epsilon$.
Let ${\cal T'}={\cal T}\cap {\cal S}$ and $r_1(\bfx)=\sum_{S\in {\cal T'}} \coeffgen{f}{S} \basis{S}(\bfx)$.
Then $\EE_D[(f(\bfX)-r_1(\bfX))^2] = \sum_{S\not \in {\cal T}} \coeff{S}^2+\sum_{S\in {\cal T}\setminus {\cal T'}} \coeff{S}^2 $.
The first term is bounded by $\epsilon$. For the second term, 
$\sum_{S\in {\cal T}\setminus {\cal T'}} \coeff{S}^2 \leq T (\epsilon/T)=\epsilon$.
Thus $\EE_D[(f(\bfX)-r_1(\bfX))^2]\leq 2\epsilon$.

Now 
$\EE_D[(f(\bfX)-h_1(\bfX))^2] = \sum_{S \not\in {\cal S}} \coeff{S}^2 \leq
\sum_{S\not \in {\cal T}} \coeff{S}^2+\sum_{S\in {\cal T}\setminus {\cal T'}} \coeff{S}^2 \leq 2\epsilon$.

Similarly, 
$\EE_D[(f(\bfX)-h_2(\bfX))^2] = \sum_{S \not\in {\cal S}} \coeff{S}^2 
+ \sum _{S \in {\cal S}} (\coeff{S}-\approxcoeff{S})^2
\leq 2\epsilon + \gamma^2 T/\epsilon\leq 3\epsilon$.

Similarly, 
$\EE_D[(f(\bfX)-h_3(\bfX))^2] = \sum_{S \not\in {\cal S^*}} \coeff{S}^2 
+ \sum _{S \in {\cal S^*}} (\coeff{S}-\approxcoeff{S})^2
\leq 2\epsilon + \gamma^2 4T/\epsilon\leq 3\epsilon$.
\end{proof}

The next lemma shows that the conditions of Lemma~\ref{lm:sparse-approx-by-f} can be achieved if our function can be approximated by another function with a low spectral norm. This has been used implicitly in several papers (see \cite{Mansour1995}).

\begin{lemma} %
\label{lm:sparse-approx-via-L1a} 
Consider any distribution $D$ specified by a BN and its corresponding Fourier basis, and any Boolean function $f$
that can be approximated in square norm by a function $h$ with low L1 spectrum, that is, 
$\EE_D[(f(\bfX)-h(\bfX))^2]\leq\epsilon/4$ and $\|\hat{h}\|_1\leq L_1$.
Then there exists a $g$ such that 
(1)
$\|\hat{g}\|_1\leq L_1$,
(2) $g$ is $T=4 L_1^2/\epsilon$ sparse,
and 
(3) 
$\EE_D[(f(\bfX)-g(\bfX))^2]\leq \epsilon$. 
\end{lemma}
\begin{proof}
Let ${\cal S}=\{S \mbox{ s.t. } |\coeffgen{h}{S}|\geq \epsilon/4L_1\}$ and let $g(\bfx)=\sum_{S\in {\cal S}} \coeffgen{h}{S} \basis{S}(\bfx)$.
We have $\|\hat{g}\|_1\leq\|\hat{h}\|_1\leq L_1$.
In addition, 
$L_1\geq \sum_S |\coeffgen{{h}}{S}|\geq \sum_{S\in {\cal S}} |\coeffgen{{h}}{S}|\geq |{\cal S}| \epsilon/4L_1$
and $|{\cal S}| \leq 4 L_1^2/\epsilon$.

For (3) note that
$\EE_D[(g(\bfX)-h(\bfX))^2]=\sum_{S\not \in {\cal S}}\coeffgen{h}{S}^2\leq 
\max_{S\not \in {\cal S}}|\coeffgen{h}{S}|
*
\sum_{S\not \in {\cal S}}|\coeffgen{h}{S}|
\leq (\epsilon/4 L_1)* L_1=\epsilon/4.
$
We can therefore use $\|A+B\|^2\leq 2(\|A\|^2+\|B\|^2)$, to get $\EE_D[(f(\bfX)-g(\bfX))^2]\leq 2(\EE_D[(f(\bfX)-h(\bfX))^2]+\EE_D[(g(\bfX)-h(\bfX))^2])\leq \epsilon$.
\end{proof}

Recall that DNFs are disjunctions of conjunctions and that in disjoint DNF the conjunctions are mutually exclusive. 
Decision trees whose node tests are individual variables are a subset of disjoint DNF, because they can be captured by the set of paths to leaves labeled 1. 
The next lemma shows that, when a bound $L_1(d)$ exists, disjoint DNF satisfies the conditions of Lemma~\ref{lm:sparse-approx-via-L1}.

\begin{lemma}
\label{lm:sparse-approx-of-disj-DNF} 
Consider any distribution $D$ with its corresponding Fourier basis where $L_1(d)$ is a bound on the spectral norm of conjunctions of size $d$.
For any $s$-term disjoint DNF $f$, 
there exists a function $h(\bfx):\{0,1\}^n \rightarrow\{0,1\}$ such that 
$\EE_D[(f(\bfX)-h(\bfX))^2]\leq\epsilon/4$ and $\L1(h)\leq L_1$
where
$d=\log_{1-c}(\epsilon/4s)$,
and
$L_1=s L_1(d)$.
\end{lemma}
\begin{proof}
Let $f=\vee_{i=1}^s t_i=\sum_{i} t_i$ where we can swap the disjunction by an addition due to the disjointness. 
Let $L=\{i| t_i \mbox{ is not longer than }d\}$ and let $h=\sum_{i\in L} t_i$. For all $S$, $\coeffgen{h}{S}=\sum_{i\in L}\coeffgen{t_i}{S}$ and $|\coeffgen{h}{S}|\leq \sum_{i\in L}|\coeffgen{t_i}{S}|$, 
and by the assumption of the theorem this implies $\|\hat{h}\|_1\leq L_1 = s L_1(d)$. Now using Lemma~\ref{lm:boundedprob} and the union bound we have
$\EE_D[(f(\bfX)-h(\bfX))^2]=\EE_D[\sum_{i\not \in L} t_i]\leq s (1-c)^{d}=\epsilon/4$.
\end{proof}

Combining the lemmas with Theorem~\ref{thm:EKM} we have:

\begin{customcorollary}{\ref{cor:KM DNF}}
[cf. Theorem 3.10 of \cite{KM1993}] 
\putcorKMdisjointDNF
\end{customcorollary}

\begin{proofnote}
The flow is $L_1 \rightarrow T \rightarrow \theta \rightarrow \gamma$. 
We have $T=4 L_1^2/\epsilon'$, 
$\theta=\sqrt{\epsilon'/T} = \epsilon'/2L_1$
and
$\gamma^2 \leq \epsilon'^2/4T={\epsilon'^3}/{16L_1^2}$.
\end{proofnote}

\subsection{Learning DNF Through Feldman's Algorithm}\label{app:feldman dnf}


We note that 
the original results \cite{Feldman2012} are developed with two restrictions. The first is a restriction to $c$-bounded product distributions.
The second is to restrict the scope of coefficients to be for sets of ``low degree", that is, zero out any higher order coefficients. 
However, these restrictions are only used through their implied properties of boundedness 
which was shown to hold in Lemma~\ref{lm:boundedprob} and a bound on the spectral norm of conjunctions which we assume generically in the form $L_1(d)$.
With these in place, the proofs go through 
without an explicit filtering of high degree coefficients. 
In addition, \citet{Feldman2012} considers the input domain to be $\{-1,1\}^n$ instead of $\{0,1\}^n$ which changes the value of $c$ for $c$-bounded product distributions. 
However, while this changes the constants slightly it does not affect the arguments. 
We therefore re-state the results in their general version without repeating the proofs.

Consider a functions $f:\cuben\rightarrow\{-1,1\}$ and any function $p:\cuben\rightarrow\reals$. We say that $p()$ 1-sign represents $f$ if $\forall \bfx, f(\bfx)=\sign(p(\bfx))$ and $|p(\bfx)|\geq 1$.

The key lemma allows us to use an approximation of the Fourier representation of $f$ (which in addition must have a bounded range) to approximate $f$.

\begin{lemma} [cf. Lemma 7 in \cite{Feldman2012}]
Consider any distribution $D$ specified by a BN and its corresponding Fourier basis, and any Boolean function $f$. 
Let $p()$ be a function that 1-sign represents $f$, $p'()$ any other function, and $g()$ a bounded function with range in $[-1,1]$.
Then: 
$$\PP_D[f(\bfX)\not = \sign(g(\bfX))] \leq \EE_D[|f(\bfX)-g(\bfX)|]\leq \|\hat{f}-\hat{g}\|_\infty \cdot \|\hat{p'}\|_1 + 2 \EE_D[|p'(\bfX)-p(\bfX)|].$$
\end{lemma}
\begin{proofnote}
Both the statement of the lemma and its
proof are identical to the original one except that we do not restrict the set of coefficients to those of low degree.
\end{proofnote}

Note that the only requirements for the predictor $g()$ is a bounded range, and an L$_\infty$ approximation of the Fourier coefficients of $f$. This allows us to avoid going through a square error approximation. 
The functions $p()$ and $p()'$ present structural requirements on the class of target functions $f$ but they are not directly required by the intended learning algorithm. 
Recall that the KM algorithm already provides a L$_\infty$ approximation. But the corresponding function may not be bounded. 
The algorithm and analysis of \citet{Feldman2012} show how this can be achieved.

The next lemma, uses a decomposition similar to Lemma~\ref{lm:sparse-approx-of-disj-DNF} to apply the previous result to DNF.

\begin{lemma} [cf. Theorem 11 in \cite{Feldman2012}]
\label{lm:dnfFeldman}
Consider any distribution $D$ with its corresponding Fourier basis where $L_1(d)$ is a bound on the spectral norm of conjunctions of size $d$.
Let $f$ be any $s$-term DNF expression, 
$d=\log_{1-c}(\epsilon/4s)$,
$L_1=2s L_1(d)+1$,
and $g()$ a bounded function with range in $[-1,1]$.
Then
$$\PP_D(f(\bfX)\not = \sign(g(\bfX))) \leq \EE_D[|f(\bfX)-g(\bfX)|]\leq L_1 \cdot \|\hat{f}-\hat{g}\|_\infty + \epsilon.$$
\end{lemma}
\begin{proofnote}
The proof is identical to the original one with a slight change in the values of the quantities. We have the same polynomial construction $p=2\sum t_i -1$ with the consequence that $\|\hat{p}\|_1\leq L_1$. Removing terms longer than $d$ yields $\EE_D[|p'(\bfX)-p(\bfX)|]\leq \epsilon/2$.
\end{proofnote}

Hence the algorithm {\em PTFconstruct} of \citet{Feldman2012}, with input $\gamma^*$, aims to approximate $\hat{f}$ with $\hat{g}$ while maintaining a bounded range for $g()$.
The algorithm runs in two phases.\footnote{The original presentation assumes the output of phase (1) as input and combines them for learning DNF. Here we combine them directly to simplify the presentation.
} 
(1) It first runs KM to get approximations $\approxcoeff{S}$ of all large coefficients of $f$ with $\theta=\gamma=\gamma^*$, 
(2) It then iteratively adds coefficients to the unbounded approximation $g'$ which is initialized to $g'=0$. 
The bounded $g$ is defined through a restriction of $g'$ to the range $[-1,1]$.
The algorithm iteratively finds $S$ such that $\approxcoeff{S}$ is far from $\approxcoeffgen{g}{S}$, and sets 
$g'=g'+c \ \basis{S}$ for a suitable value of $c$ ($c\in\pm \gamma^*$). Note that the updates are done on $g'$ that has an explicit Fourier representation, and that $g$ is defined through a restriction of the range of $g'$. Hence to find the violating coefficient $S$, the algorithm runs KM on the function $g$ 
with $\theta=\gamma=\gamma^*/2$
and compares the succinctly represented outputs of KM for $f$ and $g$. 
As noted by \citet{Feldman2012} the analysis of the algorithm does not depend on any property of the basis (other than orthonormality). 

\begin{lemma} [cf. Theorem 21 in \cite{Feldman2012}]
Consider any distribution $D$ specified by a BN and its corresponding Fourier basis, and any Boolean function $f$. 
Algorithm $\mbox{PTFconstruct}(D,f,\gamma,\delta)$ is given access to the BN representation of $D$ and a MQ oracle for $f$ and two accuracy parameters.
The algorithm runs in time polynomial in $n$,$1/\gamma$,$1/\delta$ and returns a list of sets ${\cal S}=\{S\}$
of bounded size $|{\cal S}| \leq 1/(2\gamma^2)$,
values for the corresponding coefficients $\approxcoeffgen{g}{S}$, and a hypothesis
$g(\bfx)=P_1[g'(\bfx)]$ where
$g'(\bfx) =\sum_{S\in {\cal S}} \approxcoeffgen{g}{S} \basis{S}(\bfx)$,
and $P_1(z)=\sign(z)\cdot \min(1,|z|)$.

With probability at least $1-\delta$, the function $g$ satisfies: $\|\hat{f}-\hat{g}\|_\infty \leq 5 \gamma$. 
\end{lemma}

Combining the previous two results and setting $\gamma$ appropriately, we get:

\begin{customcorollary}{\ref{cor:FC15}} [cf. Corollary 15 in \cite{Feldman2012}]
\putcorFeldmanDNF
\end{customcorollary}
\begin{proofnote}
The proof is identical to the original one with a slight change in the values of the quantities to adapt to the setting of Lemma~\ref{lm:dnfFeldman}. 
With this setting the error is bounded by $5\gamma L_1+\epsilon'=\epsilon$.
\end{proofnote}

\subsection{Agnostic Learning of Disjoint DNF Through \citet{Gopalan2008}'s Algorithm}\label{app:goplan dnf}

The algorithm of \citet{Gopalan2008} is based on the observation that, under the uniform distribution, decision trees with $t$ leaves have  spectral norm at most $t$.
Therefore, one can use the class of functions with bounded spectral norm for agnostic learning of decision trees.
By Lemma~\ref{lm:sparse-approx-of-disj-DNF} a $s$-term disjoint DNF function $f$ can be approximated by a function $g$ with error $\leq \epsilon/4$ using 
$d=\log_{1-c}(\epsilon/4s)$,
and
$t=L_1(g)=s L_1(d)$.
Therefore,
an agnostic learner for the class of functions with spectral norm bounded by $t$ with error guarantee of $\epsilon/2$ yields an agnostic guarantee for disjoint DNF with error  $< \epsilon$.



The algorithm of \citet{Gopalan2008} is based on a projected subgradient descent for the following problem:
\begin{equation}\label{eq:l1 reg}
    \min_{P \in K_t}  \err_f(P),
\qquad
\err_f(P) = \EE_D[|P(X) - f(X)|],
\end{equation}
where $K_t=\{P: L_1(P)\leq t\}$.
This is a convex optimization problem, but it range over $2^n$ variables and it is non-differentiable. 
To use gradient descent it is sufficient to compute a subgradient. For a polynomial $P$, define the pointwise sign function
\[
\nabla_f P(x) = \sign(P(x) - f(x)), \qquad \forall x \in \{0, 1\}^n.
\]
We note that, even in the non-uniform case, this is a valid subgradient of $\err_f$ at $P$.
\begin{lemma}[cf. Lemma 5 in \cite{Gopalan2008}]
    For any polynomial $P$, the function $\nabla_f P$ is a valid subgradient of $\err_f$ at $P$; that is for all polynomials $Q$,
    \[
\EE_D\qty[\nabla_f P(X) (P(X) - Q(X))] \ge \err_f(P) - \err_f(Q).
\]
\end{lemma}
\begin{proofnote}
    The proof follows the same steps as in the uniform case as in Lemma 5 of \cite{Gopalan2008} with the expectation taken over $D$ instead of the uniform distribution.
\end{proofnote}

Let $\mathrm{proj}_{K_t}$ denote the projection onto the set $K_t$ under the $\ell_2$ norm: 
\[
\mathrm{proj}_{K_t}(P) = \arg\min_{Q: L_1(Q)\leq t} \EE_D [\abs{P(X) - Q(X)}^2].
\]
Then the projected subgradient descent is given by the iteration:
\[
P_{k}' = P_{k-1} - \eta \nabla_f P_{k-1},
\qquad
P_k = \mathrm{proj}_{K_t}(P_k').
\]
Standard analysis implies convergence to within $O(\tfrac{t}{\sqrt{T}})$ of optimal error in $T$ iterations,
but this idealized algorithm is exponential-time due to the $2^n$-dimensional
representation.
To make this efficient, \citet{Gopalan2008} used the KM algorithm to approximate the subgradient $\nabla_f P$ at each step.  The key observation is that the Fourier coefficients of $\nabla_f P$ can be approximated using membership queries to $f$ and evaluation queries to $P$.  Since $P$ is $t$-sparse, evaluating $P(x)$ can be done efficiently.  Thus the KM algorithm can be used to find all large Fourier coefficients of $\nabla_f P$, which suffices to approximate the subgradient. 

In the non-uniform case, we can use the extended KM algorithm (Algorithm \ref{alg:KM}) to approximate the Fourier coefficients of $\nabla_f P$ under the distribution $D$ (see Algorithm \ref{alg:km-gradient-descent}).  
Given oracle access to a function $g$, and the Fourier basis corresponding to $D$, the extended-KM algorithm returns a sparse polynomial $Q$ satisfying
\[
\|\widehat{g} - \widehat{Q}\|_\infty \le \theta.
\]
which is an $\ell_\infty$ approximation of the Fourier
coefficients. Therefore, we can use the extended-KM to approximate $\nabla_f P_{k-1}$ at each iteration, implying that $\nabla_f P_{k-1}$ can be efficiently computed pointwise
using oracle access to $f$. 

Projection onto $K_t$ in the non-uniform case is the same as \citet{Gopalan2008}'s in the uniform case: For $\lambda \ge 0$, define
\[
{\mathrm{shrink}}(P,\lambda)(\CS) :=
\begin{cases}
\widehat{P}_\CS - \lambda & \widehat{P}_\CS \ge \lambda,\\
\widehat{P}_\CS + \lambda & \widehat{P}_\CS \le -\lambda,\\
0 & \text{otherwise}.
\end{cases}
\]
where $\widehat{P}_\CS$ are the Fourier coefficients of $P$ corresponding to distribution $D$.  The rest of the analysis follows through with minor modifications.

\begin{lemma}
    [cf. Lemma 6 in \cite{Gopalan2008}]
    The projection onto the set $K_t$ is given by
    $\mathrm{proj}_{K_t}(P) = \mathrm{shrink}(P,\lambda)$
for the smallest $\lambda$ such that
$L_1(\mathrm{shrink}(P,\lambda)) \le t$.
\end{lemma}

\begin{algorithm}[t]
\caption{Gradient Descent using Extended-KM}
\label{alg:km-gradient-descent}
\KwIn{Number of iterations $T$, step size $\eta \in (0,1)$, accuracy parameter $\theta \in (0,1)$}

Initialize $P_0 \gets 0$\;

\For{$k = 1$ \KwTo $T$}{
    $P_k' \gets P_{k-1} - \eta \cdot \mathrm{ExtendedKM}(\nabla_f P_{k-1}, \theta)$\;
    $P_k \gets \mathrm{ExtendedKM}(\mathrm{proj}_{K_t}(P_k'), \theta)$\;
}

\KwRet{$\arg\min_{k \in \{1,\dots,T\}} err_f(P_k)$}
\end{algorithm}

Another technical point is that $\mathrm{proj}_{K_t}$
is stable under $\ell_\infty$ perturbations of the Fourier coefficients (see Lemma 8 of \cite{Gopalan2008}). Hence, although the extended-KM only returns estimated  coefficients within additive error, the projection step ensures that the square error is also controlled.


Finally, \citet{Gopalan2008} shows that show that agnostic learning of a class $\mathcal{C}$, which is approximable via an $L_1$-bounded Fourier representation, is reducible to the sparse regression problem of (\ref{eq:l1 reg}).
To perform learning, the algorithm chooses a threshold $\alpha$ and outputs $\mathbf{1}[P(\gamma)\ge \alpha]$, where $\alpha$ minimizes empirical error on an additional sample. This argument holds for any distribution and deterministic  target function. The generalization to the agnostic setting where the target function is not necessarily deterministic also holds to our setting. The only difference is that Lemma~19 of \citet{Gopalan2008} is modified to the following form:

\begin{lemma}[cf. Lemma 19 in \cite{Gopalan2008}]
\label{lm:agnostic-learning}
Let $A^f$ be any algorithm that makes $q$ distinct queries to some fixed function $f : \{-1, 1\}^n \to \{-1,1\}$ and outputs  a function $h^f : \{-1,1\}^n \to \{-1,1\}$. Then
$$\PP_{(X,Y)\sim D}[h^f(X) \neq Y] \leq \EE_{f\sim F(D)} \qty[\PP[h^f(X) \neq f(X)| f]] + q (1-c)^n,$$
where $F(D)$ is the distribution over functions induced by $D$.
\end{lemma}
\begin{proofnote}
The proof is identical to the original one with a slight change in the values of the quantities. The  proof relies on the fact that if $f,f'\sim F(D)$ are two independent draws from $F(D)$, then  $\PP_{(X,Y)\sim D}[h^f(X)\neq f(X)]$ and $\PP_{(X,Y)\sim D}[h^{f}(X)\neq f'(X)]$  differ only on the $q$ queried inputs that contribute to a total mass of at most $q (1-c)^n$, as $D(x) \leq (1-c)^n$ for all $x$. 
\end{proofnote}

With this lemma, we can show that Theorem 20 of \citet{Gopalan2008} holds for any distribution $D$ with its corresponding Fourier basis where $L_1(d)$ is a bound on the spectral norm of conjunctions of size $d$:  
\begin{customcorollary}{\ref{cor:DT}}
\putagnosticDTcor
\end{customcorollary}
%

\newpage
\bibliography{ref_Learning}

\end{document}